\documentclass[11pt,a4paper]{article}
\usepackage{dsfont}
\newcommand\fakeslant[1]{%
  \pdfliteral{1 0 0.167 1 0 0 cm}#1\pdfliteral{1 0 -0.167 1 0 0 cm}}
\newcommand\mathbbmsl[1]{\mathds{\fakeslant{#1}}}
\usepackage{CJK}
\usepackage{hyperref}
\usepackage{doi}
\usepackage[charter]{mathdesign}
\usepackage{graphicx}
\usepackage{fullpage}
\usepackage{amsmath}
\usepackage{amsthm}
\usepackage[dvipsnames]{xcolor}
\hypersetup{
    colorlinks,
    linkcolor={red!50!black},
    citecolor={blue!50!black},
    urlcolor={blue!80!black}
}
\usepackage[inline]{enumitem}
\setlist{noitemsep}
\newtheorem{thm}{Theorem}[section]
\newtheorem{lem}[thm]{Lemma}
\newcommand{\HH}{\mathbbmsl{H}}
\newcommand{\HHdagger}{{\mathbbmsl{H}\hspace*{0.12em}}^\dagger}
\newcommand{\HHprime}{{\mathbbmsl{H}\hspace*{0.08em}}'}
\newcommand{\HHstar}{{\mathbbmsl{H}\hspace*{0.1em}}^*}
\newcommand{\BB}{\mathbbmsl{B}}
\newcommand{\TT}{\mathbbmsl{T}}
\newcommand{\GG}{\mathbbmsl{G}}
\newcommand{\CC}{\mathbbmsl{C}}
\newcommand{\FF}{\mathbbmsl{F}}
\newcommand{\KK}{\mathbbmsl{K}}
\newcommand{\PP}{\mathbbmsl{P}}
\newcommand{\QQ}{\mathbbmsl{Q}}
\newcommand{\RR}{\mathbbmsl{R}}
\newcommand{\UU}{\mathbbmsl{U}}
\newcommand{\VV}{\mathbbmsl{V}}
\newcommand{\WW}{\mathbbmsl{W}}
\newcommand{\LL}{\mathbbmsl{L}}

\renewcommand{\SS}{\mathbbmsl{S}}
\newcommand{\Xomit}[1]{}
\title{Three-in-a-Tree in Near Linear Time\thanks{An extended
    abstract to appear in {\em Proceedings of the 52nd Annual ACM
      Symposium on Theory of Computing}, 2020.}}
\author{Kai-Yuan Lai\thanks{Department of Computer Science and
    Information Engineering, National Taiwan University. A preliminary
    version of the paper appeared as the master's thesis of this
    author~\cite{Lai2018}. {baaldiablo3@gmail.com}.}
\and
Hsueh-I Lu\thanks{Corresponding author. Department of Computer Science
  and Information Engineering, National Taiwan University.  Address: 1
  Roosevelt Road, Section 4, Taipei 106, Taiwan, ROC. Research of this
  author is supported 
  by MOST grants
107--2221--E--002--032--MY3 and
104--2221--E--002--044--MY3.
{hil@csie.ntu.edu.tw}.
}
\and
  Mikkel Thorup\thanks{
  Department of Computer Science,
  University of Copenhagen, Denmark. 
  {mikkel2thorup@gmail.com}.
 Research of this author is supported by VILLUM Investigator Grant 16582, Basic Algorithms Research Copenhagen (BARC).  
} 
  }
\date{}

\begin{document}
\begin{CJK*}{UTF8}{bkai}
\maketitle
\begin{abstract}
The three-in-a-tree problem is to determine if a simple undirected
graph contains an induced subgraph which is a tree connecting three
given vertices. Based on a beautiful characterization that is proved
in more than twenty pages, Chudnovsky and Seymour~[{\em
    Combinatorica}~2010] gave the previously only known
polynomial-time algorithm, running in $O(mn^2)$ time, to solve the
three-in-a-tree problem on an $n$-vertex $m$-edge graph. Their
three-in-a-tree algorithm has become a critical subroutine in several
state-of-the-art graph recognition and detection algorithms.

In this paper we solve the three-in-a-tree problem in $\tilde{O}(m)$
time, leading to improved algorithms for recognizing perfect graphs
and detecting thetas, pyramids, beetles, and odd and even holes.  Our
result is based on a new and more constructive characterization than
that of Chudnovsky and Seymour. Our new characterization is stronger
than the original, and our proof implies a new simpler proof for the
original characterization. The improved characterization gains the
first factor $n$ in speed. The remaining improvement is based on
dynamic graph algorithms.
\end{abstract}


\section{Introduction}
\label{section:section1}

The graphs considered in this paper are all assumed to be undirected.
Also, it is convenient to think of them as connected.  Let $G$ be such
a graph with $n$ vertices and $m$ edges.  An {\em induced} subgraph of
$G$ is a subgraph $H$ that contains all edges from $G$ between
vertices in $H$.  For the {\em three-in-a-tree} problem, we are given
three specific terminals in $G$, and we want to decide if $G$ has an
induced tree $T$, that is, a tree $T$ which is an induced subgraph of
$G$, containing these terminals.  Chudnovsky and
Seymour~\cite{ChudnovskyS10} gave the formerly only known
polynomial-time algorithm, running in $O(mn^2)$ time, for the
three-in-a-tree problem. In this paper, we reduce the complexity of
three-in-a-tree from $O(mn^2)$ to $O(m\log^2 n)=\tilde O(m)$ time.

\begin{thm}
\label{theorem:theorem1.1}
It takes $O(m\log^2 n)$ time to solve the three-in-a-tree problem on
an $n$-vertex $m$-edge simple graph.
\end{thm}
To prove Theorem~\ref{theorem:theorem1.1}, we first improve the
running time to $O(m n)$ using a simpler algorithm with a simpler
correctness proof than that of Chudnovsky and Seymour. The remaining
improvement is done employing dynamic graph algorithms.

\subsection{Significance of three-in-a-tree}
The three-in-a-tree problem may seem like a toy problem, but it has
proven to be of general importance because many difficult graph
detection and recognition problems reduce to it. The reductions are
often highly non-trivial and one-to-many, solving three-in-a-tree on
multiple graph instances with different placements of the three
terminals. With our near-linear three-in-a-tree algorithm and some
improved reductions, we get the results summarized
Figure~\ref{figure:figure1}. These results will be explained in more
detail in Section~\ref{subsection:subsection1.1}.

\begin{figure}
\centerline{\scalebox{0.9}{
\renewcommand{\arraystretch}{1.2}
\begin{tabular}{|r||c|c|}
\hline
 &best previously known results
 &our work\\
\hline
\hline
three-in-a-tree
& $O(n^4)$~\cite{ChudnovskyS10} & 
$\tilde{O}(n^2)$: Theorem~\ref{theorem:theorem1.1} \\
\hline
theta
& $O(n^{11})$~\cite{ChudnovskyS10}&   $\tilde{O}(n^6)$: Theorem~\ref{theorem:theorem1.2}\\
\hline
pyramid
& $O(n^9)$~\cite{ChudnovskyCLSV05} & $\tilde{O}(n^5)$: Theorem~\ref{theorem:theorem1.3}\\
\hline
perfect graph
& $O(n^9)$~\cite{ChudnovskyCLSV05} & $O(n^8)$: Theorem~\ref{theorem:theorem1.4}\\
\hline
odd hole
& $O(n^9)$~\cite{ChudnovskySSS19} & $O(n^8)$: Theorem~\ref{theorem:theorem1.4}\\
\hline
beetle
& $O(n^{11})$~\cite{ChangL15} &  $\tilde{O}(n^6)$: Theorem~\ref{theorem:theorem1.5}\\
\hline
even hole
& $O(n^{11})$~\cite{ChangL15}  & $O(n^9)$: Theorem~\ref{theorem:theorem1.6}\\
\hline
\end{tabular}
\renewcommand{\arraystretch}{1.2}
}
}
\caption{Comparing our work with the best previously known results for
  an $n$-vertex graph.}
\label{figure:figure1}
\end{figure}

Showcasing some of the connections, our improved three-in-a-tree
algorithm leads to an improved algorithm to detect if a graph has an
\emph{odd hole}, that is, an induced cycle of odd length above
three. This is via the recent odd-hole algorithm of Chudnovsky, Scott,
Seymour, and Spirkl~\cite{ChudnovskySSS19}. A highly nontrivial
consequence of odd-hole algorithm is that we can use it to recognize
if a graph $G$ is \emph{perfect}, that is, if the chromatic number of
each induced subgraph $H$ of $G$ equals the clique number of $H$.  The
celebrated Strong Perfect Graph Theorem states that a graph is perfect
if and only if neither the graph nor its complement has an odd hole.
An odd-hole algorithm can therefore trivially test if a graph is
perfect. The Strong Perfect Graph Theorem, implying the last reduction
was a big challenge to mathematics, conjectured by Berge in
1960~\cite{Berge60,Berge61,Berge85} and proved by Chudnovsky,
Robertson, Seymour, and Thomas~\cite{ChudnovskyRST06}, earned them the
2009 Fulkerson prize. Our improved three-in-a-tree algorithm improves
the time to recognize if a graph is perfect from $O(n^9)$ to
$O(n^8)$. While this is a modest polynomial improvement, the point is
that three-in-a-tree is a central sub-problem on the path to solve
many other problems.

The next obvious question is why three-in-a-tree? Couldn't we have
found a more general subproblem to reduce to? The dream would be to
get something like disjoint paths and graph minor theory where we
detect a constant sized minor or detect if we have disjoint paths
connecting of a constant number of terminal pairs (one path connecting
each pair) in $O(n^2)$ time. This is using the algorithm of
Kawarabayashi, Kobayashi, and Reed~\cite{KawarabayashiKR12}, improving
the original cubic algorithm of Robertson and
Seymour~\cite{RobertsonS95b}.

In light of the above grand achievements, it may seem unambitious for
Chudnovsky and Seymour to work on three-in-a-tree as a general
tool. The difference is that the above disjoint paths and minors are
not necessarily induced subgraphs. Working with induced paths, many of
the most basic problems become NP-hard. Obviously, we can decide if
there is an induced path between two terminals, but Bienstock
\cite{Bienstock91} has proven that it is NP-hard to decide
two-in-a-cycle, that is, if two terminals are in an induced
cycle. From this we easily get that it is NP-hard to decide
three-in-a-path, that is if there is an induced path containing three
given terminals.  Both of these problems would be trivial if we could
solve the induced disjoint path problem for just two terminal
pairs. In connection with the even and odd holes and perfect graphs,
Bienstock also proved that it is NP-hard to decide if there is an even
(respectively, odd) hole containing a given terminal.

In light of these NP-hardness results it appears quite lucky that
three-in-a-tree is tractable, and of sufficient generality that it can
be used as a base for solving other graph detection and recognition
problems nestled between NP-hard problems. In fact, three-in-a-tree
has become such a dominant tool in graph detection that authors
sometimes explained when they think it cannot be
used~\cite{ChudnovskyST13,TrotignonV10}, e.g., Trotignon and
Vu\v{s}kovi\'{c}~\cite{TrotignonV10} wrote ``A very powerful tool for
solving detection problems is the algorithm three-in-a-tree of
Chudnovsky and Seymour [...] But as far as we can see, three-in-a-tree
cannot be used to solve $\Pi_{H_{1|1}}$.''

While proving that a problem is in P is the first big step in
understanding the complexity, there has also been substantial prior
work on improving the polynomial complexity for many of the problems
considered in this paper. In the next subsection, we will explain in
more detail how our near-linear three-in-a-tree algorithm together
with some new reductions improve the complexity of different graph
detection and recognition problems. In doing so we also hope to
inspire more new applications of three-in-a-tree in efficient graph
algorithms.

\subsection{Implications}
\label{subsection:subsection1.1}

\begin{figure}[t]
\centering
\centerline{\scalebox{0.38}{{\includegraphics{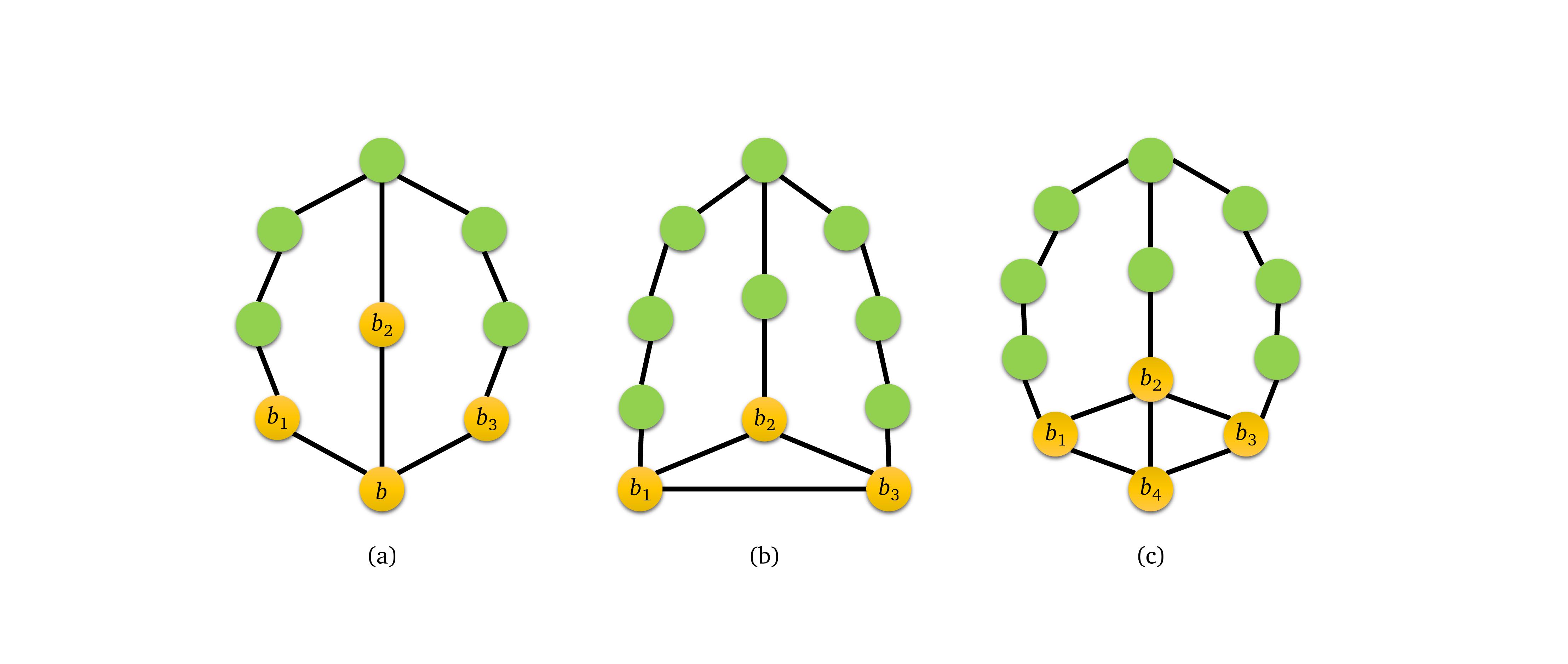}}}}
\caption{(a) Theta. (b) Pyramid. (c) Beetle.}
\label{figure:figure2}
\end{figure}
We are now going to describe the use of our three-in-a-tree algorithm
to improve the complexity of several graph detection and recognition
problems. The reader less familiar with structural graph theory may
find it interesting to see how the route to solve the big problems
takes us through several toy-like subproblems, starting from
three-in-a-tree.  Often we look for some simple configuration implying
an easy answer.  If the simple configuration is not present, then this
tells us something about the structure of the graph that we can try to
exploit.

We first define the big problems in context. A {\em hole} is an
induced simple cycle with four or more vertices. A graph is
\emph{chordal} if and only if it has no hole. Rose, Tarjan, and
Leuker~\cite{RoseTL76} gave a linear-time algorithm for recognizing
chordal graphs. A hole is {\em odd} (respectively, {\em even}) if it
consists of an odd (respectively, even) number of vertices.  $G$ is
{\em Berge} if $G$ and its complement are both odd-hole-free. The
celebrated Strong Perfect Graph Theorem, which was conjectured by
Berge~\cite{Berge60,Berge61, Berge85} and proved by Chudnovsky,
Robertson, Seymour, and Thomas~\cite{ChudnovskyRST06}, states that $G$
is Berge if and only if $G$ is perfect, i.e., the chromatic number of
each induced subgraph $H$ of $G$ equals the clique number of $H$.

The big problems considered here are the detection of odd and even
holes, but related to this we are going to look for ``thetas'',
``pyramids'', and ``beetles'', as illustrated in
Figure~\ref{figure:figure2}. These are different induced subdivisions
where a \emph{subdivision} of a graph is one where edges are replaced
by paths of arbitrary length.  A hole is thus an induced subdivision
of a length-4 cycle, and a minimal three-in-a-tree is an induced
subdivision of a star with two or three leaves that are all
prespecified terminals.

The first problem Chudnovsky and Seymour~\cite{ChudnovskyS10} solved
using their three-in-tree algorithm was to detect a {\em theta} which
is any induced subdivision of
$K_{2,3}$~\cite{Bang-JensenHT12}. Chudnovsky and Seymour are
interested in thetas because they trivially imply an even hole.  They
developed the previously only known polynomial-time algorithm, running
in $O(n^{11})$ time, for detecting thetas in $G$ via solving the
three-in-a-tree problem on $O(n^7)$ subgraphs of $G$.  Thus,
Theorem~\ref{theorem:theorem1.1} reduces the time to
$\tilde{O}(n^9)$. Moreover, we show in Lemma~\ref{lemma:lemma6.1} that
thetas in $G$ can be detected via solving the three-in-a-tree problem
on $O(mn^2)$ $n$-vertex graphs, leading to an $\tilde{O}(n^6)$-time
algorithm as stated in Theorem~\ref{theorem:theorem1.2}.

\begin{thm}
\label{theorem:theorem1.2}
It takes $O(mn^4\log^2n)$ time to detect thetas in an $n$-vertex
$m$-edge graph.
\end{thm}
The next problem Chudnovsky and Seymour solved using their
three-in-tree algorithm was to detect a {\em pyramid} which is an
induced subgraph consisting of an apex vertex $u$ and a triangle
$v_1v_2v_3$ and three paths $P_1$, $P_2$, and $P_3$ such that $P_i$
connects $u$ to $v_i$ and touch $P_j$, $j\neq i$, only in $u$, and
such that at most one of $P_1$, $P_2$, and $P_3$ has only one
edge. The point in a pyramid is that it must contain an odd hole. An
$O(n^9)$-time algorithm for detecting pyramids was already contained
in the perfect graph algorithm of Chudnovsky et
al.~\cite[\S2]{ChudnovskyCLSV05}, but Chudnovsky and Seymour use their
three-in-a-tree to give a more natural ``less miraculous'' algorithm
for pyramid detection, but with a slower running time of $O(n^{10})$.
With our faster three-in-a-tree algorithm, their more natural pyramid
detection also becomes the faster algorithm with a running time of
$\tilde{O}(n^8)$. Moreover, as for thetas, we improve the reductions
to three-in-a-tree. We show (see Lemma~\ref{lemma:lemma6.2}) that
pyramids in $G$ can be detected via solving the three-in-a-tree
problem on $O(mn)$ $n$-vertex graphs, leading to an
$\tilde{O}(mn^3)$-time algorithm as stated in
Theorem~\ref{theorem:theorem1.3}.

\begin{thm}
\label{theorem:theorem1.3}
It takes $O(mn^3\log^2n)$ time to detect pyramids in an $n$-vertex
$m$-edge graph.
\end{thm}

We now turn to odd holes and perfect graphs. Since a graph is perfect
if and only if it and its complement are both odd-hole-free, an
odd-hole algorithm implies a perfect graph algorithm, but not vice
versa. Cornu\'{e}jols, Liu, and Vu\v{s}kovi\'{c}~\cite{CornuejolsLV03}
gave a decomposition-based algorithm for recognizing perfect graphs
that runs in $O(n^{18})$ time, which was reduced to $O(n^{15})$ time
by Charbit, Habib, Trotignon, and
Vu\v{s}kovi\'{c}~\cite{CharbitHTV12}.  The best previously known
algorithm, due to Chudnovsky, Cornu\'{e}jols, Liu, Seymour, and
Vu\v{s}kovi\'{c}~\cite{ChudnovskyCLSV05}, runs in $O(n^9)$ time.
However, the tractability of detecting odd holes was open for
decades~\cite{ChudnovskyS18,ConfortiCKV99,ConfortiCLVZ06,Hsu87} until
recently.  Chudnovsky, Scott, Seymour, and
Spirkl~\cite{ChudnovskySSS19} announced an $O(n^9)$-time algorithm for
detecting odd holes, which also implies a simpler $O(n^9)$-time
algorithm for recognizing perfect graphs. An $O(n^9)$-time bottleneck
of both of these perfect-graph recognition algorithms was the above
mentioned algorithm for detecting
pyramids~\cite[\S2]{ChudnovskyCLSV05}.

By Theorem~\ref{theorem:theorem1.3}, the pyramids can now be detected
in $\tilde{O}(mn^3)$-time, but Chudnovsky et al.'s odd-hole algorithm
has six more $O(n^9)$-time subroutines~\cite[\S4]{ChudnovskySSS19}.
By improving all these bottle-neck subroutines, we improve the
detection time for odd holes to $O(m^2n^4)$, hence the recognition
time for perfect graphs to $O(n^8)$.
\begin{thm}
\label{theorem:theorem1.4}
\begin{enumerate*}[label=(\arabic*), ref=\arabic*]
\item  
\label{item1:theorem1.4}
It takes $O(m^2n^4)$ time to detect odd holes in an $n$-vertex
$m$-edge graph, and hence
\item
\label{item2:theorem1.4}
it takes $O(n^8)$ time to recognize an $n$-vertex perfect graph.
\end{enumerate*}
\end{thm}

Even-hole-free graphs have been extensively
studied~\cite{Addario-BerryCHRS08,ConfortiCKV00,
  ConfortiCKV02a,daSilvaV07,daSilvaV13,FraserHH18,KloksMV09,SilvaSS10}.
Vu\v{s}kovi\'{c} \cite{Vuskovic10} gave a comprehensive survey.
Conforti, Cornu{\'e}jols, Kapoor, and
Vu\v{s}kovi\'{c}~\cite{ConfortiCKV97con,ConfortiCKV02b} gave the first
polynomial-time algorithm for detecting even holes, running in
$O(n^{40})$ time.  Chudnovsky, Kawarabayashi, and
Seymour~\cite{ChudnovskyKS05} reduced the time to $O(n^{31})$.  A {\em
  prism}\label{b2} consists of two vertex-disjoint triangles together
with three vertex-disjoint paths between the two triangles such that
the union of every two of the three paths induces a cycle.  Chudnovsky
et al.~\cite{ChudnovskyKS05} also observed that the time of detecting
even holes can be further reduced to $O(n^{15})$ as long as detecting
prisms is not too expensive, but this turned out to be
NP-hard~\cite{MaffrayT05}.  However, Chudnovsky and
Kapadia~\cite{ChudnovskyK08} and Maffray and
Trotignon~\cite[Algorithm~2]{MaffrayT05} devised $O(n^{35})$-time and
$O(n^5)$-time algorithms for detecting prisms in theta-free and
pyramid-free graphs $G$, respectively.  Later, da~Silva and
Vu\v{s}kovi\'{c}~\cite{daSilvaV13} improved the time of detecting even
holes in $G$ to $O(n^{19})$.  The best formerly known algorithm, due
to Chang and Lu~\cite{ChangL15}, runs in $O(n^{11})$ time.  One of its
two $O(n^{11})$-time bottlenecks \cite[Lemma~2.3]{ChangL15} detects
so-called beetles in $G$ via solving the three-in-a-tree problem on
$O(n^7)$ subgraphs of $G$.  Theorem~\ref{theorem:theorem1.1} reduces
the time to $\tilde{O}(n^9)$.  Moreover, we show in
Lemma~\ref{lemma:lemma6.3} that beetles can be detected via solving
the three-in-a-tree problem on $O(m^2)$ $n$-vertex graphs, leading to
an $\tilde{O}(n^6)$-time algorithm as stated in
Theorem~\ref{theorem:theorem1.5}.

\begin{thm}
\label{theorem:theorem1.5}
It takes $O(m^2 n^2\log^2 n)$ time to detect beetles in an $n$-vertex
$m$-edge graph.
\end{thm}

Combining our faster beetle-detection algorithm with our $O(n^9)$-time
algorithm in~\S\ref{subsection:subsection6.3}, which is carefully
improved from the other $O(n^{11})$-time bottleneck
subroutine~\cite[Lemma~2.4]{ChangL15}, we reduce the time of detecting
even holes to $O(n^9)$ as stated in Theorem~\ref{theorem:theorem1.6}.
\begin{thm}
\label{theorem:theorem1.6}
It takes $O(m^2n^5)$ time to detect even holes in an $n$-vertex
$m$-edge graph.
\end{thm}

For other implications of Theorem~\ref{theorem:theorem1.1},
L{\'{e}}v{\^{e}}que, Lin, Maffray, and Trotignon gave $O(n^{13})$-time
and $O(n^{14})$-time algorithms for certain properties $\Pi_{B_4}$ and
$\Pi_{B_6}$, respectively~\cite[Theorems~3.1 and~3.2]{LevequeLMT09}.
By Theorem~\ref{theorem:theorem1.1} and the technique
of~\S\ref{subsubsection:subsubsection6.2.1}, the time can be reduced
by a $\Theta({n^5}/{\log^2 n})$ factor.
Theorem~\ref{theorem:theorem1.1} also improves the algorithms of
van~'t~Hof, Kaminski, and Paulusma~\cite[Lemmas~4 and~5]{HofKP12}.  We
hope and expect that three-in-a-tree with its new near-optimal
efficiency will find many other applications in efficient graph
algorithms.

\subsection{Other related work}
For the general {\em $k$-in-a-tree} problem, we are given $k$ specific
terminals in $G$, and we want to decide if $G$ has an induced tree
$T$. The $k$-in-a-tree problem is NP-complete~\cite{DerhyP09} when $k$
is not fixed. With our Theorem~\ref{theorem:theorem1.1}, it can be
solved in near-linear time for $k\leq 3$, and the tractability is
unknown for any fixed $k \geq 4$~\cite{GolovachPL12}. Solving it in
polynomial time for constant $k$ would be a huge result. It is,
however, not clear that $k$-in-a-tree for $k>3$ would be as powerful a
tool in solving other problems as three-in-a-tree has proven to be.

While $k$-in-a-tree with bounded $k$ is unsolved for general graphs,
there has been substantial work devoted to $k$-in-a-tree for special
graph classes. Derhy, Picouleau, and Trotignon~\cite{DerhyPT09} and
Liu and Trotignon~\cite{LiuT10} studied $k$-in-a-tree on graphs with
girth at least~$k$ for $k=4$ and general $k\geq 4$, respectively.
Dos~Santos, da~Silva, and Szwarcfiter~\cite{dosSantosMurilodSS15}
studied the $k$-in-a-tree problem on chordal graphs.  Golovach,
Paulusma, and van~Leeuwen~\cite{GolovachPL12} studied the
$k$-in-a-tree, $k$-in-a-cycle, and $k$-in-a-path problems on AT-free
graphs~\cite{LekkerkerkerB62}.  Bruhn and Saito~\cite{BruhnS12},
Fiala, Kaminski, Lidick{\'{y}}, and Paulusma~\cite{FialaKLP12}, and
Golovach, Paulusma, and van~Leeuwen~\cite{GolovachPL15} studied the
$k$-in-a-tree and $k$-in-a-path problems on claw-free graphs.

See~\cite{AboulkerRTV12, Bang-JensenHM15, BoncompagniRV19,
  CameronCH18, ChangKL15, ChudnovskyL17, ChudnovskyLMTV19,
  ChudnovskyMSS18a, ChudnovskyMSS18, ChudnovskySSS20, ConfortiCLVZ06,
  DiotRTV18, DiotTT14, ErdosSS86, FominTV15, GitlerRV17, Hoang18,
  RadovanovicTV19, ScottS16} for more work on graph detection,
recognition, and characterization.  Also
see~\cite[Appendix~A]{BrandstadtLS99} for a survey of the recognition
complexity of more than $160$ graph classes.

On the hardness side, recall that three-in-a-tree can also be viewed
as three in a subdivided star with two or three terminal leaves.
However, detecting such a star with 4 terminal leaves is NP-hard.
(This follows from Bienstock's NP-hardness of 2-in-a-cycle
\cite{Bienstock91}, asking if there exists a hole containing two
vertices $u$ and $v$, which may be assumed to be nonadjacent: Add two
new leaves $u_1$ and $u_2$ adjacent to $u$ and then, for every two
neighbors $v_1$ and $v_2$ of $v$, check if the new graph contains an
induced subdivision of a star with exactly four terminal leaves
$u_1,u_2,v_1,v_2$.)  Even without terminals, it is NP-hard to detect
induced subdivisions of any graph with minimum degree at least
four~\cite{Bang-JensenHM15,LevequeLMT09}. Finally, we note that if we
allow multigraphs with parallel edges, then even 2-in-a-path becomes
NP-hard. This NP-hardness is an easy exercise since the induced path
cannot contain both end-points of parallel edges.

We note that it is the subdivisions that make induced graph detection
hard for constant sized pattern graphs. Without subdivisions, we can
trivially check for any induced $k$-vertex graph in $O(n^k)$
time. Nesetril and Poljak has improved this to roughly
$O(n^{k\omega/3})$ where $\omega$ is the exponent of matrix
multiplication \cite{NesetrilP85}.  On the other hand, the ETH
hypothesis implies that we cannot detect if a $k$-clique is a(n
induced) subgraph in $n^{o(k)}$ time~\cite{ImpagliazzoP01}.  A more
general understanding of the hardness of detecting induced graphs has
been presented recently in~\cite{DalirrooyfardVW19}.

\subsection{Techniques}
\label{subsection:subsection1.2}

Chudnovsky and Seymour's $O(n^2m)$-time algorithm for the
three-in-a-tree problem is based upon their beautiful characterization
for when a graph with three given terminals are contained in some
induced tree~\cite{ChudnovskyS10}. The aim is to either find a
three-in-a-tree or a witness that it cannot exist. During the course
of the algorithm, they develop the witness to cover more and more of
the graph. In each iteration, they take some part that is not covered
by the current witness and try to add it in, but then some other part
of the witness may pop out.  They then need a potential function
argument to show progress in each iteration.

What we do is to introduce some extra structure to the witness when no
three-in-a-tree is found, so that when things are added, nothing pops
out. This leads to a simpler more constructive algorithm that is
faster by a factor $n$. Our new witness has more properties than that
of Chudnovsky and Seymour, so our characterization of no
three-in-a-tree is strictly stronger, yet our overall proof is
shorter. Essentially the point is that by strengthening the inductive
hypothesis, we get a simpler inductive step. The remaining improvement
in speed is based on dynamic graph algorithms.

\subsection{Road map}
\label{subsection:subsection1.3}

The rest of the paper is organized as
follows. Section~\ref{section:section2} is a background section where
we review Chudnovsky and Seymour's characterization for
three-in-a-tree, sketch how it is used algorithmically, as well as the
bottleneck for a fast implementation. Section~\ref{section:section3}
presents our new stronger characterization as well as a high level
description of the algorithms and proofs leading to our $\tilde{O}(m)$
implementation. Section~\ref{section:section4} proves the correctness
of our new characterization. Section~\ref{section:section5} provides
an efficient implementation. Finally, Section~\ref{section:section6}
shows how our improved three-in-a-tree algorithm, in tandem with other
new ideas, can be used to improve many state-of-the-art graph
recognition and detection algorithms.  Section~\ref{section:section7}
concludes the paper.

\newpage
\section{Background}
\label{section:section2}

\subsection{Preliminaries}
\label{subsection:subsection2.1}

Let $|S|$ denote the cardinality of set $S$.  Let $R\setminus S$ for
sets $R$ and $S$ consist of the elements of $R$ not in $S$.  Let $G$
and $H$ be graphs.  Let $V(G)$ (respectively, $E(G)$) consist of the
vertices (respectively, edges) of $G$.  Let $u$ and $v$ be vertices.
Let $U$ and $V$ be vertex sets.  Let $N_G(u)$ consist of the neighbors
of $u$ in $G$.  The {\em degree} of $u$ in $G$ is $ |N_G(u)|$.  Let
$N_G[u]=N_G(u)\cup \{u\}$.  Let $N_G(U)$ be the union of
$N_G(u)\setminus U$ over all vertices $u\in U$.  Let
$N_G(u,H)=N_G(u)\cap V(H)$ and $N_G(U,H)=N_G(U)\cap V(H)$.  The
subscript $G$ in notation $N_G$ may be omitted.  A {\em leaf} of $G$
is a degree-one vertex of $G$.  Let $\nabla(G)$ denote the graph
obtained from $G$ by adding an edge between each pair of leaves of
$G$.  Let $G[H]$ denote the subgraph of $G$ induced by $V(H)$.  Let
$G-U=G[V(G)\setminus U]$.  Let $G-u=G-\{u\}$.  Let $uv$ denote an edge
with end-vertices $u$ and $v$.  Graphs $H_1$ and $H_2$ are {\em
  disjoint} if $V(H_1)\cap V(H_2)=\varnothing$.  Graphs $H_1$ and
$H_2$ are {\em adjacent} in $G$ if $H_1$ and $H_2$ are disjoint and
there is an edge $uv$ of $G$ with $u\in V(H_1)$ and $v\in V(H_2)$.  A
{\em $UV$-path} is either a vertex in $U\cap V$ or a path having one
end-vertex in $U$ and the other end-vertex in $V$.  A {\em
  $UV$-rung}~\cite{ChudnovskyS10} is a vertex-minimal induced
$UV$-path.  If $U=\{u\}$, then a $UV$-path is also called a {\em
  $uV$-path} and a {\em $Vu$-path}.  If $U=\{u\}$ and $V=\{v\}$, then
a $UV$-path is also called a {\em $uv$-path}. Let {\em $Uv$-rung},
{\em $uV$-rung}, and {\em $uv$-rung} be defined similarly.

For the three-in-a-tree problem, we assume without loss of generality
that the three given terminals of the input $n$-vertex $m$-edge simple
undirected graph $G$ are exactly the leaves of $G$. A {\em sapling} of
$G$ is an induced tree containing all three leaves of $G$, so the
three-in-a-tree problem is the problem of finding a sapling.

\subsection{Chudnovsky and Seymour's characterization}
\label{subsection:subsection2.2}

\begin{figure}[t]
\centerline{\scalebox{0.5}{{\includegraphics{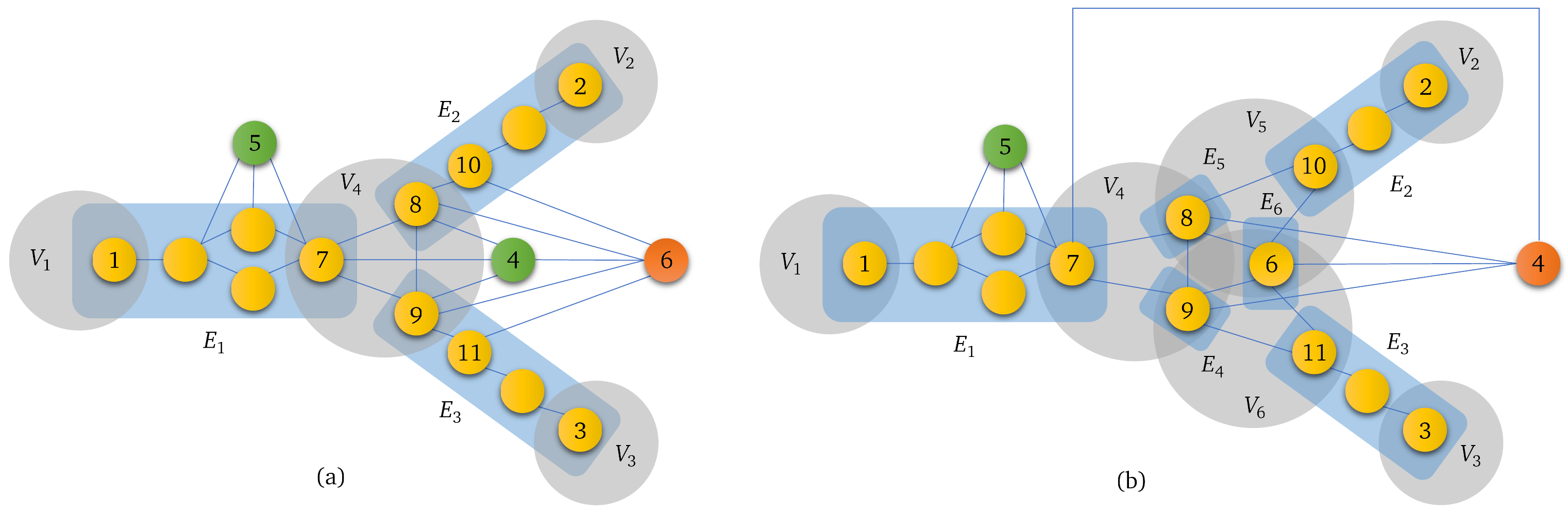}}}}
\caption{(a) An $X$-net $\HH$ with nodes $V_1,\ldots,V_4$ and arcs
  $E_1,E_2,E_3$, where $X$ consists of the vertices other than
  $4,5,6$. Vertices $4$ and $5$ are $\HH$-local.  Vertex $6$ is
  $\HH$-nonlocal.  (b) A nonlocal net $\HH$ having a triad
  $\Delta(V_4,V_5,V_6)=\{6,8,9\}$. Vertex $5$ is $\HH$-local.  Vertex
  $4$ is $\HH$-nonlocal.}
\label{figure:figure3}
\end{figure}

Let $\HH$ be a graph such that each member of $V(\HH)$ and $E(\HH)$,
called {\em node} and {\em arc} respectively, is a subset of
$X\subseteq V(G)$.  $\HH$ is an {\em $X$-net} of $G$ if the following
Conditions~\ref{condition:N} hold (see
Figure~\ref{figure:figure3}(a)):
\begin{enumerate}[label={}, ref={N}, leftmargin=0pt]
\item 
\label{condition:N}
\begin{enumerate}[label={\em \ref{condition:N}\arabic*:},ref={\ref{condition:N}\arabic*},leftmargin=*]
    \item
    \label{N1}
    Graph $\HH$ is connected and graph $\nabla(\HH)$ is biconnected.

    \item
    \label{N2}
    The arcs of $\HH$ form a nonempty disjoint partition of the vertex
    set $X$.

    \item
    \label{N3} 
    Graph $\HH$ has exactly three leaf nodes, each of which
    consists of a leaf vertex of $G$.

    \item
    \label{N4}
    For any arc $E=UV$ of $\HH$, each vertex of $X$ in $E$ is on a
    $UV$-rung of $G[E]$.
    
    \item
    \label{N5}
    For any arc $E$ and node $V$ of $\HH$, $E\cap V\ne\varnothing$
    if and only if $V$ is an end-node of $E$ in $\HH$.

    \item
    \label{N6}
    For any vertices $u$ and $v$ in $X$ contained by distinct arcs $E$
    and $F$ of $\HH$, $uv$ is an edge of $G$ if and only if arcs $E$
    and $F$ share a common end-node $V$ in $\HH$ with
    $\{u,v\}\subseteq V$.
\end{enumerate}
\end{enumerate}
An arc $E=UV$ is {\em simple} if $G[E]$ is a $UV$-rung.  A {\em net}
is an $X$-net for an $X$. A {\em base net} is a net obtained via the
next lemma, for which we include a proof to make our paper
self-contained.
\begin{lem}[Chudnovsky and Seymour~\cite{ChudnovskyS10}]
\label{lemma:lemma2.1}
It takes $O(m)$ time to find a sapling of $G$ or a net of $G$ whose
arcs are all simple.
\end{lem}

\begin{proof}
Let $s_1,s_2,s_3$ be the leaves of $G$.  Obtain vertex sets $R$ and
$S$ such that $G[S]$ is an $s_2s_3$-rung of $G$ and $G[R]$ is an
$s_1S$-rung of $G$.  Let $x_1\in R\setminus S$ be closest to $S$ in
$G[R]$.  Let each $x_j\in N(x_1,S)$ with $j\in\{2,3\}$ be closest to
$s_j$ in $G[S]$.  Since $s_2$ and $s_3$ are leaves of $G$, $x_2$ and
$x_3$ are internal vertices of path $G[S]$.  If $x_2=x_3$, then
$G[R\cup S]$ is a sapling of $G$.  If $x_2$ and $x_3$ are distinct and
nonadjacent, then $G[R\cup S]-I$ is a sapling of $G$, where $I$
consists of the internal vertices of the $x_2x_3$-path in $G[S]$.  If
$x_2$ and $x_3$ are adjacent in $G$, then $G$ admits an $R\cup S$-net
having nodes $V_0=\{x_1,x_2,x_3\}$ and $V_i=\{s_i\}$ with
$i\in\{1,2,3\}$ and simple arcs $E_i=V_0V_i$ with $i\in\{1,2,3\}$
consisting of the vertices of the $s_ix_i$-rung of $G[R\cup S]$.
\end{proof}

The original definition of Chudnovsky et al.~only used nets with no
parallel arcs, but for our own more efficient construction, we need to
use parallel arcs.  A {\em triad} of $\HH$ is
$\Delta(V_1,V_2,V_3)=(V_1\cap V_2)\cup (V_2\cap V_3)\cup (V_3\cap
V_1)$ for nodes $V_1$, $V_2$, and $V_3$ that induce a triangle in
graph $\HH$.  A subset $S$ of $X$ is {\em $\HH$-local} if $S$ is
contained by a node, arc, or triad of $\HH$~\cite{ChudnovskyS10}.  A
set $Y\subseteq V(G-X)$ is {\em $\HH$-local} if $N(Y,X)$ is
$\HH$-local.  $\HH$ is {\em local} if every $Y\subseteq V(G-X)$ with
connected $G[Y]$ is $\HH$-local.  See Figure~\ref{figure:figure3}.
The following theorem is Chudnovsky and Seymour's characterization.

\begin{thm}[{Chudnovsky and Seymour~\cite[3.2]{ChudnovskyS10}}]
\label{theorem:theorem2.2}
$G$ is sapling-free if and only if $G$ admits a local net with no
parallel arcs.
\end{thm}

The proof of Theorem~\ref{theorem:theorem2.2} in~\cite{ChudnovskyS10}
takes up more than 20 pages. We will here present a stronger
characterization with a shorter proof, which moreover leads to a much
faster implementation.  Our results throughout the paper do not rely
on Theorem~\ref{theorem:theorem2.2}. Moreover, our paper delivers an
alternative self-contained proof for Theorem~\ref{theorem:theorem2.2}.

Chudnovsky and Seymour's proof of Theorem~\ref{theorem:theorem2.2} is
algorithmic maintaining an $X$-net $\HH$ with $X\subseteq V(G)$ having
no parallel arcs until a sapling of $G$ is found or $\HH$ becomes
local, implying that $G$ is sapling-free by the if direction of
Theorem~\ref{theorem:theorem2.2}.  In each iteration, if $\HH$ is not
local, they find a minimal set $Y\subseteq V(G-X)$ with connected
$G[Y]$ such that $Y$ is $\HH$-nonlocal.  Their proof for the only-if
direction of Theorem~\ref{theorem:theorem2.2} shows that if $G[X\cup
  Y]$ is sapling-free, then $\HH$ can be updated to an $X'$-net with
$Y\subseteq X'\subseteq X\cup Y$.  Although $Y$ joins the resulting
$X'$-net $\HH$, a subset of $X$ may have to be moved out of $\HH$ to
preserve Conditions~\ref{condition:N} for $\HH$.  To bound the number
of iterations, Chudnovsky and Seymour showed that the potential $|X| +
(n + 1) \cdot |V(\HH)|$ of $\HH$ stays $O(n^2)$ and is increased by
each iteration, implying that the total number of iterations is
$O(n^2)$. In the next section, we will present a new stronger
characterization that using parallel arcs with particular properties
avoids the aforementioned in-and-out situation. More precisely, our
$X$ will grow in each iteration, reducing the number of iterations to
at most $n$.

\section{Our stronger characterization}
\label{section:section3}

A base net of $G$ contains only simple arcs.  However, we do need
other more complex arcs, but we will show that it suffices that all
non-simple arcs are ``flexible'' in the sense defined below.  For
vertex sets $S$, $V_1$, and $V_2$, an {\em $(S,V_1,V_2)$-sprout} is an
induced subgraph of $G$ in one of the following {\em
  Types~\ref{type:S}}:
\begin{enumerate}[label={},ref={S},leftmargin=0pt]
\item 
\label{type:S}
\begin{enumerate}[label={\em \ref{type:S}\arabic*:},ref={S\arabic*},leftmargin=*]
\item 
\label{S1} 
A tree intersecting each of $S$, $V_1$, and $V_2$ at exactly one
vertex.

\item 
\label{S2} 
An $SV_1$-rung  
not intersecting $V_2$
plus a disjoint $SV_2$-rung
not intersecting $V_1$.

\item 
\label{S3}
A $V_1V_2$-rung 
not intersecting $S$
plus a disjoint $SV$-rung with $V=V_1\cup V_2$.
\end{enumerate}
\end{enumerate}

\begin{figure}[t]
\centerline{\scalebox{0.47}{\includegraphics{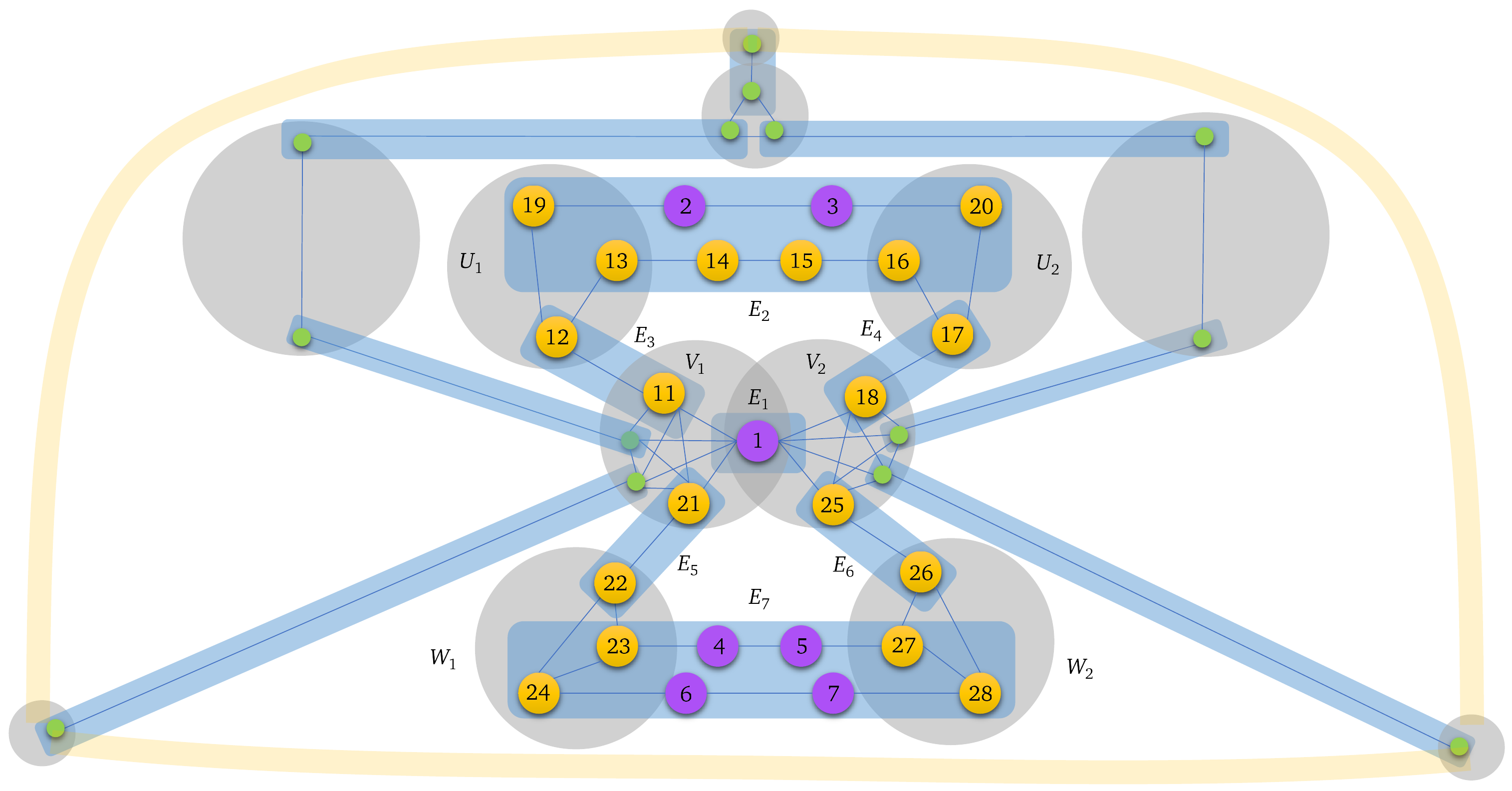}}}
\caption{A web $\HH$. The arcs of $\nabla(\HH)$ between the three
  leaves of $\HH$ are in yellow.}
\label{figure:figure4}
\end{figure}

Let $S=\{1,\ldots,7\}$ for the example in Figure~\ref{figure:figure4}.
Vertex $1$ is an $(S,V_1,V_2)$-sprout of Type~\ref{S1}.  The set
$\{2,19,12,11,13,14,15,16\}$ induces an $(S,V_1,U_2)$-sprout of
Type~\ref{S1}.  The only $(S,U_1,U_2)$-sprout and $(S,W_1,W_2)$-sprout
of Type~\ref{S1} contain vertex $1$.  The set $\{23,4,7,28\}$ induces
an $(S,W_1,W_2)$-sprout of Type~\ref{S2}.  The set
$\{19,2,13,14,15,16\}$ induces an $(S,U_1,U_2)$-sprout of
Type~\ref{S3}.  An arc $E=UV$ of $\HH$ is {\em flexible} if $G[E]$
contains an $(S,U,V)$-sprout for each nonempty vertex set $S\subseteq
E$.  For the example in Figure~\ref{figure:figure4}, arcs
$E_1,E_3,E_4,E_5,E_6$ are simple and arcs $E_1,E_2,E_7$ are flexible.
An $X$-net $\HH$ is an {\em $X$-web} if all arcs of $\HH$ are simple
or flexible.  A {\em web} is an $X$-web for some $X$.  A base net of
$G$ is a web of $G$.  Let $\HH$ be a net.  A {\em split component}
$\GG$ for $\HH$ is either an arc $UV$ of $\HH$ or a subgraph of $\HH$
containing a cutset $\{U,V\}$ of $\nabla(\HH)$ such that $\GG$ is a
maximal subgraph of $\nabla(\HH)$ in which $U$ and $V$ are nonadjacent
and do not form a cutset~\cite{SPQR}.  For both cases, we call
$\{U,V\}$ the {\em split pair} of $\GG$ for $\HH$.  The split
components having split pair $\{V_1,V_2\}$ in
Figure~\ref{figure:figure4} are (1) the $V_1V_2$-path with an arc
$E_1$, (2) the $V_1V_2$-path with arcs $E_3,E_2,E_4$, and (3) the
$V_1V_2$-path with arcs $E_5,E_7,E_6$.  Thus, even if $\HH$ has no
parallel arcs, there can be more than one split components sharing a
common split pair.  One can verify that each split component $\GG$ of
$\HH$ contains at most one leaf node of $\HH$ and, if $\GG$ contains a
leaf node $V$ of $\HH$, then $V$ belongs to the split pair of $\GG$.
A vertex subset $C$ of $G$ is a {\em chunk} of $\HH$ if $C$ is the
union of the arcs of one or more split components for $\HH$ that share
a common split pair $\{U,V\}$ for $\HH$. In this case, we call
$\{U,V\}$ the {\em split pair} of $C$ for $\HH$ and call $C$ a {\em
  $UV$-chunk} of $\HH$.  A chunk of $\HH$ is {\em maximal} if it is
not properly contained by any chunk of $\HH$.  A node of $\HH$ is a
{\em maximal split node} if it belongs to the split pair of a maximal
chunk for $\HH$.  For the net $\HH$ of $G$ in
Figure~\ref{figure:figure4}, $E_1$, $E_3$, $E_3\cup E_2$, $E_3\cup
E_2\cup E_4$, and $E_1\cup E_3\cup E_2\cup E_4$ are all chunks of
$\HH$.  If we consider only the subsets of $V(G)$ that intersect the
numbered vertices, then $E_1\cup \cdots\cup E_7$ is the only maximal
chunk and $V_1$ and $V_2$ are the only maximal split nodes.  Given an
$X$-net $\HH$, a subset $S$ of $X$ is {\em $\HH$-tamed} if every pair
of vertices from $S$ is either in the same arc or together in some
node of $\HH$.  A set $Y\subseteq V(G-X)$ is {\em $\HH$-tamed} if
$N(Y,X)$ is $\HH$-tamed.  $\HH$ is {\em taming} if every $Y\subseteq
V(G-X)$ with connected $G[Y]$ is $\HH$-tamed.  If $S\subseteq X$ is
$\HH$-local, then $S$ is $\HH$-tamed. The converse does not hold: If
$\HH$ has simple arcs $E$ and $F$ between nodes $U$ and $V$, $G[E]$ is
an edge $uv$ with $u\in U$ and $v\in V$, and $G[F]$ is a vertex $w\in
U\cap V$, then $\{u,v,w\}$ is $\HH$-tamed and $\HH$-nonlocal.
However, if $\HH$ has no parallel arcs, then each $\HH$-tamed subset
of $X$ is $\HH$-local, as shown in
Lemma~\ref{lemma:lemma3.5}\eqref{item2:lemma3.5}.

A {\em non-trivial} $V_1V_2$-chunk $C$ of $\HH$ is one that is not an
arc in $\HH$.  We then define the operation $\textsc{merge}(C)$ which
for a $V_1V_2$-chunk $C$ of $\HH$ replaces all arcs of $\HH$
intersecting $C$ by an arc $E = V_1V_2$ with $E = C$ and deletes the
nodes whose incident arcs are all deleted. We shall prove that this
$\textsc{merge}$ operation preserves that $\HH$ is a net (see
Lemma~\ref{lemma:lemma3.4}).  Let $\HHdagger$ denote the $X$-net
obtained from $\HH$ by applying $\textsc{merge}(C)$ on $\HH$ for each
maximal chunk $C$ of $\HH$. We call $\HHdagger$ the $X$-net that
{\em aids} $\HH$. Such an aiding net has no non-trivial chunks and no
parallel arcs.  See Figure~\ref{figure:figure5} for examples.  The
simple graph $\nabla(\HHdagger$) is triconnected.  $V$ is node of
$\HHdagger$ if and only if $V$ is a maximal split node of $\HH$.
$E$ is an arc of $\HHdagger$ if and only if $E$ is a maximal chunk
of $\HH$ (respectively, $\HHdagger$).  The next theorem is our
characterization, which is the basis for our much more efficient
near-linear time algorithm.

\begin{figure}[t]
\centerline{\scalebox{0.47}{\includegraphics{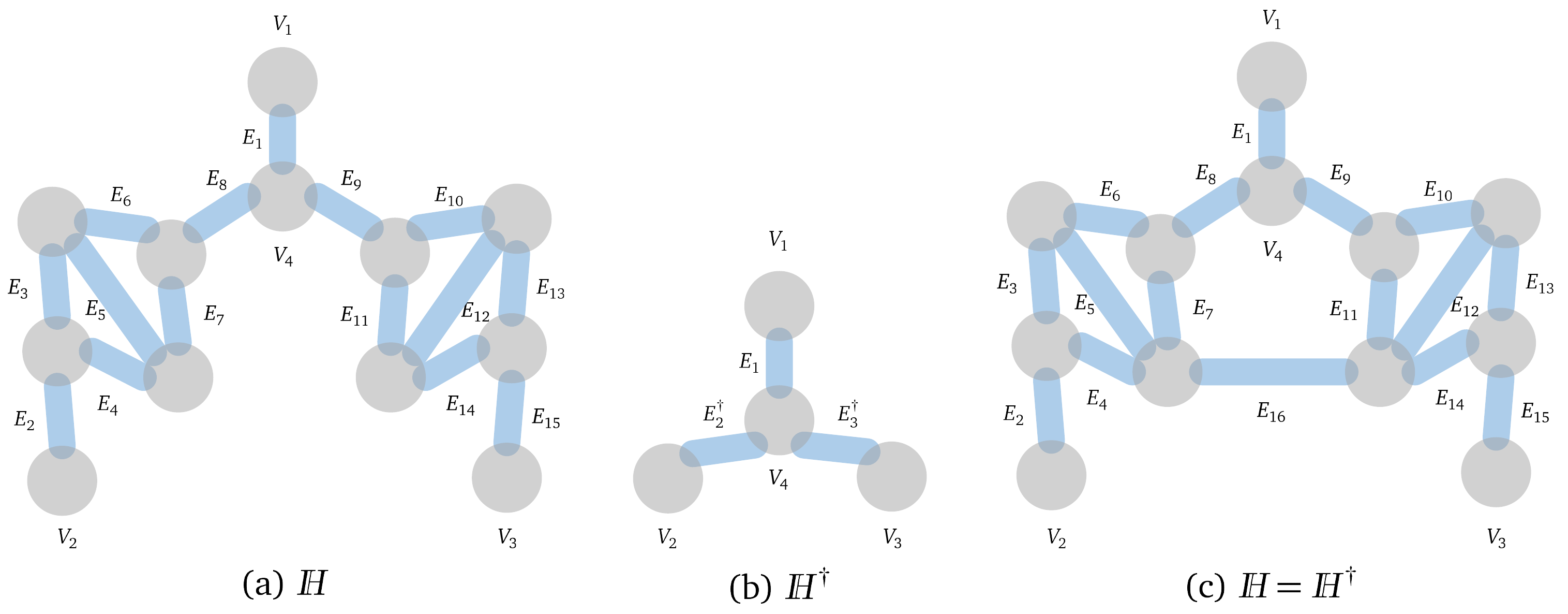}}}
\caption{The aiding net of the $\HH$ in (a) is the $\HHdagger$ in
  (b) with $E^\dagger_2=E_2\cup \cdots\cup E_8$ and
  $E^\dagger_3=E_9\cup\cdots\cup E_{15}$. The net $\HH$ in (c) aids
  itself.}
\label{figure:figure5}
\end{figure}

\begin{thm}
\label{theorem:theorem3.1}
$G$ is sapling-free if and only if $G$ admits a web $\HH$ with a
taming aiding net $\HHdagger$.
\end{thm}

Theorem~\ref{theorem:theorem3.1} is stronger than Chudnovsky and
Seymour's Theorem~\ref{theorem:theorem2.2} in that our proof of
Theorem~\ref{theorem:theorem3.1} provides as a new shorter proof of
Theorem~\ref{theorem:theorem2.2}. To quantify the difference, the
proof of Theorem~\ref{theorem:theorem2.2} in~\cite{ChudnovskyS10}
takes up more than 20 pages while our proof of our stronger
Theorem~\ref{theorem:theorem3.1} is self-contained and takes up~13
pages
(pages~\pageref{section:section2}--\pageref{endofcharacterization})
including the review of their structure, many more figures, and a
simpler $O(mn)$-time algorithm.  For the relation between the two
structural theorems, we will prove in
Lemma~\ref{lemma:lemma3.5}\eqref{item2:lemma3.5} that every taming net
of $G$ having no parallel arcs is local. Since the aiding net
$\HHdagger$ in Theorem~\ref{theorem:theorem3.1} has no parallel
arcs, $\HHdagger$ is local as required by
Theorem~\ref{theorem:theorem2.2}.  The algorithmic advantage of
Theorem~\ref{theorem:theorem3.1} is that we know that $\HHdagger$ is
the aiding net of a web $\HH$ which has more structure than an
arbitrary net.

To get a self-contained proof of the easy if-direction of
Theorem~\ref{theorem:theorem3.1}, we prove more generally that if $G$
admits a taming net, then $G$ is sapling-free
(Lemma~\ref{lemma:lemma3.5}\eqref{item1:lemma3.5}). This proof holds
for any net including nets with parallel arcs like our web $\HH$.
Proving the only-if direction is the hard part for both structural
theorems. Our new proof follows the same general pattern as the old
one stated after the statement of Theorem~\ref{theorem:theorem2.2},
but with crucial differences to be detailed later.

We grow an $X$-web $\HH$ with $X\subseteq V(G)$ until a sapling of $G$
is found or $\HHdagger$ becomes taming, implying that $G$ is
sapling-free by the if direction of Theorem~\ref{theorem:theorem3.1}.
In each iteration, if $\HHdagger$ is not taming, we find a minimal
set $Y\subseteq V(G-X)$ with connected $G[Y]$ such that $Y$ is not
$\HHdagger$-tamed.  To prove the only-if direction of
Theorem~\ref{theorem:theorem3.1}, we show that if $G[X\cup Y]$ is
sapling-free, then $\HH$ can be expanded to an $X'$-web with $X'=X\cup
Y$.

Comparing with the proof of Chudnovsky and Seymour that we sketched
below Theorem~\ref{theorem:theorem2.2}, we note that in their case,
their new $X'$-net would be for some $Y\subseteq X'\subseteq X\cup Y$,
whereas we get $X'=X\cup Y$. This is why we can guarantee termination
in $O(n)$ rounds while they need a more complicated potential function
to demonstrate enough progress in $O(n^2)$ rounds.

Another important difference is that we operate both on a web $\HH$
and its aiding net $\HHdagger$. Recall that the web $\HH$ is a net
allowing parallel arcs, but with the special structure that all arcs
are simple or flexible. This special structure is crucial to our
simpler inductive step where we can always add $Y$ as above to get a
new web over $X'=X\cup Y$.  If we just used $\HH$, then we would have
too many untamed sets. This is where we use the aiding net
$\HHdagger$ which generally has fewer untamed sets. It is only for
the minimally $\HHdagger$-untamed sets $Y\subseteq V(G-X)$ that we
can guarantee progress as above. Thus we need the interplay between
the well-structured fine grained web $\HH$ and its more coarse grained
aiding net $\HHdagger$ to get our shorter more constructive proof of
Theorem~\ref{theorem:theorem3.1}.  On its own, our more constructive
characterization buys us a factor $n$ in speed.  This has to be
combined with efficient data structures to get down to near-linear
time.

\begin{figure}[t]
\centerline{\scalebox{0.47}{\includegraphics{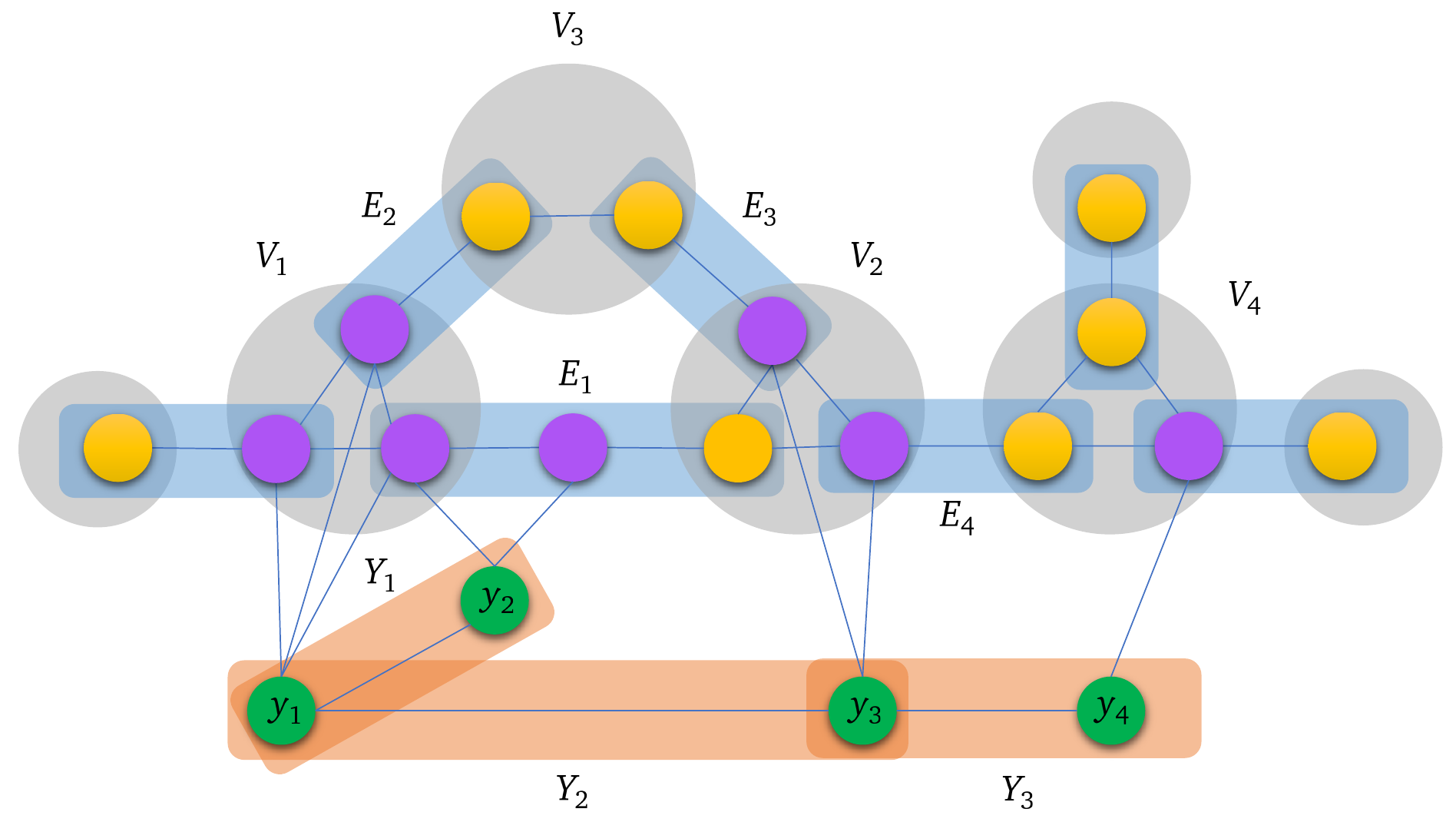}}}
\caption{An $X$-web $\HH$, where $X$ consists of the vertices other
  than $y_1,y_2,y_3,y_4$.  Vertices $y_1,\ldots,y_4$ are all
  $\HH$-tamed and $\HHdagger$-tamed.  
  $Y_1$ and $Y_2$ are $\HH$-wild 
and $\HHdagger$-nonwild.
  $Y_3$ is $\HH$-wild and $\HHdagger$-wild.  
$Y_1$ is $\HH$-solid.  $Y_2$ and
  $Y_3$ are $\HH$-nonsolid.  $E_1$, $E_1\cup E_2\cup E_3$, and
  $E_1\cup E_2\cup E_3\cup E_4$ are pods of $Y_1$ and $Y_2$ in $\HH$.
  $Y_3$ is $\HH$-unpodded.  $Y_1$ and $Y_2$ are $\HH$-sticky and $Y_3$
  is $\HH$-nonsticky.}
\label{figure:figure6}
\end{figure}

\subsection{Two major lemmas and our algorithm for detecting saplings}
\label{subsection:subsection3.1}

Let $\HH$ be an $X$-net.  An {\em $\HH$-wild} set is a minimally
$\HH$-untamed $Y\subseteq V(G-X)$ such that $G[Y]$ is a path.  In
Figure~\ref{figure:figure6}, $Y_1\cup Y_2$ is $\HH$-untamed but not
$\HH$-wild, since $Y_1\subsetneq Y_1 \cup Y_2$ is $\HH$-untamed.
$\HH$ is not taming if and only if $G$ admits an $\HH$-wild set.  An
$S\subseteq X$ is {\em $\HH$-solid} if $S$ is a node of $\HH$ or $S$
is a subset of an arc $E=UV$ of $\HH$ such that $G[E]$ contains no
$(S,U,V)$-sprout.  If $S$ is a subset of a simple arc of $\HH$, then
$S$ is $\HH$-solid if and only if $G[S]$ is an edge, since a sprout
has to be an induced subgraph of $G$.  Let $Y\subseteq V(G-X)$ such
that $G[Y]$ is a path.  $Y$ is {\em $\HH$-solid} if (1) $N(Y,X)$ is
the union of two $\HH$-solid sets and (2) $N(y,X)=\varnothing$ for
each internal vertex $y$, if any, of path $G[Y]$.  A {\em pod} of $Y$
in $\HH$ is a $V_1V_2$-chunk $C$ of $\HH$ with the following {\em
  Conditions~\ref{condition:P}}:
\begin{enumerate}[label={}, ref={P}, leftmargin=0pt]
\item 
\label{condition:P}
\begin{enumerate}[label={\em \ref{condition:P}\arabic*:},ref={\ref{condition:P}\arabic*},leftmargin=*]
\item 
\label{P1}
$N(Y,X)\subseteq V_1\cup C\cup V_2$.

\item 
\label{P2}
For each $i\in\{1,2\}$, $N(y,V_i)\subseteq C$ or $V_i\subseteq C\cup
N(y)$ holds for an end-vertex $y$ of path $G[Y]$. \end{enumerate}
\end{enumerate}
$Y$ is {\em $\HH$-podded} if $Y$ admits a pod in $\HH$.  $Y$ is {\em
  $\HH$-sticky} if $Y$ is $\HH$-solid or $\HH$-podded.  See
Figure~\ref{figure:figure6}.
 
\begin{lem}
\label{lemma:lemma3.2}
Let $Y$ be an $\HHdagger$-wild set for an $X$-web $\HH$.
\begin{enumerate*}[label=(\arabic*), ref=\arabic*]
\item 
\label{item1:lemma3.2}
If  $Y$ is 
$\HH$-nonsticky, then $G[X\cup  Y]$ contains a sapling.
\item 
\label{item2:lemma3.2}
If $Y$ is $\HH$-sticky, then 
$\HH$ can be expanded to an $X \cup Y$-web.
\end{enumerate*}
\end{lem}

By Lemmas~\ref{lemma:lemma2.1} and \ref{lemma:lemma3.2} and
Theorem~\ref{theorem:theorem3.1}, the following algorithm detects
saplings in $G$:

\fbox{\em Algorithm~\ref{algorithm:A}}\\[-18pt]
\begin{enumerate}[label={}, ref={A}, leftmargin=0pt]
\item 
\label{algorithm:A}
\begin{enumerate}[label={\em Step~\ref{algorithm:A}\arabic*:},ref={\ref{algorithm:A}\arabic*},leftmargin=*,
itemindent=1.5cm]
\item 
\label{A1}
If a sapling of $G$ is found (Lemma~\ref{lemma:lemma2.1}), then exit
the algorithm.
\item 
\label{A2}
Let $X$-web $\HH$ be the obtained base net of $G$ and then repeat the
following steps: 
\begin{enumerate}[leftmargin=*, label={(\alph*)},
    ref={\ref{A2}(\alph*)}]
\item 
\label{A2a}
If $\HHdagger$ is taming, then report that $G$ is sapling-free
(if-direction of Theorem~\ref{theorem:theorem3.1}) and exit.

\item 
\label{A2b}
If $\HHdagger$ is not taming, then obtain an $\HHdagger$-wild set
$Y$.

\item 
\label{A2c}
If $Y$ is $\HH$-nonsticky, then report that $G[X \cup Y]$ contains a
sapling (Lemma~\ref{lemma:lemma3.2}\eqref{item1:lemma3.2}) and exit.

\item 
\label{A2d}
If $Y$ is $\HH$-sticky, then expand $\HH$ to an $X\cup Y$-web
(Lemma~\ref{lemma:lemma3.2}\eqref{item2:lemma3.2}).
\end{enumerate}
\end{enumerate}
\end{enumerate}

\begin{lem}
\label{lemma:lemma3.3}
Algorithm~\ref{algorithm:A} can be implemented to run in $O(m\log^2
n)$ time.
\end{lem}

\subsection[]{Reducing
Theorems~\ref{theorem:theorem1.1},~\ref{theorem:theorem2.2},
and~\ref{theorem:theorem3.1} to Lemmas~\ref{lemma:lemma3.2}
and~\ref{lemma:lemma3.3} via aiding net}
\label{subsection:subsection3.2}

We need a relationship between simple paths in $\HH$ and induced paths
in $G$. For any simple $UV$-path $\PP$ of $\HH$ (i.e., $U$ and $V$ are
the end-nodes of $\PP$ in $\HH$), we define a {\em $\PP$-rung} of $G$
as a $UV$-rung of $G$ where all edges are contained in the arcs of
$\PP$.  Such a $\PP$-rung always exists by Conditions~\ref{N4}
and~\ref{N6} of $\HH$ as long as $U\ne V$.  For the degenerate case
$U=V$, let {\em $\PP$-rung} be defined as the empty vertex set.  For
any distinct nodes $U_1$ and $U_2$ of $\HH$ intersecting a
$V_1V_2$-chunk $C$ of $\HH$, there are disjoint $\UU\VV$-rungs $\PP_1$
and $\PP_2$ of $\HH$ with $\UU=\{U_1,U_2\}$ and $\VV=\{V_1,V_2\}$ by
Condition~\ref{N1} of $\HH$.  Since $\PP_1$ and $\PP_2$ are disjoint,
any $\PP_1$-rung and $\PP_2$-rung of $G$ are disjoint and nonadjacent
by Conditions~\ref{N2} and~\ref{N6} of $\HH$.  Consider the
$V_1V_2$-chunk $C=E_1\cup\cdots\cup E_7$ in
Figure~\ref{figure:figure4}. Let $\VV=\{V_1,V_2\}$.  Let $\PP_1$ be
the path of $\HH$ with arc $E_3$.  Let $\PP_2$ be the path of $\HH$
with arc $E_4$.  Let $\PP_3$ be the path of $\HH$ with arcs $E_6$ and
$E_7$.  Let $\PP_4$ be the degenerate path of $\HH$ consisting of a
single node $V_1$.  If $\UU=\{U_1,U_2\}$, then $\PP_1$ and $\PP_2$ are
disjoint $\UU\VV$-rungs.  If $\UU=\{U_1,W_1\}$, then $\PP_1$ and
$\PP_3$ are disjoint $\UU\VV$-rungs of $\HH$.  If $\UU=\{V_1,W_1\}$,
then $\PP_3$ and $\PP_4$ are disjoint $\UU\VV$-rungs of $\HH$.  The
path of $G$ induced by vertex set $\{11,12\}$ is the unique
$\PP_1$-rung of $G$.  The path of $G$ induced by vertex set
$\{17,18\}$ is the unique $\PP_2$-rung of $G$.  The paths induced by
vertex sets $\{25,26,27,5,4,23\}$ and $\{25,26,28,7,6,24\}$ are the
two $\PP_3$-rungs of $G$.  The empty vertex set is the unique
$\PP_4$-rung of $G$.

\begin{lem}
\label{lemma:lemma3.4}
If $C$ is a $V_1V_2$-chunk of an $X$-net $\HH$, then applying
$\textsc{merge}(C)$ on $\HH$ results in an $X$-net.
\end{lem}

\begin{proof}
Let $\HHprime$ be the resulting $\HH$.  Since any node cutset of $\HHprime$ is
also a node cutset of $\HH$, Conditions~\ref{N1} of $\HHprime$ holds.
Conditions~\ref{N2} and~\ref{N3} of $\HHprime$ hold trivially.
Conditions~\ref{N5} and~\ref{N6} of $\HHprime$ follow from those of $\HH$.
To see Condition~\ref{N4} of $\HHprime$, let $x$ be a vertex in $C$.  Let
$E=U_1U_2$ be the arc of $\HH$ containing $x$.  There are disjoint
$\UU\VV$-rungs $\PP_1$ and $\PP_2$ of $\HH$ with $\UU=\{U_1,U_2\}$ and
$\VV=\{V_1,V_2\}$.  Let each $P_i$ with $i\in\{1,2\}$ be a
$\PP_i$-rung of $G[C]$.  Let $Q$ be a $U_1U_2$-rung of $G[E]$
containing $x$.  $G[P_1\cup Q\cup P_2]$ is a $V_1V_2$-rung of $G[C]$
containing $x$.
\end{proof}

\begin{lem}
\label{lemma:lemma3.5}
\begin{enumerate*}[label=(\arabic*), ref=\arabic*]
\item 
\label{item1:lemma3.5}
If $G$ admits a taming net, then $G$ is sapling-free.

\item 
\label{item2:lemma3.5}
If an $X$-net $\HH$ has no parallel arcs, then every $\HH$-tamed
subset of $X$ is $\HH$-local.
\end{enumerate*}
\end{lem}

Since any $\HH$-local subset of $X$ for any $X$-net $\HH$ is
$\HH$-tamed, Lemma~\ref{lemma:lemma3.5}\eqref{item1:lemma3.5} implies
the if direction of Chudnovsky et al.'s
Theorem~\ref{theorem:theorem2.2}. Moreover, by
Lemma~\ref{lemma:lemma3.5}\eqref{item2:lemma3.5}, the only-if
direction of Theorem~\ref{theorem:theorem3.1} implies the only-if
direction of Theorem~\ref{theorem:theorem2.2}.  Thus, our proofs for
Lemma~\ref{lemma:lemma3.5} and the only-if direction of
Theorem~\ref{theorem:theorem3.1} form a self-contained proof for
Theorem~\ref{theorem:theorem2.2}.

\begin{proof}[{Proof of Lemma~\ref{lemma:lemma3.5}}]
Statement~\ref{item1:lemma3.5}: Assume a taming net $\HH$ and a
sapling $T$ of $G$ for contradiction.  By Condition~\ref{N6} of $\HH$,
any two adjacent vertices in $T$ contained by distinct arcs of $\HH$
belong to a node.  If $G[Y]$ is a connected component of $T-X$, then
vertices $u$ and $v$ of $T$ in $N(Y,X)$ belong to an arc of $\HH$: If
$u$ and $v$ were in distinct arcs, then $\{u,v\}$ would be contained
by a node of $\HH$, since $\HH$ is taming.  By Condition~\ref{N6} of
$\HH$, $uv$ is an edge of $G$, contradicting that $T$ is an induced
tree.  By Conditions~\ref{N2},~\ref{N3}, and~\ref{N5} of $\HH$, the
nodes and arcs of $\HH$ intersecting $T$ form a three-leaf connected
subgraph $\TT$ of $\HH$. Thus, $T$ intersects a node of $\TT$ and
three of its incident arcs in $\TT$.  Condition~\ref{N6} implies a
triangle in $T$, contradiction.

Statement~\ref{item2:lemma3.5}: Assume an $\HH$-tamed $\HH$-nonlocal
$S\subseteq X$ for contradiction.  Let $E_1,\ldots,E_\ell$ with
$\ell\geq 2$ be the arcs of $\HH$ intersecting $S$.  Since $S$ is
$\HH$-tamed, any $E_i$ and $E_j$ with $1\leq i<j\leq\ell$ share a
common end-node. If there is a common end-node $V$ of
$E_1,\ldots,E_\ell$, then the other end-nodes of $E_i$ with
$i\in\{1,\ldots,\ell\}$ are pairwise distinct, since $\HH$ has no
parallel arcs.  Since $S$ is $\HH$-tamed, Condition~\ref{N6} implies
$S\subseteq V$, contradicting that $S$ is $\HH$-nonlocal.  Thus,
$\ell=3$ and $E_1,E_2,E_3$ form a triangle of $\HH$. Since $S$ is
$\HH$-nonlocal, there a vertex $x_i\in E_i\setminus V_i$ with
$i\in\{1,2,3\}$ for an end-node $V_i$ of $E_i$.  Let $x_j$ be a vertex
of $S$ in the arc $E_j$ with $j\in\{1,2,3\}\setminus \{i\}$ incident
to $V_i$. $\{x_i,x_j\}\subseteq S$ is $\HH$-untamed, contradicting
that $S$ is $\HH$-tamed.
\end{proof}

\begin{proof}[Proof of Theorems~\ref{theorem:theorem1.1} and~\ref{theorem:theorem3.1}]
The if direction of Theorem~\ref{theorem:theorem3.1} follows from
Lemma~\ref{lemma:lemma3.5}\eqref{item1:lemma3.5}.  To see the only-if
direction of Theorem~\ref{theorem:theorem3.1}, let $\HH$ be an $X$-web
with maximum $|X|$ as ensured by Lemma~\ref{lemma:lemma2.1}.  If
$\HHdagger$ were not taming, then any $\HHdagger$-wild $Y$ would
be $\HH$-sticky by Lemma~\ref{lemma:lemma3.2}\eqref{item1:lemma3.2},
which in turn implies an $X\cup Y$-web by
Lemma~\ref{lemma:lemma3.2}\eqref{item2:lemma3.2}, contradicting the
maximality of $\HH$.  Thus Theorem~\ref{theorem:theorem3.1} follows.
By Lemmas~\ref{lemma:lemma2.1} and \ref{lemma:lemma3.2} and the if
direction of Theorem~\ref{theorem:theorem3.1},
Algorithm~\ref{algorithm:A} correctly detects saplings in $G$.  Thus,
Theorem~\ref{theorem:theorem1.1} follows from
Lemma~\ref{lemma:lemma3.3}.
\end{proof}

Lemma~\ref{lemma:lemma3.3} is not needed in the above reduction of
Theorem~\ref{theorem:theorem3.1} or else our proof of
Theorem~\ref{theorem:theorem2.2} would not be shorter than that
in~\cite{ChudnovskyS10}.  To complete proving
Theorems~\ref{theorem:theorem2.2} and~\ref{theorem:theorem3.1}, we
prove Lemma~\ref{lemma:lemma3.2} in~\S\ref{section:section4}.  After
that, to complete proving Theorem~\ref{theorem:theorem1.1}, we prove
Lemma~\ref{lemma:lemma3.3} in~\S\ref{section:section5}.
 
\section[]{Proving Lemma~\ref{lemma:lemma3.2}}
\label{section:section4}
The following lemma is needed in the proofs of
Lemma~\ref{lemma:lemma3.2}\eqref{item1:lemma3.2}
in~\S\ref{subsection:subsection4.1} and
Lemma~\ref{lemma:lemma3.2}\eqref{item2:lemma3.2}
in~\S\ref{subsection:subsection4.2}.  For any chunk $C$ of a net
$\HH$, the {\em arc set} $\CC$ of $\HH$ for $C$ consists of the arcs
of $\HH$ that intersect $C$.

\begin{lem}
\label{lemma:lemma4.1}
Let $\HH$ be an $X$-web.
\begin{enumerate*}[label=(\arabic*), ref=\arabic*]
\item  
\label{item1:lemma4.1}
If $Y$ is an $\HHdagger$-wild set, then $Y$ is $\HHdagger$-podded
if and only if $Y$ is $\HH$-podded.

\item 
\label{item2:lemma4.1}
Each $\HHdagger$-solid subset of $X$ is $\HH$-solid.
\end{enumerate*}
\end{lem}

\begin{proof}
Statement~\ref{item1:lemma4.1}: The only-if direction is
straightforward, since each $V_1V_2$-chunk of $\HHdagger$ is a
$V_1V_2$-chunk of $\HH$.  For the if direction, let $C$ be a
$V_1V_2$-chunk of $\HH$ that satisfies all
Conditions~\ref{condition:P} for $Y$.  The maximal chunk of $\HH$
containing $C$ is an arc $E^\dagger=W_1W_2$ of $\HHdagger$.  By
$N(Y,X)\subseteq V_1\cup C\cup V_2\subseteq W_1\cup E^\dagger \cup
W_2$, Condition~\ref{P1} holds for $E^\dagger$ in $\HHdagger$.
Since $Y$ is $\HHdagger$-untamed, $N(Y,X)$ intersects $(W_1\cup
W_2)\setminus E^\dagger$, implying $\{V_1,V_2\}\cap
\{W_1,W_2\}\ne\varnothing$.  Let $V_1=W_1$ and $W_1\setminus
E^\dagger\subseteq V_1\setminus C\subseteq N(Y,X)$ without loss of
generality.  If $N(Y,X)$ does not intersect $W_2\setminus E^\dagger$,
then Condition~\ref{P2} holds for $Y$ in $\HHdagger$.  Otherwise, we
have $V_2=W_2$ and $W_2\setminus E^\dagger\subseteq V_2\setminus
C\subseteq N(Y,X)$, also implying Condition~\ref{P2} of $Y$ in
$\HHdagger$.  Thus, $E^\dagger$ is a pod of $Y$ in $\HHdagger$.

Statement~\ref{item2:lemma4.1}: It suffices to consider the case that
the $\HHdagger$-solid subset $S$ of $X$ is not a node of $\HH$,
implying that $S$ is not a node of $\HHdagger$.  Let
$\WW=\{W_1,W_2\}$ for the arc $C=W_1W_2$ of $\HHdagger$ with
$S\subseteq C$.  $G[C]$ contains no $(S,W_1,W_2)$-sprout.  The rest of
the proof lets all sprouts be $(S,W_1,W_2)$-sprouts unless explicitly
specified otherwise.  Let $E_i$ with $1\leq i\leq|\CC|$ be the arcs in
the arc set $\CC$ of $\HH$ for $C$.  Let $\VV_i$ consist of the
end-nodes of $E_i$.  For any $i$ and $j$ that may not be distinct, let
$\PP_{i,j}$ and $\QQ_{i,j}$ be disjoint $\WW\VV_i$-rung and
$\WW\VV_j$-rung of $\HH$.  Let $P_{i,j}$ be a $\PP_{i,j}$-rung of $G$.
Let $Q_{i,j}$ be a $\QQ_{i,j}$-rung of $G$.  Let $U_{i,j}$ be the
end-node of $\PP_{i,j}$ in $\VV_i$.  Let $V_{i,j}$ be the end-node of
$\QQ_{i,j}$ in $\VV_j$.  If $S\subseteq E_i$ for an
$i\in\{1,\ldots,|\CC|\}$, then $G[E_i]$ contains no
$(S,U_{i,i},V_{i,i})$-sprout $T$ or else $G[T\cup P_{i,i}\cup
  Q_{i,i}]$ would be a sprout of $G[C]$.  Thus, $S$ is $\HH$-solid.
The rest of the proof assumes for contradiction that $S$ intersects
two or more arcs of $\CC$.

We first show that $S$ is contained by a node of $\HH$.  For any
distinct $i$ and $j$ such that $S$ intersects both $E_i$ and $E_j$,
let $r$ be an arbitrary vertex in $S\cap E_i$ and $s$ be an arbitrary
vertex in $S\cap E_j$.  Let $P=G[P_{i,j}\cup P']$ and $Q=G[Q_{i,j}\cup
  Q']$ for arbitrary $rU_{i,j}$-rung $P'$ of $G[E_i]$ and
$sV_{i,j}$-rung $Q'$ of $G[E_j]$.  By Conditions~\ref{N2}
and~\ref{N5}, $P-r$ and $Q-s$ are disjoint and nonadjacent, implying
that $r$ and $s$ are adjacent or else $G[P\cup Q]$ would contain a
sprout of Type~\ref{S2} in $G[C]$.  Since $r$ and $s$ are arbitrary,
Condition~\ref{N6} implies $S\subsetneq U$ for a node $U$ of $\HH$: If
$S$ is not contained by any node of $\HH$, then $S$ is contained by
$\Delta(V_1,V_2,V_3)$ and intersects $V_1\cup V_2$, $V_2\cap V_3$, and
$V_3\cap V_1$ for nodes $V_1,V_2,V_3$ of $\HH$.  Let
$\VV=\{V_1,V_2,V_3\}$.  Let $\PP_i$ and $\PP_j$ with $\{i,j\}\subseteq
\{1,2,3\}$ be disjoint $\VV\WW$-rungs of $\CC$ such that $V_i$ and
$V_j$ are the end-nodes of $\PP_i$ and $\PP_j$ in $\VV$.  $G[P_i\cup
  \{v\}\cup P_j]$ for $v\in S\cap V_i\cap V_j$ and $\PP_i$-rung $P_i$
and $\PP_j$-rung $P_j$ of $G[C]$ is a sprout of Type~\ref{S1},
contradiction.

For any arcs $E_i=UV_i$ and $E_j=UV_j$ of $\CC$ with $V_i\ne V_j$, we
say that $E_i$ {\em evades} $E_j$ if there are disjoint $\VV\WW$-rungs
$\PP$ and $\QQ$ of $\HH$ with $\VV=\{V_i,V_j\}$ such that $\PP\cup
\QQ$ does not intersect~$U$.  $E_i$ evades $E_j$ if and only if $E_j$
evades $E_i$.  If $E_i$ evades $E_j$ and $E_i$ intersects $S$, then
$E_j\cap U\subseteq S$ or else $G[C]$ would contain a sprout
$G[P_i\cup Q_j\cup P\cup Q]$ of Type~\ref{S1}, where $P_i$ is an
$E_i$-rung intersecting $S$, $Q_j$ is an $E_j$-rung intersecting
$U\setminus S$, $P$ is a $\PP$-rung, and $Q$ is a $\QQ$-rung.

By $S\subsetneq U$, $E_j\cap U\not\subseteq S$ holds for an arc
$E_j=UV_j$.  If each arc $E_i=UV_i$ intersecting $S$ satisfies
$V_i=V_j$, then $G[P_{i,i}\cup Q_{i,i}\cup R]$ for any $UV_i$-rung $R$
of $G[E_i]$ that intersects $S$ is a sprout of Type~\ref{S1},
contradiction.  Thus, an arc $E_i=UV_i$ with $V_i\ne V_j$ intersects
$S$. By $E_j\cap U\nsubseteq S$ and $E_i\cap S\ne\varnothing$, $E_i$
does not evade $E_j$.  We show contradiction by identifying an arc
$E_k=UV_k$ with $V_k\notin\{V_i,V_j\}$ such that $E_i$ evades $E_k$,
implying $E_k\cap U\subseteq S$, and $E_k$ evades $E_j$, implying
$E_k\cap S=\varnothing$.  Let $\PP_i$ and $\PP_j$ be disjoint
$\VV\WW$-rungs of $\HH$ with $\VV=\{V_i,V_j\}$.  Since $E_i$ does not
evade $E_j$, $\PP_i\cup \PP_j$ intersects $U$. Let $\PP_j$ intersect
$U$ without loss of generality.  $U$ is the neighbor of $V_j$ in
$\PP_j$.  Let $E_k=UV_k$ be the incident arc of $U$ in $\PP_j$ with
$V_k\ne V_j$.  Let $\QQ=\PP_j-\{U,V_j\}$.  Since $\PP_i$ and $\QQ$ are
disjoint $\VV\WW$-rungs of $\HH$ with $\VV=\{V_i,V_k\}$ and $\PP_i\cup
\QQ$ does not intersect $U$, $E_i$ evades $E_k$.  Let $\RR'$ be a rung
of $(\CC \cup \{W_1W_2\})-U$ between $V_j$ and $\PP_i$.  $\RR'$ does
not intersect $\QQ$ or else $E_i$ would evade $E_j$.  Let $\RR$ be the
$V_j\WW$-rung of $\PP_i\cup \RR'$.  Since $\QQ$ and $\RR$ are disjoint
$\VV\WW$-rungs of $\HH$ with $\VV=\{V_k,V_j\}$ and $\QQ\cup \RR$ does
not intersect $U$, $E_k$ evades $E_j$.
\end{proof}

\subsection[]{Proving Lemma~\ref{lemma:lemma3.2}\eqref{item1:lemma3.2}}
\label{subsection:subsection4.1}

A net {\em self-aids} if it aids itself.  Since the aiding net of any
web self-aids, Lemma~\ref{lemma:lemma3.2}\eqref{item1:lemma3.2} is
immediate from Lemma~\ref{lemma:lemma4.2} by
Lemma~\ref{lemma:lemma4.1}.

\begin{lem}
\label{lemma:lemma4.2}
For self-aiding $X$-net $\HH_0$ and $\HH_0$-wild $\HH_0$-nonsticky set
$Y$, $G[X\cup Y]$ contains a sapling.
\end{lem}

The rest of the subsection proves Lemma~\ref{lemma:lemma4.2} using
Lemmas~\ref{lemma:lemma4.3},~\ref{lemma:lemma4.4},
and~\ref{lemma:lemma4.5}.  Let $\LL$ consist of the leaves of the
self-aiding net $\HH$ in
Lemma~\ref{lemma:lemma4.3},~\ref{lemma:lemma4.4},
or~\ref{lemma:lemma4.5}.  Since $\nabla(\HH)$ is triconnected, each
nonleaf node of $\HH$ has degree at least three in $\HH$ and any
three-node set $\UU$ of $\HH$ admits pairwise disjoint $\UU\LL$-rungs
$\PP_1,\PP_2,\PP_3$ of $\HH$.  By Condition~\ref{N6} of $\HH$, any
$\PP_i$-rungs $P_i$ of $G$ with $i\in\{1,2,3\}$ are pairwise disjoint
and nonadjacent.

\begin{lem}
\label{lemma:lemma4.3}
If $Y$ is an $\HH$-wild $\HH$-nonsticky set for a self-aiding $X$-net
$\HH$ of $G$ such that $N_G(Y,X)=M_1\cup M_2$ and each of $M_1$ and
$M_2$ is contained by a node or arc of $\HH$, then $G[X\cup Y]$
contains a sapling.
\end{lem}

\begin{proof}
Let $N=N_G(Y,X)$.  We start with proving the following statement.
 
Claim~1: {\em If $M_i\subseteq U$ with $\{i,j\}=\{1,2\}$ holds for a
  node $U$ and $M_j\subseteq U_1\cup F$ holds for an end-node $U_1$ of
  an arc $F$ with $U_1\ne U$, $U\setminus F\nsubseteq M_i$, and
  $M_i\nsubseteq F$, then $G[X\cup Y]$ contains a sapling.}

Let $\RR_1=\{U_1\}$.  Since the degree of $U$ is at least three,
$U\setminus F\nsubseteq M_i$ and $M_i\nsubseteq F$ imply that the node
set consisting of the neighbors of $U$ other than $U_1$ in $\HH$
admits a nonempty disjoint partition $\RR_2$ and $\RR_3$ such that (a)
each arc between $U$ and $\RR_2$ intersects $M_i$ and (b) each arc
between $U$ and $\RR_3$ intersects $U\setminus M_i$.  Let $\HHprime$ be
the triconnected graph obtained from $\nabla(\HH)$ by (1) replacing
node $U$ and its incident arcs with a triangle on a set
$\WW=\{W_1,W_2,W_3\}$ of three new nodes and (2) adding an arc between
$W_i$ and each node in $\RR_i$ for all $i\in\{1,2,3\}$.  There are
pairwise disjoint $\WW\LL$-rungs $\PP_1,\PP_2,\PP_3$ of $\HHprime$ such
that each $\PP_i$ with $i\in\{1,2,3\}$ is a $W_iL_i$-rung with $L_i\in
\LL$.  Let $\QQ_1$ be the path of $\HH$ consisting of arc $F$ and path
$\PP_1-W_1$.  Let $Q_1$ be the $NL_1$-rung of a $\QQ_1$-rung of $G$
intersecting $M_j$.  Let $\QQ_2$ be the $L_2L_3$-path of $\HH$
obtained from $\PP_2\cup \PP_3$ by replacing the two arcs $W_2U_2$ and
$W_3U_3$ with the two arcs $UU_2$ and $UU_3$.  Let $Q_2$ be a
$\QQ_2$-rung of $G$ intersecting exactly one vertex in $M_i\cap U$.
$N$ intersects each of $Q_1$ and $Q_2$ at exactly one vertex.  Thus,
$G[Y\cup Q_1\cup Q_2]$ contains a sapling of $G[X\cup Y]$.  Claim~1 is
proved.

Claim~2: {\em If $G[X\cup Y]$ is sapling-free, then each $M_i$ with
  $i\in\{1,2\}$ is $\HH$-solid}.

To prove Claim~2 by Claim~1, let each $M_i$ with $i\in\{1,2\}$ be
contained by a node $V_i$ or an arc $E_i$. We first show that if
$M_1\subseteq V_1$ and $M_2\subseteq V_2$, then $V_1V_2$ is not an
arc.  Assume an arc $E=V_1V_2$ for contradiction.  Since $Y$ is
$\HH$-wild and $\HH$-unpodded, we have $V_i\nsubseteq (E\cup M_i)$ and
$M_i\nsubseteq E$ for $\{i,j\}=\{1,2\}$, contradicting Claim~1 with
$U=V_i$, $U_1=V_j$, and $F=E$.

To see Claim~2(a): {\em $M_i\subseteq V_i$ for $\{i,j\}=\{1,2\}$
  implies $M_i=V_i$}, assume $V_i\nsubseteq M_i$.  If $M_j\subseteq
V_j$, then $V_iV_j$ is not an arc, contradicting Claim~1 with $U=V_i$,
$U_1=V_j$, and $F$ being an incident arc of $V_j$.  $M_j\subseteq E_j$
contradicts Claim~1 with $U=V_i$, $F=E_j$, and $U_1$ being an end-node
of $E_j$ that is not $V_i$.  To see Claim~2(b): {\em $M_i\subseteq
  E_i=UV$ for $\{i,j\}=\{1,2\}$ implies that $M_i$ is $\HH$-solid},
assume an $(M_i,U,V)$-sprout $T$ of $G[E_i]$.  If $M_j\subseteq V_j$,
then let $W=V_j$.  By Claim~2(a), $W$ is not incident to $E_i$ or else
$E_i$ would be a pod of $Y$ in $\HH$.  If $M_j\subseteq E_j$, then let
$W$ be an end-node of $E_j$ not incident to $E_i$.  Let
$\UU=\{U,V,W\}$.  Let $\PP_1,\PP_2,\PP_3$ be pairwise disjoint
$\UU\LL$-rungs of $\HH$.  Let each $P_k$ with $k\in\{1,2,3\}$ be a
$\PP_k$-rung of $G$.  $G[P_1\cup P_2\cup P_3\cup Y\cup T\cup E_j]$
contains a sapling.

To prove the lemma by Claim~2, assume for contradiction that $G[X\cup
  Y]$ is sapling-free.  Since $Y$ is $\HH$-nonsolid, Claim~2 implies
an internal vertex $y$ of path $G[Y]$ with nonempty $N_G(y,X)\subseteq
M_1\cap M_2$.  By Condition~\ref{N2}, $M_i=V_i$ holds for
$\{i,j\}=\{1,2\}$.  By Condition~\ref{N5}, if $M_j\subseteq V_j$, then
$\HH$ has an arc $E=V_iV_j$, which is a pod of $Y$ in $\HH$; and if
$M_j\subseteq E_j$, then $V_i$ is incident to $E_j$, which is thus a
pod of $Y$ in $\HH$. Both cases contradict that $Y$ is
$\HH$-nonsticky.
\end{proof}

If $Y$ is $\HH$-wild for an $X$-net $\HH$, then let $\ell(Y,\HH,G)$
denote the minimum number of $\HH$-tamed subsets of $X$ whose union is
$N_G(Y,X)$.  A net is {\em simple} if all of its arcs are simple.  If
$\HH$ is a simple self-aiding $X$-net of $G$, then $G[X]$ is
isomorphic to the line graph of a subdivision of $\HH$.

\begin{lem}
\label{lemma:lemma4.4}
If $Y$ is an $\HH$-wild set for a simple self-aiding $X$-net $\HH$ of
$G$ with $\ell(Y,\HH,G)=2$ such that $N_G(Y,X)$ contains a triad of
$\HH$, then $G[X\cup Y]$ contains a sapling.
\end{lem}

\begin{proof}
Let $\UU=\{U_1,U_2,U_3\}$ for nodes with
$\Delta=\Delta(U_1,U_2,U_3)\subseteq N=N_G(Y,X)$.  Let
$\PP_1,\PP_2,\PP_3$ be pairwise disjoint $\UU\LL$-rungs of $\HH$.  For
each $i\in\{1,2,3\}$, let $L_i\in\LL$ such that the $\PP_i$-rung $P_i$
of $G$ is a $U_iL_i$-rung.  Since $N$ is untamed, $N\setminus
\Delta\ne\varnothing$.  Since $\HH$ is simple, each arc intersecting
$N\setminus \Delta$ is incident to at most one node of $\UU$.  By
$\ell(Y,\HH,G)=2$, $N\setminus \Delta$ intersects at most one of
$P_1$, $P_2$, and $P_3$.  If $N$ intersects $P_i$ for
$\{i,j,k\}=\{1,2,3\}$, then $G[Y\cup Q_i\cup (U_j\cap U_k)\cup P_j\cup
  P_k]$ contains a sapling for the $NL_i$-rung $Q_i$ of $P_i$.  It
remains to consider $(N\setminus \Delta)\cap V(P_1\cup P_2\cup
P_3)=\varnothing$.

Case~1: Each arc $E$ intersecting $N\setminus\Delta$ satisfies $|E|=1$
and is incident to $\PP_i$ and $\PP_j$ for $\{i,j,k\}=\{1,2,3\}$.  Let
$E$ be an arc intersecting $N\setminus \Delta$.  Let $V_i\in V(\PP_i)$
and $V_j\in V(\PP_j)$ be end-nodes of $E$ with $U_i\ne V_i$.  Let
$\QQ_1$ be the $U_iL_k$-path of $\HH$ consisting of arc $E_j=U_iU_k$
and $\PP_k$.  Let $Q_1$ be the $\QQ_1$-rung of $G$.  Let $\QQ_2$ be
the $L_iL_j$-path of $\HH$ consisting of $E$, the $V_iL_i$-rung of
$\PP_i$, and the $V_jL_j$-rung of $\PP_j$.  Let $Q_2$ be the
$\QQ_2$-rung of $G$.  By $U_i\ne V_i$, $\QQ_1$ and $\QQ_2$ are
disjoint.  By $(N\setminus\Delta)\cap V(P_1\cup P_2\cup
P_3)=\varnothing$, $Q_1$ (respectively, $Q_2$) intersects $N$ exactly
at the vertex in arc $E_j$ (respectively, $E$).  Thus, $G[Y\cup
  Q_1\cup Q_2]$ contains a sapling of $G[X\cup Y]$.

Case~2: An arc $E$ intersecting $N\setminus\Delta$ violates the
condition of Case~1.  Let $\QQ$ be a shortest path of $\HH$ between
$V(E)$ and $V(\PP_1\cup \PP_2\cup \PP_3)$. Since $E$ violates the
condition of Case~1, we may require that if $U\in V(E)$ and $V_i\in
V(\PP_i)$ with $\{i,j,k\}=\{1,2,3\}$ are the end-nodes of $\QQ$, then
the $NU$-rung $Q_i$ of $G[E]$ is not adjacent to $P_j\cup P_k$.  Let
$E_i=U_jU_k$.  Let $Q$ be the $\QQ$-rung of $G$.  Let $R_i$ be the
$V_iL_i$-rung of $P_i$.  $G[P_j\cup P_k\cup E_i\cup Q_i\cup Q\cup
  R_i]$ contains a sapling of $G[X\cup Y]$.
\end{proof}

\begin{lem}
\label{lemma:lemma4.5}
Let $Y$ be an $\HH$-wild set for a simple self-aiding $X$-net $\HH$ of
graph $G$ with $\ell(Y,\HH,G)\geq 3$.  If $G[X\cup Y]$ is
sapling-free, then $Y$ is $\HH$-podded for $G$.
\end{lem}

\begin{proof}
Since $Y$ is $\HH$-wild with $\ell=\ell(Y,\HH,G)\geq 3$, $Y$ consists
of a vertex $y$.  Let $N_1,\ldots,N_\ell$ be pairwise disjoint
$\HH$-tamed subsets of $X$ whose union is $N=N_G(Y,X)$.  Let $L$
consist of the leaves of $H=G[X\cup Y]$.  Let each graph $H_{i,j}$
with $1\leq i<j\leq \ell$ be obtained from $H$ by deleting the edges
between $y$ and $N\setminus (N_i\cup N_j)$.  We claim that {\em each
  $H_{i,j}$ is sapling-free}.  Since $\HH$ is a simple self-aiding
$X$-net of $H_{i,j}$ with $\ell(Y,\HH,H_{i,j})=2$, $Y$ is $\HH$-sticky
for $H_{i,j}$ by Lemmas~\ref{lemma:lemma4.3} and~\ref{lemma:lemma4.4}.
Thus, each $N_i$ with $i\in\{1,\ldots,\ell\}$ is either contained by a
node or arc of $\HH$.  Assume that $N_1,\ldots,N_k$ are $\HH$-solid
and $N_{k+1},\ldots,N_\ell$ are not.  If $k<\ell$, then $Y$ is
$\HH$-podded for all $H_{i,\ell}$ with $i\in\{1,\ldots,\ell-1\}$.  If
$N_\ell$ is contained by a node $U$, then there is exactly one vertex
$u$ in $U\setminus N_\ell$, implying $\ell=3$ and that the arc
containing $u$ is a pod of $Y$ in $\HH$ for $G$.  If $N_\ell$ is not
contained by a node, then $\ell=3$ and the arc containing $N_\ell$ is
a pod of $Y$ in $\HH$ for $G$.  As for $k=\ell$, observe that there
cannot be a $3$-node set $\UU=\{U_{i_1},U_{i_2},U_{i_3}\}$ with
$\{i_1,i_2,i_3\}\subseteq \{1,\ldots,\ell\}$ such that each node
$U_{i_j}$ with $j\in\{1,2,3\}$ is either a solid set $N_{i_j}$ or an
end-node of the arc $E_{i_j}$ containing a solid set $N_{i_j}$: Assume
for contradiction that such a $\UU$ exists.  Let vertex set $E$ be the
union of the arcs $E_{i_j}$ with $N_{i_j}\subseteq E_{i_j}$. Let
$\PP_1$, $\PP_2$, and $\PP_3$ be pairwise disjoint $\UU\LL$-rungs of
$\HH$.  For each $j\in\{1,2,3\}$, let $P_j$ be a $\PP_j$-rung of $G$.
$G[Y\cup P_1\cup P_2\cup P_3\cup E]$ contains a sapling,
contradiction.  The observation implies $\ell=3$ and that $Y$ is
$\HH$-podded for $G$.

To prove the claim, assume a sapling $T$ of $H_{i,j}$.  Since any edge
in $H[T]\setminus T$ is between $y$ and $N_{i,j}$, the following
statements hold or else $H$ would contain a sapling in which $y$ is
the degree-$3$ vertex: (1) The degree of $y$ in $T$ is two. (2)
$H[T]\setminus T$ has exactly one edge~$e$, implying that $y$ and a
vertex $u_1\in N_{i,j}$ are the end-vertices of $e$.  (3) The
degree-$3$ vertex $u_2$ of $T$ is adjacent to $y$ and $u_1$ in $T$,
implying $u_2\in N_i\cup N_j$.  That is, $H[T]$ consists of a triangle
on $U=\{y,u_1,u_2\}$ and pairwise disjoint $UL$-rungs $P_1,P_2,P_3$ of
$T$ with $N\cap V(P_1)=\{u_1\}\subseteq N_{i_1}$, $N\cap
V(P_2)=\{u_2\}\subseteq N_{i_2}$, and $y\in V(P_3)$ for distinct $i_1$
and $i_2$ in $\{1,\ldots,\ell\}$.  By
Lemma~\ref{lemma:lemma3.5}\eqref{item2:lemma3.5}, each $N_{i_k}$ with
$k\in\{1,2\}$ is contained by a node, arc, or triad $S_k$.  $N\cap
(S_1\cup S_2)$ is $\HH$-untamed or else there would be $\ell-1$
pairwise disjoint $\HH$-tamed subsets of $X$ whose union is $N$.  Let
each $E_k$ with $k\in\{1,2\}$ be the simple arc with $u_k\in E_k$.  We
show that $H$ contains a sapling in which $y$ is the degree-$3$
vertex.

Case~1: If $E_1=E_2$.  Since $N\cap (S_1\cup S_2)$ is $\HH$-untamed, a
vertex $v_k\in S_k\cap (N\setminus E_k)$ with $k\in\{1,2\}$.  Since
$\HH$ is simple and $\{u_k,v_k\}\subseteq S_k$, $S_k$ is not a triad.
$S_k$ is not an arc or else $E_k=S_k$ would intersect $N\setminus
E_k$.  Thus, $S_k$ is an end-node of $E_k$ with $\{u_k,v_k\}\subseteq
S_k$ and $u_{3-k}\notin S_k$.  By $u_k\in S_k$ and Condition~\ref{N6},
$S_k$ is not adjacent to $(P_3-y)\cup (P_{3-k}-u_{3-k})$ in $H$.
Since $\HH$ is simple, $v_k\in N\cap S_k$ implies that $H[(T-u_k)\cup
  \{v_k\}]$ is a sapling of $H$.

Case~2: $E_1\ne E_2$.  By Condition~\ref{N5}, $\{u_1,u_2\}\subseteq V$
for a common end-node $V$ of arcs $E_1$ and $E_2$.  By
Condition~\ref{N6}, $E_k\subseteq V(P_k)$ for each $k\in\{1,2\}$.
Since $N\cap (S_1\cup S_2)$ is $\HH$-untamed, a vertex $v_k\in S_k\cap
(N\setminus V)$ with $k\in\{1,2\}$.  Since $\HH$ is simple and
$\{u_k,v_k\}\subseteq S_k$, $S_k$ is not a triad.  $S_k$ is not an arc
or else $\{u_k,v_k\}\subseteq S_k=E_k\subseteq V(P_k)$ would
contradict $N\cap V(P_k)=\{u_k\}$.  Thus, $S_k$ is a node.  By $u_k\in
E_k\cap S_k$, $S_k$ is an end-node of $E_k$ containing $u_k$.  By
$v_k\in S_k\setminus V$, $S_k\ne V$.  By $u_k\in S_k$ and
Condition~\ref{N6}, $S_k$ is not adjacent to $V(P_3-y)\cup
V(P_{3-k}-u_{3-k})$ in $H$.  Since $\HH$ is simple, $v_k\in N\cap S_k$
implies that $H[(T-u_k)\cup \{v_k\}]$ is a sapling of $H$.
\end{proof}

\begin{proof}[Proof of Lemma~\ref{lemma:lemma4.2}]
Assume for contradiction that $G[X\cup Y]$ is sapling-free.  A vertex
set $D\subseteq X$ is an {\em inducing} set of $\HH_0$ if $G[E_0\cap
  D]$ for each arc $E_0=U_0V_0$ is an $U_0V_0$-rung of $G[E_0]$.  For
any inducing set $D$ of $\HH_0$, let $\HH_0(D)$ denote the simple
self-aiding $D$-net of graph $H_0(D)=G[Y\cup D]$ obtained from $\HH_0$
by replacing each arc $E_0$ of $\HH_0$ with the arc $E=E_0\cap D$ and
replacing each node $V_0$ of $\HH_0$ with the node $V=V_0\cap D$.  Let
$N=N_G(Y,X)$.  Let $\ell=\ell(Y,\HH_0,G)$. If $\ell=2$, then
Lemma~\ref{lemma:lemma4.3} implies $N\nsubseteq S_1\cup S_2$ for any
node or arc $S_i$ of $\HH_0$ with $i\in\{1,2\}$.  Thus, $N$ contains a
triad $\Delta$ and $N\setminus \Delta$ is not contained by any arc of
$\HH_0$ between two nodes of $\Delta$.  By $\ell=2$, there is an
inducing set $D$ of $\HH_0$ with $\ell(Y,\HH_0(D),H_0(D))=2$ and
$\Delta\subseteq N_{H_0(D)}(Y,D)$, contradicting
Lemma~\ref{lemma:lemma4.4}.  Thus, $\ell\geq 3$, implying a
three-vertex set $S\subseteq N$ such that every two-vertex subset of
$S$ is $\HH_0$-untamed.  Let $D$ be an inducing set of $\HH_0$ with
$S\subseteq D$, implying $\ell(Y,\HH_0(D),H_0(D))\geq 3$.  By
Lemma~\ref{lemma:lemma4.5}, there is a pod $E=UV$ of $Y$ in $\HH_0(D)$
for $H_0(D)$ such that $N_{H_0(D)}(Y)$ intersects $E\setminus (U\cup
V)$, $U\setminus E$, and $V\setminus E$.  Let $E_0=U_0V_0$ be the arc
of $\HH_0$ with $E=E_0\cap D$, $U=U_0\cap D$, and $V=V_0\cap D$.
Since $E_0$ is not a pod of $Y$ in $\HH_0$ and $N$ intersects
$E_0\setminus (U_0\cup V_0)$, $U_0\setminus E_0$, and $V_0\setminus
E_0$, a vertex $x$ belongs to $N\setminus (U_0\cup E_0\cup V_0)$ or
$(U_0\cup V_0)\setminus (E_0\cup N)$.  Let $D'$ be an inducing set
$(D\setminus E_0)\cup V(P)$ of $\HH_0$, where $E_0=U_0V_0$ is the arc
of $\HH_0$ containing $x$ and $P$ is a $U_0V_0$-rung of $G[E_0]$
containing $x$.  One can verify that $Y$ is $\HH_0(D')$-unpodded for
$H_0(D')$ with $\ell(Y,\HH_0(D'),H_0(D'))\geq 3$, contradicting
Lemma~\ref{lemma:lemma4.5}.
\end{proof}

\subsection[]{Proving Lemma~\ref{lemma:lemma3.2}\eqref{item2:lemma3.2}}
\label{subsection:subsection4.2}
This subsection shows that if $Y$ is $\HH$-sticky for an $X$-web
$\HH$, then $\HH$ can be expanded to an $X\cup Y$-web via
Subroutine~\ref{subroutine:B} below.  Let $\HH$ be an $X$-net.  For
any $\HH$-solid subset $S$ of $X$ contained by a simple arc $F=U_1U_2$
of $\HH$, define {\em Operation $\textsc{subdivide}(S)$} to (1) create
a new node $S$ and (2) replace the simple arc by new simple arcs
$SU_i$ with $i\in\{1,2\}$ consisting of the vertices of the
$SU_i$-rung of $G[F]$.  Define {\em Subroutine~\ref{subroutine:B}}
with $N=N(Y,X)$ as follows (see Figure~\ref{figure:figure7}):

\fbox{\em Subroutine~\ref{subroutine:B}}\\[-20pt]
\begin{enumerate}[label={}, ref={B}, leftmargin=0pt]
\item 
\label{subroutine:B}
\begin{enumerate}[label={\em Step~\ref{subroutine:B}\arabic*:},ref={\ref{subroutine:B}\arabic*},leftmargin=*,
itemindent=1.2cm]
\item
\label{B1}
{\em $Y$ is $\HH$-solid}.  Let $S_1$ and $S_2$ be $\HH$-solid sets
with $N=S_1\cup S_2$.

\begin{enumerate}[label={(\alph*)}, ref={\ref{B1}(\alph*)}, itemindent=-0.27cm]
\item 
\label{B11}
For each $i\in\{1,2\}$, if $S_i$ is contained by a simple arc, then
create node $S_i$ by $\textsc{subdivide}(S_i)$.

\item 
\label{B12}
Add each end-vertex $y$ of path $G[Y]$ into the nodes $S_i$ with
$i\in\{1,2\}$ and $S_i\subseteq N(y)$.

\item 
\label{B13}
Make a simple arc $Y=S_1S_2$.
\end{enumerate}

\item 
\label{B2}
{\em $Y$ is $\HH$-nonsolid}.  Thus, $Y$ is $\HH$-podded.  Let
$V_1V_2$-chunk $C$ of $\HH$ be a minimal pod of $Y$ in $\HH$. Since
$Y$ is $\HHdagger$-wild, assume $V_1\in V(\HHdagger)$ and
$V_1\subseteq C\cup N$ without loss of generality.
\begin{enumerate}[label={(\alph*)}, ref={\ref{B2}(\alph*)}, itemindent=-0.27cm]
\item 
\label{B21}
If $V_2$ is incident to exactly one arc $F=VV_2$ in the arc set for
$C$, $N\cap V_2\subseteq F$, and $F$ is simple, then $N$ intersects
$F\setminus V$ by the minimality of $C$.  Let $v_2$ be the end-vertex
of the $NV_2$-rung $P$ of $G[F]$ in $N$.  Let $v$ be the neighbor of
$v_2$ not in $P$.  Call $\textsc{subdivide}(\{v,v_2\})$ to create a
new node $V_2=\{v,v_2\}$.  Delete $V(P)$ from $C$ to preserve that $C$
is a $V_1V_2$-chunk that is a minimal pod of $Y$ in $\HH$.

\item 
\label{B22}
Update $\HH$ by $\textsc{merge}(C)$.  Let $E=V_1V_2$ be the arc of
$\HH$ with $E=C$.

\item
\label{B23}
Add $Y$ to arc $E$ and add each end-vertex $y$ of path $G[Y]$ to the
nodes $V_i$ with $V_i\subseteq C\cup N(y)$.
\end{enumerate}
\end{enumerate}
\end{enumerate}

\begin{figure}[t]
\centerline{\scalebox{0.47}{\includegraphics{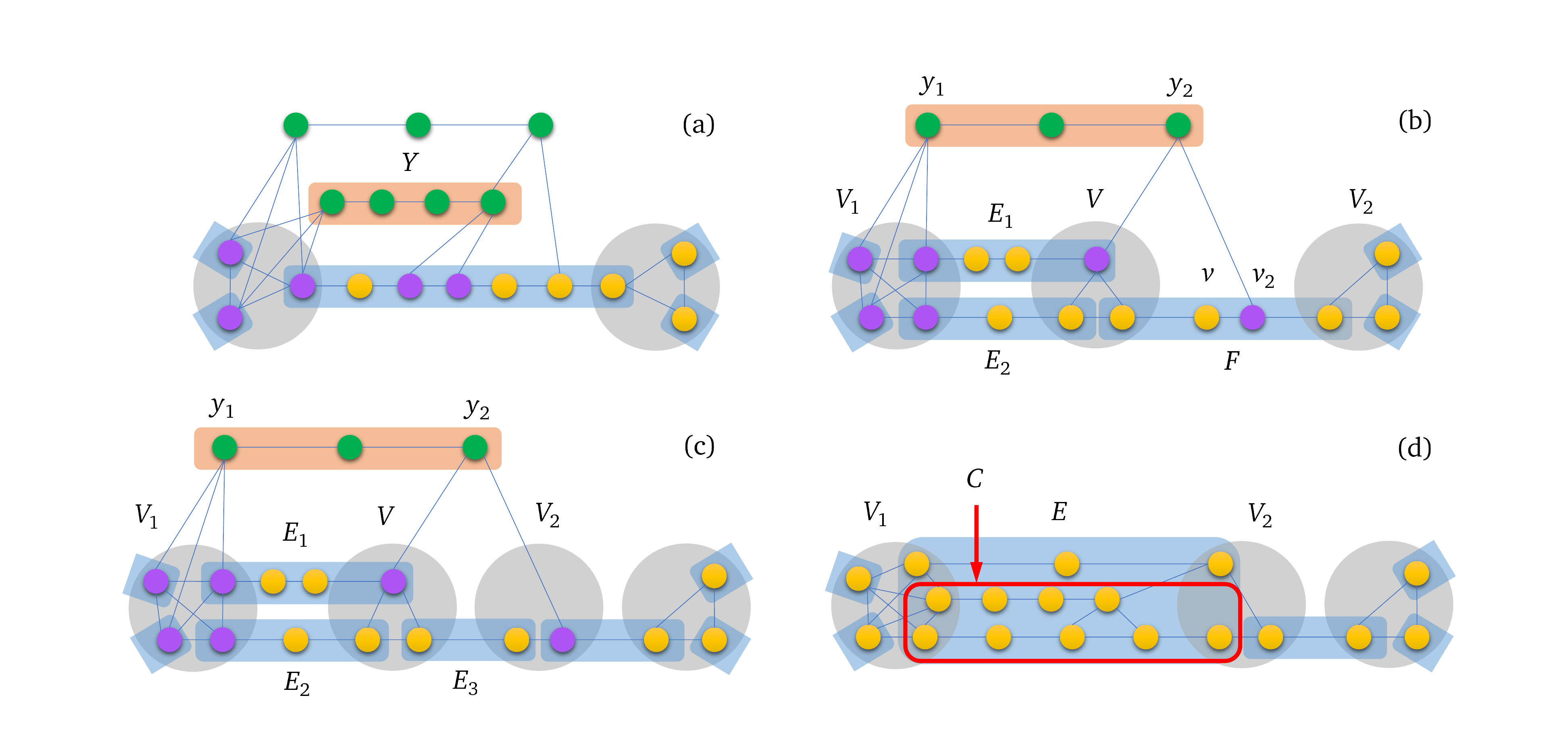}}}
\caption{Applying Step~\ref{B1} on the example in (a) results in the
  example in (b), in which $E_1\cup E_2\cup F$ is a minimal pod of the
  green $y_1y_2$-rung.  Applying Step~\ref{B21} on the example in (b)
  results in the example in (c), in which $E_1\cup E_2\cup E_3$ is a
  minimal pod of the green $y_1y_2$-rung. Applying Steps~\ref{B22}
  and~\ref{B23} on the example in (c) results in the example in (d).}
\label{figure:figure7}
\end{figure}

\begin{proof}[Proof of Lemma~\ref{lemma:lemma3.2}\eqref{item2:lemma3.2}] 
The resulting $\HH$ of Step~\ref{B1} is an $X\cup Y$-web, since all
steps preserve Conditions~\ref{condition:N} and all new arcs are
simple.  The rest of the proof shows that the resulting $\HH$ of
Step~\ref{B2} is also an $X\cup Y$-web.  At the beginning of
Step~\ref{B22} one can verify that, no matter whether $\HH$ is updated
by Step~\ref{B21} or not, $Y$ is $\HH$-nonsolid and $\HH$-podded and
$\HH$ is an $X$-web with the following {\em Condition~F}: If $V_2$ is
incident to exactly one arc $F$ in the arc set for the minimal pod $C$
of $Y$ in $\HH$ and $F$ is simple, then $N(Y,V_2)$ intersects
$V_2\setminus C$.  By Lemma~\ref{lemma:lemma3.4}, $\HH$ is an $X$-net
(respectively, $X\cup Y$-net) at the end of Step~\ref{B22}
(respectively, Step~\ref{B23}).  It remains to show that $E=C\cup Y$
is a flexible arc by identifying an $(S,V_1,V_2)$-sprout of $G[E]$ for
any nonempty subset $S$ of $E$.  The rest of the proof lets $\HH$
denote the $X$-web at the beginning of Step~\ref{B22} and lets all
sprouts be $(S,V_1,V_2)$-sprouts of $G[E]$ unless specified otherwise.
Let $y_1$ and $y_2$ be the end-vertices of path $G[Y]$ with
$V_1\subseteq C\cup N(y_1,X)$.  If $|Y|=1$, then $y_1=y_2$.  If
$|Y|\geq 2$, then let $N_i=N(y_i,X)$ and $M_i=N(Y\setminus
\{y_{3-i}\},X)$ for each $i\in\{1,2\}$ and let $M=M_1\cap M_2$.  Let
$S_C=S\cap C$ and $S_Y=S\cap Y$.  If $S_C\ne\varnothing$, then $S_C$
is assumed to be $\HH$-solid, since any $(S_C,V_1,V_2)$-sprout of
$G[C]$ is a sprout.  If $S_Y \ne\varnothing$, then let each $P_i$ with
$i\in\{1,2\}$ be the $Sy_i$-rung of $G[Y]$.  Let $C^*=W_1W_2$ with
$W_1=V_1$ be the arc of $\HHdagger$ containing $C$.

Case~1: $S_C=\varnothing$.  $G[S]$ is an edge in $G[Y]$ or else a
$V_1V_2$-rung of $G[E]$ containing $Y$ contains a sprout of
Type~\ref{S1} or~\ref{S2}.  By $|S|=2$, $|Y|\geq 2$.  Since $Y$ is
$\HHdagger$-wild, $M_1\subseteq V_1$.  We may assume $M_1=V_1$,
since otherwise $G[P_1\cup Q]$ is a spout of Type~\ref{S3} for a
$V_1V_2$-rung $Q$ of $G[C]$ intersecting $V_1\setminus M_1$.
Case~1(a): {\em $M_2$ is $\HH$-nonsolid}.
Lemma~\ref{lemma:lemma4.1}\eqref{item2:lemma4.1} implies an
$(M_2,W_1,W_2)$-sprout $T^*$ of $G[C^*]$. Let $T=G[T^*[C]\cup P_2]$.
If $T^*$ is of Type~\ref{S1} or~\ref{S2}, then $T$ contains a sprout
of Type~\ref{S1}.  If $T^*$ is of Type~\ref{S3}, then $T$ is a sprout
of Type~\ref{S3}.  Case~1(b): {\em $M_2$ is $\HH$-solid}.  Since $M_1$
is $\HH$-solid and $Y$ is $\HH$-nonsolid, we have $M\ne\varnothing$
and $M\subseteq V_1\cap M_2$.  If $M_2$ were contained by a simple arc
$F$ of $\HH$, then $F=V_1V_2$ by $V_1\cap M_2\ne\varnothing$ and
minimality of $C$, contradicting Condition~F.  Thus, $M_2$ is a node
of $\HH$. By $V_1\cap M_2\ne\varnothing$, $F=V_1M_2$ is an arc of
$\HH$. By $M\subseteq V_1$, we have $M_2\subseteq F\cup N_2$.  By
minimality of $C$, $M_2=V_2$.  By $M\ne\varnothing$ and $|Y|\geq 3$,
$G[Y\cup M]$ contains a sprout of Type~\ref{S1}.
\begin{figure}[t]
\centerline{\scalebox{0.47}{\includegraphics{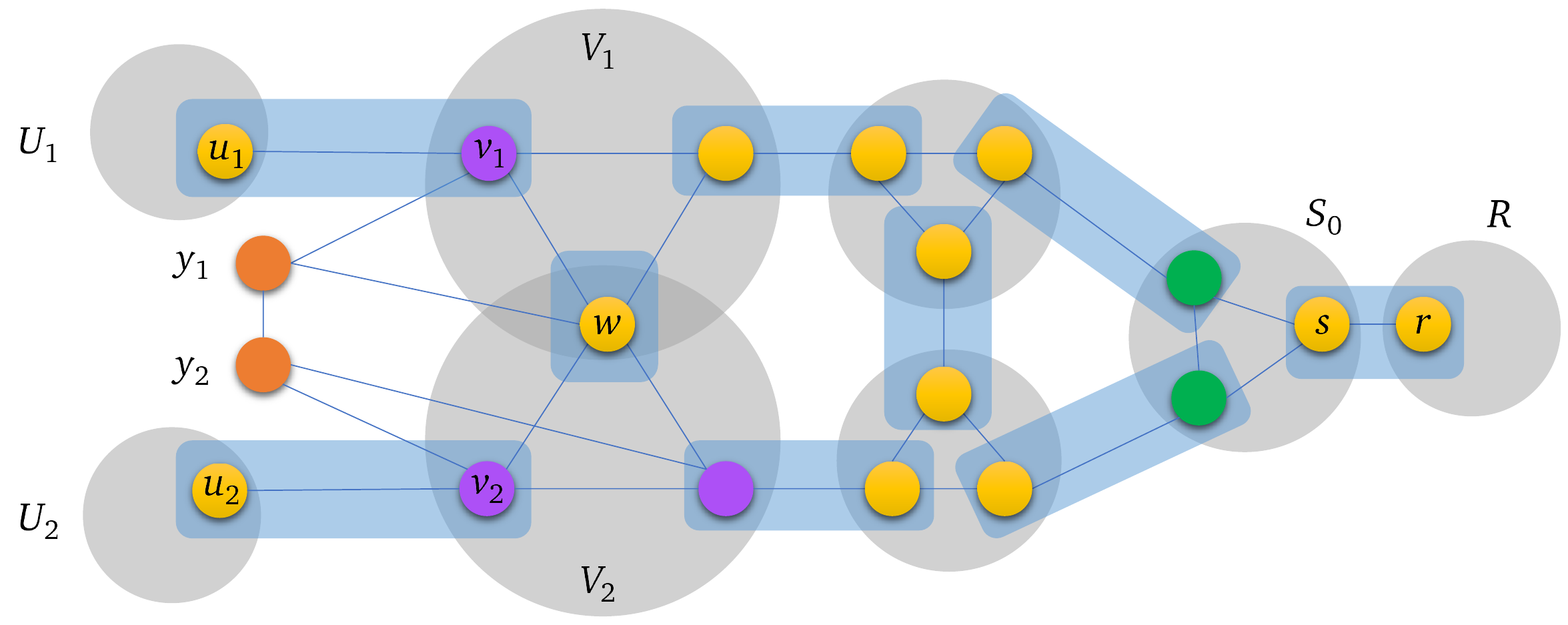}}}
\caption{An example of $\HH_0$.}
\label{figure:figure8}
\end{figure}

Case~2: $S_Y=\varnothing$.  $S=S_C$ is $\HH$-solid.  Let $v_1\in
V_1\setminus C$, $v_2\in V_2\setminus C$, and a set of new vertices
$B=\{r,s,u_1,u_2, w\}$.  Define an $X_0$-net $\HH_0$ of a graph $G_0$
on $X_0\cup Y$ with $X_0=B\cup C\cup \{v_1,v_2\}$ as follows (see
Figure~\ref{figure:figure8}): Initially, let $G_0=G[C\cup Y\cup
  \{v_1,v_2\}]$ and let $\HH_0$ consist of the nodes and arcs of $\HH$
that intersect $C$.  For each $i\in\{1,2\}$, update $V_i$ by deleting
all vertices not in $C$ except for $v_i$ and then adding $w$.  Make a
new simple arc $V_1V_2$ consisting of $w$.  Add a minimum number of
edges to make $N_{G_0}(w)=(\{y_1\}\cup V_1\cup V_2) \setminus \{w\}$.
Make new nodes $R=\{r\}$, $U_1=\{u_1\}$, and $U_2=\{u_2\}$.  If $S$ is
a node, then let $S_0=S$; otherwise, make a new node $S_0$ via
$\textsc{subdivide}(S)$.  Add $s$ into $S_0$. Make a simple arc $RS_0$
consisting of $r$ and $s$.  For each $i\in\{1,2\}$, make a simple arc
$U_iV_i$ consisting of vertices $u_i$ and $v_i$.  Add a minimum number
of edges to make $N_{G_0}(s)=R\cup S$, $N_{G_0}(r)=\{s\}$,
$N_{G_0}(u_1)=\{v_1\}$, and $N_{G_0}(u_2)=\{v_2\}$.

$\HH_0$ is an $X_0$-net of $G_0$ with leaf nodes $R$, $U_1$, and $U_2$
and leaf vertices $r$, $u_1$, and $u_2$.  Since $\HH$ is an $X$-web of
$G$ and all new arcs of $\HH_0$ are simple, $\HH_0$ is an $X_0$-web
of~$G_0$.  Since each $V_i$ with $i\in\{1,2\}$ is the neighbor of
$U_i$ and $V_1V_2$ is an arc of $\HH_0$, $V_1$ is a maximal split node
of $\HH_0$.  Since $Y$ is $\HHdagger$-wild, $Y$ is
$\HH_0^\dagger$-wild.  Since $Y$ is $\HH$-nonsolid, $\{w\}$ is not a
pod of $Y$ in $\HH_0$ and $Y$ is $\HH_0$-nonsolid. Since node $S_0$ is
adjacent to leaf $R$ in $\HH_0$, no $V_1V_2$-chunk of $\HH_0$
intersects $S_0$, implying no pod of $Y$ in $\HH_0$ that is a superset
of $C$.  The minimality of $C$ implies no pod of $Y$ in $\HH_0$ that
is a proper subset of $C$.  Thus, $Y$ is $\HH_0$-unpodded.
Lemma~\ref{lemma:lemma3.2}\eqref{item1:lemma3.2} implies a sapling
$T_0$ of $G_0$.  $T_0-(B\cup \{v_1,v_2\})$ is a sprout.

Case~3: $S_Y\ne\varnothing$ and $S_C\ne\varnothing$.  $S_C$ is
$\HH$-solid.  Assume $S_Y=\{y_2\}$, since otherwise $G[P_1\cup Q]$ for
an $S_CV_2$-rung $Q$ of $G[C]$ not intersecting $V_1$ is a sprout of
Type~\ref{S2}.  Assume that any $N_2V_2$-rung $Q$ of $G[C]$ intersects
$S_C$ exactly at its end-vertex in $N_2$, since otherwise $G[P_1\cup
  Q]$ contains a sprout of Type~\ref{S1} or~\ref{S2}.  Thus, each
vertex $v\in C$ admits a $vV_2$-rung $Q(v)$ of $G[C]$ with $(V_1\cup
S_C)\cap V(Q(v))\subseteq \{v\}$: Assume for contradiction a $v\in C$
such that each $vV_2$-rung $Q(v)$ with $V_1\cap V(Q(v))\subseteq
\{v\}$ intersects $S_C\setminus \{v\}$.  If $S_C$ is a node $V$ of
$\HH$, then graph $\GG(\CC)-V$ is disconnected.  If $S_C$ is contained
by a simple arc $F$ of $\HH$, then graph $\HH[\CC] - \{V_1, V\}$ is
disconnected.  Either way, the minimality of $C$ implies that $N_2$
intersects the connected component of $G[C]-S_C$ that intersects
$V_2$, implying an $N_2(V_2\setminus S_C)$-rung of $G[C]$,
contradicting the above assumption.

Case~3(a): {\em a vertex $v\in N_2\setminus S_C$}.  $Q(v)$ does not
intersect $S_C$, so $G[Y\cup Q(v)]$ is a sprout of Type~\ref{S1}.
Case~3(b): {\em a vertex $v\in S_C\setminus N_2$}.  $Q(v)$ does not
intersect $N_2$ or else the $N_2V_2$-rung of $Q(v)$ does not intersect
$S_C$ at its end-vertex in $N_2$, contradiction.  Thus, $G[Y\cup
  Q(v)]$ is a sprout of Type~\ref{S2}.
Case~3(c): {\em $N_2=S_C$}.  If $M_1\ne V_1$, then $G[P_1\cup Q(v_1)]$
for a $v_1\in V_1\setminus M_1$ contains a sprout of Type~\ref{S3}.
If $M_1=V_1$, then $M$ contains a $v_1$, since $Y$ is
$\HH$-nonsolid. We have $N=M_1\cup N_2$.  $M_1$ and $N_2$ are both
$\HH$-solid. Thus, $G[Y\cup Q(v_1)]$ contains a sprout of
Type~\ref{S1}.
\label{endofcharacterization}
\end{proof} 

This completes the proof of our characterization in
Theorem~\ref{theorem:theorem3.1} as well as Chudnovsky and Seymour's
characterization in Theorem~\ref{theorem:theorem2.2}.
Subroutine~\ref{subroutine:B} can be implemented to run in $O(m)$
time, so Steps~\ref{A2c} and~\ref{A2d} take $O(m)$ time.
Steps~\ref{A1},~\ref{A2a}, and~\ref{A2b} take $O(m)$ time.  Since the
set of vertices of $G$ in $\HH$ is enlarged by Step~\ref{A2d} and not
affected elsewhere, Step~\ref{A2} halts in $O(n)$ iterations.  Thus,
Algorithm~\ref{algorithm:A} can be implemented to run in $O(mn)$ time.
To complete proving Theorem~\ref{theorem:theorem1.1}, it remains to
implement Algorithm~\ref{algorithm:A} to run in $O(m\log^2 n)$ time
in~\S\ref{section:section5} using dynamic graph algorithms and other
data structures.

\section[]{Proving Lemma~\ref{lemma:lemma3.3}}
\label{section:section5}

\begin{figure}[t]
\centerline{\scalebox{0.47}{\includegraphics{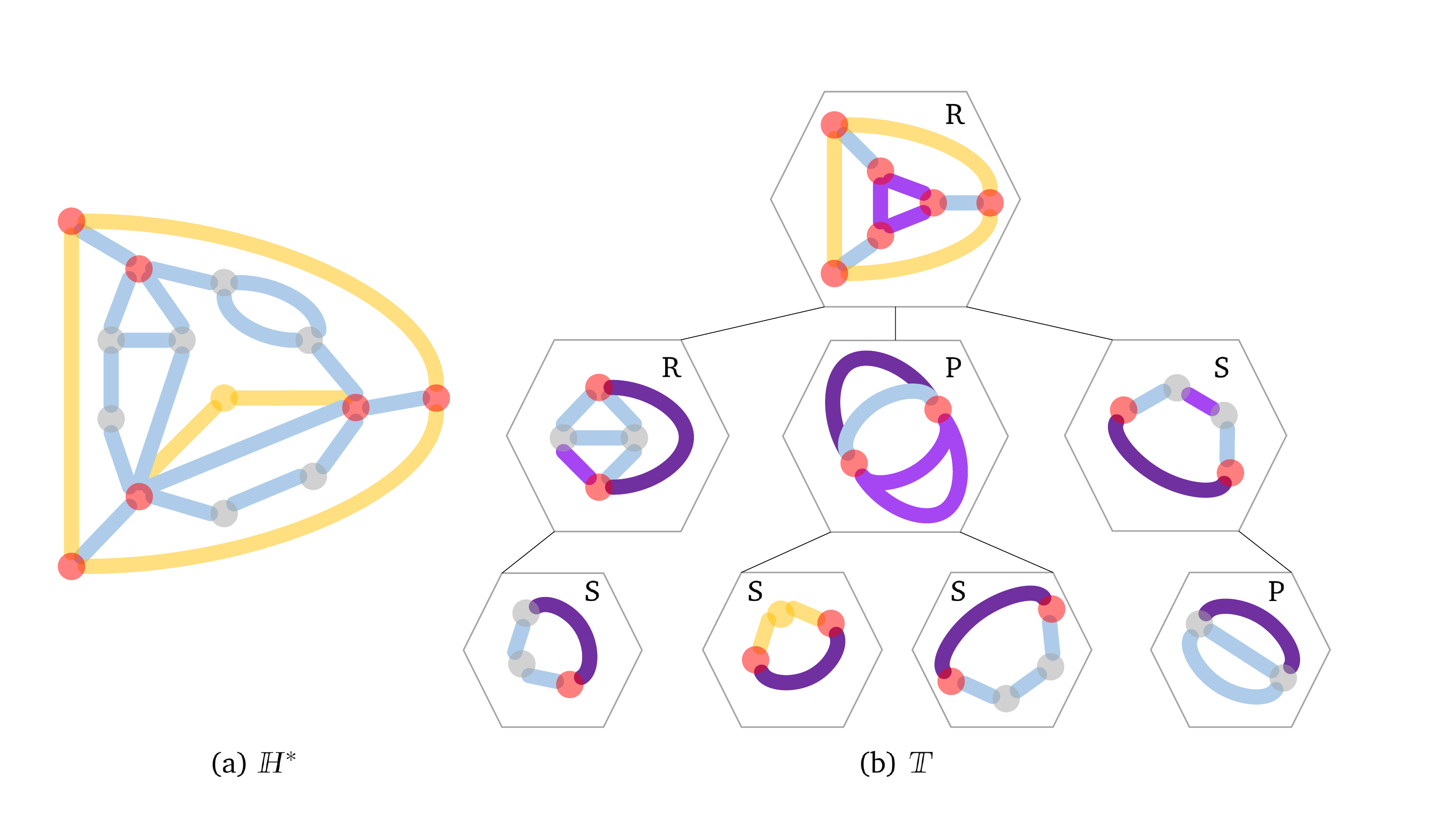}}}
\caption{An example of $\HHstar$ and $\TT$.  The Q-knots are omitted for
  brevity.  The virtual arc in dark purple in a nonroot knot $K$
  matches a light purple arc in the parent of $K$ in $\TT$. They form
  the pair of virtual arcs between the poles of $K$.  Each non-purple
  arc in a knot $K$ is a virtual arc whose corresponding arc of
  $\HHstar$ is contained by a child Q-knot of $K$. A non-purple arc is
  in yellow if and only if its corresponding arc of $\HHstar$ is
  dummy. The dummy nodes of $\HHstar$ are in yellow.  $\HH$ is the
  multigraph obtained from $\HHstar$ by deleting the yellow nodes and
  arcs.  $\HHdagger$ is the simple graph obtained from the one in
  the root of $\TT$ by deleting the yellow arcs.  The maximal split
  nodes of $\HH$, i.e., the nodes of $\HHdagger$ are in red.}
\label{figure:figure9}
\end{figure}

Let $G$ be represented by a static adjacency list.  We use a dynamic
adjacency list to represent an incremental biconnected multigraph
$\HHstar$ with $V(\HHstar)=V(\HH)$ that is a supergraph of $\nabla(\HH)$.
An arc or node of $\HHstar$ is {\em dummy} if it is an empty vertex set
of $G$.  For instance, the three arcs of $\nabla(\HH)$ between the
leaves of $\HH$ are dummy in $\HHstar$.  Other dummy nodes and arcs are
created only via operation $\textsc{merge}$.  The $X$-web $\HH$
maintained by Algorithm~\ref{algorithm:A} is exactly $\HHstar$ excluding
its dummy arcs and nodes. See Figure~\ref{figure:figure9}(a) for an
example of $\HHstar$.  Each node and arc of $\HH$ and $\HHdagger$ is
associated with a distinct {\em color} that is a positive integer such
that two vertices share a common {\em arc color} (respectively, {\em
  node color}) for $\HH$ and $\HHdagger$ if and only if they are
contained by a common arc (respectively, node) of $\HH$ and
$\HHdagger$.  For each vertex $v$ of $G$, we maintain a set of at
most six colors indicating the arc, maximal chunk, nodes, and maximal
split nodes of $\HH$ that contain $v$, which are called the {\em
  $\HH$-arc}, {\em $\HHdagger$-arc}, {\em $\HH$-node}, and {\em
  $\HHdagger$-node colors} of vertex $v$.  For each color $c$, we
store its corresponding arc or node for $\HH$ or $\HHdagger$ and
maintain the number of the vertices having the color $c$ without
keeping an explicit list of these vertices.  For each node $V$ and
each incident arc $E$ of $V$ in $\HH$, we maintain the cardinality of
the vertex set $E\cap V$.  Thus, it takes $O(1)$ time to (1) update
and query the colors of a vertex and (2) add a vertex to an arc or
node of $\HH$.  For each arc of $\HHstar$, we mark whether it is dummy,
simple, or flexible and, for each simple arc $E=V_1V_2$ of $\HHstar$, we
use a doubly linked list to store the $V_1V_2$-rung $G[E]$.  For any
vertex $v$ and vertex set $Y$ of $G$, let $d(v)=|N(v)|$ and
$d(Y)=\sum_{y\in Y}d(y)$ throughout the section.

Based on Lemma~\ref{lemma:lemma5.1}, to be proved
in~\S\ref{subsection:subsection5.4}, Steps~\ref{A2a} and~\ref{A2b} are
implemented in \S\ref{subsection:subsection5.1} to run in overall
$O(m\log^2 n)$ time throughout
Algorithm~\ref{algorithm:A}. Step~\ref{A2c} is implemented in
\S\ref{subsection:subsection5.2} to run in overall $O(m)$ time
throughout Algorithm~\ref{algorithm:A}. Step~\ref{A2d}, i.e.,
Subroutine~\ref{subroutine:B} is implemented in
\S\ref{subsection:subsection5.3} to run in overall $O(m\log
n\cdot\alpha(n,n))$ time throughout Algorithm~\ref{algorithm:A}, where
$\alpha(n,n)$ is the inverse Ackermann function.

\subsection{Steps~\ref{A2a} and~\ref{A2b} of Algorithm~\ref{algorithm:A}}
\label{subsection:subsection5.1}

Although vertex colors change only in Step~\ref{A2d}, the overall
number of changes of the $\HHdagger$-arc and $\HHdagger$-node
colors affects the analysis of our implementation of Steps~\ref{A2a}
and~\ref{A2b}.  Therefore, this subsection analyzes the time for the
change of $\HHdagger$-arc and $\HHdagger$-node colors.  The time
for the change of $\HH$-arc and $\HH$-node colors will be analyzed for
Step~\ref{A2d} in~\S\ref{subsection:subsection5.3}.  A vertex of $G$
stays uncolored until it is added into $X$.  Each vertex of $X$ has
exactly one $\HHdagger$-arc color and at most two $\HHdagger$-node
colors.  Each node $V$ of $\HHdagger$ stays a node of $\HHdagger$
and each vertex in $V$ stays in $V$ for the rest the algorithm.  Thus,
the $\HHdagger$-node colors of each vertex are updated $O(1)$ times
throughout the algorithm, implying that the overall time for updating
$\HHdagger$-node colors of all vertices is $O(n)$.  Although the
$\HHdagger$-arc color of a vertex may change many times, the overall
time for updating the $\HHdagger$-node colors of all vertices can be
bounded by $O(n\log n)$.  Observe that $\HH$ is updated by
Subroutine~\ref{subroutine:B} only via (1) subdividing a simple arc of
$\HH$, (2) merging an $\HH$-podded $Y$ into a minimal pod of $Y$ in
$\HH$, and (3) creating an arc $E=Y$ for an $\HH$-solid $Y$.  If the
simple graph $\HHdagger$ does not change, then each of these updates
takes $O(d(Y))$ time.  If the simple graph $\HHdagger$ changes, then
$Y$ is $\HH$-solid.  For instance, let $\HH$ be as in
Figure~\ref{figure:figure5}(a), implying that $\HHdagger$ is as in
Figure~\ref{figure:figure5}(b). If an $\HH$-solid $Y$ joins $\HH$ as
the arc $E_{16}$ in Figure~\ref{figure:figure5}(c), then all nodes and
arcs of $\HH$ become nodes and arcs of $\HHdagger$.  However, once
two vertices of $X$ have distinct $\HHdagger$-arc colors, they can
no longer share a common arc color for $\HHdagger$ for the rest of
the algorithm.  Thus, one can bound the overall number of changes of
$\HHdagger$-arc colors of all vertices by $O(n\log n)$ as follows:
If $E$ is an arc of the original $\HHdagger$ and $E_1,\ldots,E_k$
are the arcs of the updated $\HHdagger$ with $E_1\cup\cdots\cup
E_k\subseteq E$ and $|E_1|\leq \cdots \leq |E_k|$, then let the
vertices in $E_k$ keep their original $\HHdagger$-arc color and
assign a distinct new $\HHdagger$-arc color to the vertices in each
$E_i$ with $i\in\{1,\ldots,k-1\}$.  Since the cardinality of the arc
of $\HHdagger$ containing a specific vertex of $X$ is halved each
time its $\HHdagger$-arc color changes, its $\HHdagger$-arc color
changes $O(\log n)$ times, implying that the $\HHdagger$-arc colors
of all vertices change $O(n\log n)$ times throughout the algorithm.
With the data structure of Lemma~\ref{lemma:lemma5.1}, to be proved
in~\S\ref{subsection:subsection5.4}, the overall time for
Steps~\ref{A2a} and~\ref{A2b} throughout the algorithm is $O(m\log^2
n)$.

\begin{lem}
\label{lemma:lemma5.1}
If $X$ is an incremental subset of $V(G)$ such that each $x\in X$ has
exactly one $\HHdagger$-arc color $a$ and a set of at most two
$\HHdagger$-node colors corresponding to a subset of the two
end-vertices of $a$, then there is an $O(m+n)$-time obtainable data
structure supporting the following queries and updates:
\begin{enumerate}[leftmargin=*]
\item 
\label{item1:lemma5.1}
Move a vertex $v$ of $G-X$ to $X$ in amortized $O(d(v)\cdot \log^2 n)$
time.

\item 
\label{item2:lemma5.1}
Update the colors of a vertex $v\in X$ in amortized $O(d(v)\cdot\log
n)$ time.

\item
\label{item3:lemma5.1}
Determine if there is a set $Y\subseteq V(G-X)$ with connected $G[Y]$
such that two vertices of $N(Y,X)$ share no color and, for the
positive case, report a minimal such $Y$ in amortized $O(d(Y)\cdot
\log^2 n)$ time.
\end{enumerate}
\end{lem}

\subsection{Step~\ref{A2c} of Algorithm~\ref{algorithm:A}}
\label{subsection:subsection5.2}
Let $\SS$ be the $O(d(Y))$-time obtainable set consisting of the nodes
$V$ of $\HH$ with $V\subseteq N(Y,X)$ and the simple arcs $E$ of $\HH$
with $G[E\cap N(Y,X)]$ being an edge.  $Y$ is $\HH$-solid if and only
if $|\SS|=2$, $N(y,X)=\varnothing$ for each internal node $y$ of path
$G[Y]$, and $N(Y,X)$ is contained by the union of the nodes or arcs in
$\SS$.  Therefore, it takes $O(d(Y))$ time to determine whether $Y$ is
$\HH$-solid.  Lemma~\ref{lemma:lemma4.1}\eqref{item1:lemma4.1} implies
that $Y$ is $\HH$-podded if and only if both of the following
conditions hold: (a) $N(Y,X)$ is contained by the union of an arc $E$
of $\HHdagger$ and its end-nodes $V_1$ and $V_2$ in $\HHdagger$
and (b) $E$ is a pod of $Y$ in $\HHdagger$.  Both conditions can be
checked in $O(d(Y))$ time via the $\HHdagger$-arc and
$\HHdagger$-node colors of each vertex in $N(Y,X)$ and the
cardinalities of $V_1\setminus E$ and $V_2\setminus E$.  Therefore, it
takes $O(d(Y))$ time to determine whether $Y$ is $\HH$-podded.  Since
the $\HHdagger$-wild sets $Y$ in all iterations of the algorithm are
pairwise disjoint, it takes overall $O(m)$ time for Step~\ref{A2c} to
determine whether $Y$ is $\HH$-sticky throughout the algorithm.

\subsection[]{Step~\ref{A2d} of Algorithm~\ref{algorithm:A}, i.e., Subroutine~\ref{subroutine:B}}
\label{subsection:subsection5.3}

This subsection shows how to implement Subroutine~\ref{subroutine:B}
so that the overall time of Step~\ref{A2d} throughout
Algorithm~\ref{algorithm:A} is $O(m\log n\cdot\alpha(n,n))$.  Although
we may delete nodes and arcs from $\HH$ via $\textsc{merge}(C)$ for a
minimal pod $C$ of $Y$ in $\HH$, they stay as dummy nodes and arcs in
$\HHstar$ in order to make the multigraph $\HHstar$ incremental.  One can
verify that $\HHdagger$ aids $\HHstar$, even though $\HHstar$ is not an
$X$-net due to its dummy arcs and nodes.  Although Step~\ref{B12} may
change $\HHdagger$, the overall time for updating the
$\HHdagger$-colors has been accounted for
in~\S\ref{subsection:subsection5.1}.  Therefore, this subsection only
analyzes the time required by the change of $\HH$-arc and $\HH$-node
colors and the cardinalities of $E\cap V_1$ and $E\cap V_2$ for each
arc $E=V_1V_2$ of $\HH$.

\begin{figure}[t]
\centerline{\scalebox{0.47}{\includegraphics{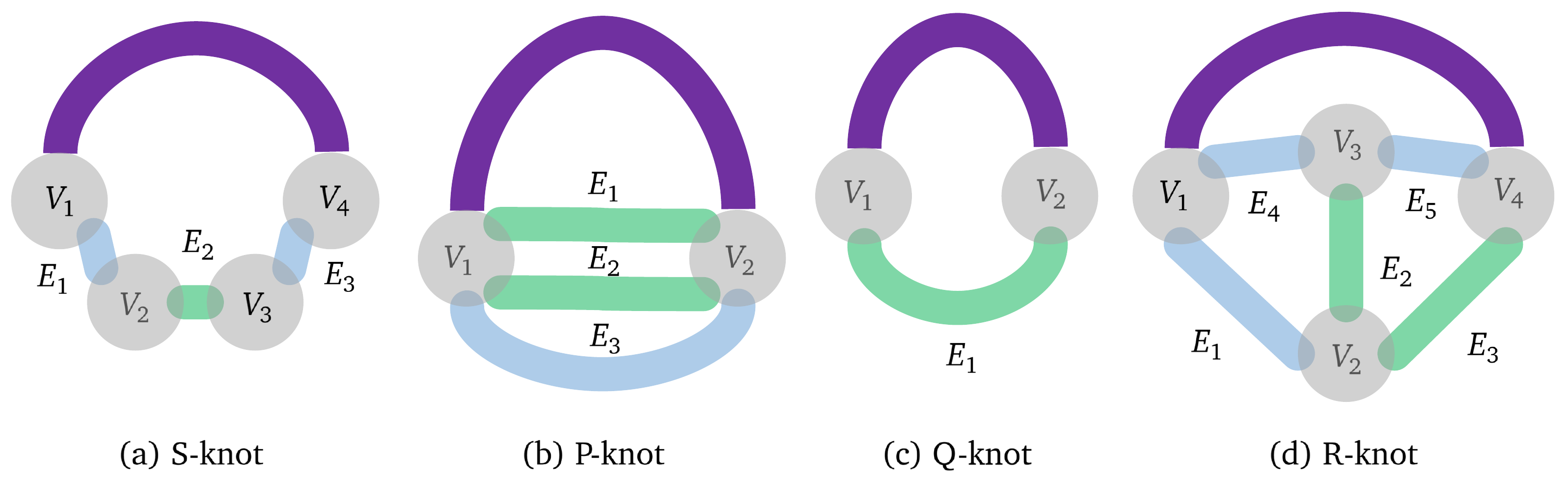}}}
\caption{Four examples of the lowest common ancestor $K$ of the
  Q-knots containing the arcs of $\HH$ in $\CC_1\cup \CC_2$, which
  equals $E_2$ in (a), $E_1\cup E_2$ in (b), $E_1$ in (c), and
  $E_2\cup E_3$ in (d).}
\label{figure:figure10}
\end{figure}

The {\em SPQR-tree} $\TT$ of the incremental multigraph $\HHstar$ is an
$O(n)$-time obtainable $O(n)$-space tree structure representing the
triconnected components of $\HHstar$~\cite{SPQR,Gutwenger}.  Each member
of $V(\TT)$, which we call a {\em knot}, is a graph homeomorphic to a
subgraph of $\HHstar$~\cite[Lemma~3]{SPQR} such that the knots induce a
disjoint partition of the arcs of $\HHstar$.  Specifically, there is a
supergraph $\GG$ of $\HHstar$ with $V(\GG)=V(\HHstar)$, where each arc of
$\GG\setminus \HHstar$ is called {\em virtual}~\cite{Tutte66}, and there
are four types of knots of $\TT$:
(1) {\em S-knot}: a simple cycle on three or more nodes.  
(2) {\em P-knot}: three or more parallel arcs.
(3) {\em Q-knot}: two parallel arcs, exactly one of which is virtual.
(4) {\em R-knot}: a triconnected simple graph that is not a cycle. 
The Q-knots are the leaves of $\TT$ and each arc of $\HHstar$ is
contained by a Q-knot.  No two S-knots (respectively, P-knots) are
adjacent in $\TT$.  Each virtual arc is contained by exactly two
adjacent knots.  Since $\HH$ has $O(n)$ arcs by Condition~\ref{N2},
$\TT$ has $O(n)$ knots.  If $U$ and $V$ are nonleaf nodes of $\HH$
such that $UV$ is a virtual arc, then $\{U,V\}$ is a split pair of
$\HH$.  If distinct nodes $U$ and $V$ admit three internally disjoint
$UV$-paths in $\HHstar$, then $U$ and $V$ are contained by a common
P-knot or R-knot of $\TT$~\cite{SPQR}.  By Condition~\ref{N1} of
$\HH$, there are three internally disjoint paths in $\nabla(\HH)$
between each pair of leaves of $\HHstar$, implying an R-knot of $\TT$
containing the leaves of $\HH$.  Let $\TT$ be rooted at this unique
R-knot.  Figure~\ref{figure:figure9}(b) is the $\TT$ for the $\HHstar$
in Figure~\ref{figure:figure9}(a).  Let $K$ be a nonroot knot of
$\TT$.  The {\em poles}~\cite{Gutwenger} of $K$ are the end-nodes of
the unique virtual arc contained by $K$ and its parent knot in $\TT$.
For the four nonroot knots $K$ in Figure~\ref{figure:figure10}, $V_1$
and $V_4$ (respectively, $V_2$) are the poles of the knots in (a) and
(d) (respectively, (b) and (c)).  Let $\CC(K)$ consist of the arcs of
$\HH$ in the descendant Q-knots of $K$ in $\TT$.  Let $C(K)$ consist
of the vertices of $G$ contained by the arcs of $\CC(K)$.  If $U$ and
$V$ are the poles of a nonroot knot $K$ of $\TT$, then $C(K)$ is a
$UV$-chunk and $\CC(K)$ is the arc set for $C(K)$.  A nonempty vertex
set $C$ is a maximal chunk of $\HH$ if and only if $C=C(K)$ holds for
a child knot $K$ of the root of $\TT$.  For instance, the $X$-net
$\HH$ in Figure~\ref{figure:figure9}(a) has six maximal chunks.  One
of them is $C(K)$ for the child R-knot (respectively, P-knot and
S-knot) $K$ of the root of $\TT$.  The remaining three are $C(K)$ for
three omitted child Q-knots $K$ of the root of $\TT$.  For any nonroot
knot $K$ of $\TT$ with $C(K)\ne\varnothing$, if $K$ is a P-knot, then
$\CC(K)$ is the union of the arc sets of all split components of
$\{U,V\}$ (e.g., three splits components of $\{V_1,V_2\}$ in the
example in Figure~\ref{figure:figure10}(b)); otherwise, $\CC(K)$ is
the arc set of a single split component of $\{U,V\}$, where $U$ and
$V$ are the poles of $K$ (e.g., exactly one split component for
$\{V_1,V_4\}$ in the examples in Figures~\ref{figure:figure10}(a)
and~\ref{figure:figure10}(d) and exactly one split component for
$\{V_1,V_2\}$ in the example in Figure~\ref{figure:figure10}(c)).

\begin{lem}[Di~Battista and Tamassia~\cite{SPQR}]
\label{lemma:lemma5.2}
Each update to $\TT$ corresponding to the following operations on the
incremental biconnected multigraph $\HHstar$ can be implemented to run
in amortized $\alpha(n,n)$ time: 
\begin{enumerate*}[label=(\arabic*), ref=\arabic*]
\item 
\label{item1:lemma5.2}
Add a new node $V$ to subdivide
an arc $V_1V_2$ of $\HHstar$ into two arcs $E_1=VV_1$ and $E_2=VV_2$.
\item
\label{item2:lemma5.2}
Add an arc $UV$ between two nodes $U$ and $V$ of $\HH$.
\end{enumerate*}
\end{lem}

We first show that, given a vertex set $S$ contained by a simple arc
$E=V_1V_2$ such that $G[S]$ is an edge,
Operation~$\textsc{subdivide}(S)$ in Steps~\ref{B11} and~\ref{B21} can
be implemented to run in amortized $O(\log n)$ time: Let each $P_i$
with $i\in\{1,2\}$ be the $V_iS$-rung of $G[E]$.  Let $j$ be an index
in $\{1,2\}$ with $|V(P_j)|\leq |V(P_{3-j})|$.  Using the doubly
linked list for the $V_1V_2$-rung $G[E]$, it takes $O(|V(P_j)|)$ time
to (1) create a new node $V=S$ with a new $\HH$-node color assigned to
both vertices in $S$, (2) create a new simple arc $E_j=VV_j$
consisting of the vertices of $P_j$, (3) assign a new $\HH$-arc color
for each vertex in $E_j$, (4) let arc $E_{3-j}$ take over the
$\HH$-arc color of $E$, and (5) obtain the doubly linked lists of
$G[E_1]$ and $G[E_2]$ from that of $G[E]$.  Each time a vertex $x$ is
recolored this way, the cardinality of the simple arc of $\HH$
containing $x$ is halved.  Therefore, the overall time for Operation
$\textsc{subdivide}(S)$ in Steps~\ref{B11} and~\ref{B21} is $O(n\log
n)$.

Step~\ref{B1}: By the above analysis for $\textsc{subdivide}$,
Step~\ref{B11} runs in amortized $O(\log n)$ time.  As for
Steps~\ref{B12} and~\ref{B13}, a new $\HH$-arc color is created for
the new arc of $\HH$.  The $\HH$-arc and $\HH$-node colors of the
vertices in $Y$ and the cardinality of each vertex set that is a node,
arc, or the intersection of a node and its incident arc can be updated
in $O(d(Y))$ time.  By Lemma~\ref{lemma:lemma5.2} and the fact that
Subroutine~\ref{subroutine:B} is executed $O(n)$ times, the overall
time for Step~\ref{B1} is $O(m\log n)$.

Step~\ref{B2}: We first assume that we are given a set $\CC$ of arcs
of $\HH$ whose union is a minimal pod $C$ of $Y$ in $\HH$ and show how
to implement Steps~\ref{B21},~\ref{B22}, and~\ref{B23} to run in
overall $O(m\log n)$ time throughout Algorithm~\ref{algorithm:A}.  Let
$C$ be a $V_1V_2$-chunk of $\HH$.

Step~\ref{B21}: It takes $O(|\CC|)$ time to determine whether $V_2$ is
incident to exactly one arc $F=VV_2$ in $\CC$ and $F$ is simple.  We
start from $V$ to traverse the $VV_2$-rung $G[F]$ to obtain the node
$v_2\in N(Y,F)$ that is closest to $V_2$ in $G[F]$. The required time
is linear in the number of traversed edges plus $d(Y)$. Observe that
Step~\ref{B21} in any remaining iteration of
Algorithm~\ref{algorithm:A} does not traverse these edges
again. Moreover, the sum of $|\CC|$ over all iterations of
Algorithm~\ref{algorithm:A} is $O(n)$.  Thus, the overall time of
Step~\ref{B21} including that of calling
$\textsc{subdivide}(\{v,v_2\})$ is $O(m\log n)$.

Step~\ref{B22}: Let $E_1,\ldots,E_k$ with $|E_1|\leq \cdots\leq |E_k|$
be the arcs of $\HH$ in $\CC$.  We show how to implement Operation
$\textsc{merge}(C)$ in Step~\ref{B22} to run in amortized $O(\log n)$
time: We create a new arc $E=V_1V_2$ in $\HHstar$ consisting of all
vertices in $C$ and mark the original arcs $E_1,\ldots,E_k$ of $\HHstar$
intersecting $C$ dummy so that $\HHstar$ is incremental as required by
Lemma~\ref{lemma:lemma5.2}.  The nodes of $\HH$ whose incident arcs
are all dummy are also marked dummy.  The cardinalities of $E$, $V_1$,
$V_2$, $E\cap V_1$, and $E\cap V_2$ can be obtained in $O(k)$ time.
Since we do not keep an explicit list of the vertices in $C$, we
simply let all vertices in $C$ adopt the $\HH$-color of the vertices
in $E_k$.  Each time a vertex $v$ is recolored this way, the
cardinality of the arc of $\HH$ containing $v$ is doubled.  Observe
that once a vertex in $X$ loses its $\HH$-node colors, it stays
without any $\HH$-node color for the rest of the algorithm.  Combining
with Lemma~\ref{lemma:lemma5.2}\eqref{item2:lemma5.2}, Step~\ref{B22}
takes overall $O(n\log n)$ time throughout
Algorithm~\ref{algorithm:A}.

Step~\ref{B23}: The $\HH$-arc and $\HH$-node colors of the vertices of
$Y$ and the cardinalities of $E\cap V_1$ and $E\cap V_2$ can be
updated in $O(d(Y))$ time.

\begin{lem}[{Alstrup, Holm, Lichtenberg, and Thorup~\cite[\S3.3]{Alstrup:2005}}]
\label{lemma:lemma5.3}
For any dynamic rooted $n$-knot tree, there is an $O(n)$-time
obtainable data structure supporting the following operations and
queries on $\TT$ in amortized $O(\log n)$ time for any given distinct
knots $K_1$ and $K_2$ of $\TT$:
\begin{enumerate}[leftmargin=*]
\item 
\label{item1:lemma5.3}
If $K_2$ is not a descendant of $K_1$, then make the subtree rooted at
$K_1$ a subtree of $K_2$ such that $K_2$ becomes the parent of $K_1$.

\item 
\label{item2:lemma5.3}
Obtain the lowest common ancestor of $K_1$ and $K_2$.

\item 
\label{item3:lemma5.3}
If $K_2$ is a descendant of $K_1$, then obtain the child knot of $K_1$
that is an ancestor of $K_2$ in $\TT$.
\end{enumerate}
\end{lem}

It remains to show that it takes overall $O(m\log n\cdot \alpha(n,n))$
time to obtain the arc set $\CC$ of a minimal pod $C$ of an
$\HH$-podded $Y$ in all iterations of Algorithm~\ref{algorithm:A}.  We
additionally construct a data structure for $\TT$ ensured by
Lemma~\ref{lemma:lemma5.3}.  By Lemmas~\ref{lemma:lemma5.2}
and~\ref{lemma:lemma5.3}\eqref{item1:lemma5.3}, the overall time for
updating the data structure reflecting the updates to $\TT$ throughout
algorithm~\ref{algorithm:A} is $O(n\log n\cdot \alpha(n,n))$.  Let
$C^*=W_1W_2$ be the arc of $\HHdagger$ with $V_1=W_1\subseteq
N(Y,W_1)\cup C^*$.  By Conditions~\ref{condition:P}, $C$ has to
contain all arcs $E$ of $\HH$ with (1) $(E\setminus V_1)\cap
N(Y,X)\ne\varnothing$ or (2) $(E\cap V_1)\setminus
N(Y,X)\ne\varnothing$.  Let $\CC_1$ and $\CC_2$ consist of the arcs of
Types~(1) and~(2), respectively.  It takes $O(d(Y))$ time to obtain
$\CC_1$ and the incident arcs of $V_1$ that are not of Type~(1)
or~(2). It then takes $O(|\CC_2|)$ time to obtain $\CC_2$.  By
Lemma~\ref{lemma:lemma5.3}\eqref{item2:lemma5.3}, it takes
$O(|\CC_1\cup \CC_2|\cdot\log n)$ time to obtain the lowest knot $K$
of $\TT$ with $\CC_1\cup\CC_2\subseteq \CC(K)$.  Since all arcs in
$\CC_1\cup \CC_2$ are merged into a single arc of $\HH$ via
$\textsc{merge}(C)$ at the end of the current iteration, the overall
time for obtaining $K$ throughout Algorithm~\ref{algorithm:A} is
$O(m\log n\cdot\alpha(n,n))$. It remains to show that $\CC$ can be
obtained from $K$ in overall $O(m\log n\cdot \alpha(n,n))$ time
throughout Algorithm~\ref{algorithm:A}.

Case~1: $K$ is an S-knot.  Let $V_1V_2\cdots V_\ell V_1$ with
$\ell\geq 3$ be the cycle of $K$ such that $V_1$ and $V_\ell$ are the
poles of $K$.  For each $i\in\{1,\ldots,\ell-1\}$, let $K_i$ be the
child knot of $K$ with poles $V_i$ and $V_{i+1}$, $\CC_i=\CC(K_1)\cup
\cdots \cup \CC(K_i)$, and let $C_i$ be the union of the arcs in
$\CC_i$.  Let $j$ be the smallest index in $\{2,\ldots,\ell-1\}$ with
$\CC_1\cup \CC_2\subseteq \CC_j$.  If $N(Y,X)\setminus (V_1\cup
C_{j-1})=V_j\setminus C_{j-1}$, then $\CC=\CC_{j-1}$; otherwise,
$\CC=\CC_j$.  For the example in Figure~\ref{figure:figure10}(a), if
$N(X,Y)\setminus (V_1\cup E_1)=V_2\setminus E_1$, then $E_1$ is a
minimal pod of $Y$ in $\HH$; otherwise, $E_1\cup E_2$ is a minimal pod
of $Y$ in $\HH$.  By Lemma~\ref{lemma:lemma5.3}\eqref{item3:lemma5.3},
the time required to obtain the index $j$ and determine whether
$\CC=\CC_{j-1}$ or $\CC=\CC_j$ is dominated by the time of obtaining
$K$ plus the time of $\textsc{merge}(C)$.

Case~2: $K$ is a P-knot.  $\CC$ equals the union of $\CC(K')$ over all
child knots $K'$ of $K$ in $\TT$ with $(\CC_1\cup \CC_2)\cap
\CC(K')\ne\varnothing$. For the example in
Figure~\ref{figure:figure10}(b), $E_1\cup E_2$ is a minimal pod of $Y$
in $\CC$.  By Lemma~\ref{lemma:lemma5.3}\eqref{item3:lemma5.3}, the
time needed to obtain $\CC$ is dominated by that of obtaining $K$.

Case~3: $K$ is a Q-knot.  As illustrated by
Figure~\ref{figure:figure10}(c), $\CC=\CC(K)$ can be obtained in
$O(1)$ time.

Case~4: $K$ is an R-knot.  If there is child knot $K'$ of $K$ in $\TT$
with poles $V_1$ and $V_2$ such that all arcs of $K$ intersecting
$\CC_1\cup\CC_2$ are incident to $V_2$ and $N(Y,X)\setminus (V_1\cup
C(K'))=V_2\setminus C(K')$, then $\CC=\CC(K')$; otherwise,
$\CC=\CC(K)$.  For the example in Figure~\ref{figure:figure10}(d), if
$N(Y,X)\setminus (V_1\cup E_1)=V_2\setminus E_1$, then $E_1$ is a
minimal pod of $Y$ in $\HH$; otherwise, $E_1\cup\cdots\cup E_5$ is a
minimal pod of $Y$ in $\HH$.  By
Lemma~\ref{lemma:lemma5.3}\eqref{item3:lemma5.3}, the time required to
identify all possible vertices $V_2$, which can be at most two, is
dominated by the time of identifying $K$.  If there are no possible
$V_2$, then we have $\CC=\CC(K)$.  Otherwise, for each of the at most
two vertices $V_2$, we spend $O(d(Y))$ time to determine whether the
child knot $K'$ with poles $V_1$ and $V_2$ satisfies $N(Y,X)\setminus
(V_1\cup C(K'))=V_2\setminus C(K')$. For the positive (respectively,
negative) case, we have $\CC=\CC(K')$ (respectively, $\CC=\CC(K)$).

Therefore, the overall time for obtaining the arc set of a minimal pod
of $Y$ in $\HH$ is $O(m\log n\cdot\alpha(n,n))$.  To complete our
proof of Lemma~\ref{lemma:lemma3.3}, it remains to prove
Lemma~\ref{lemma:lemma5.1} in~\S\ref{subsection:subsection5.4}.

\subsection[]{Proving Lemma~\ref{lemma:lemma5.1}}
\label{subsection:subsection5.4}
\newcommand{\vol}{d}
\newcommand{\Oo}{O} 
\newcommand{\boundary}{\partial}
\newcommand{\cF}{{\mathcal F}}
\newcommand{\cT}{{\mathcal T}}

The subsection omits $\HHdagger$ from the terms $\HHdagger$-wild,
$\HHdagger$-tamed, $\HHdagger$-untamed, and $\HHdagger$-node and
$\HHdagger$-arc colors.  Recall that each vertex $x$ of $X$ is
associated with exactly one arc color and at most two node colors from
which we know which arc $E$ of $\HHdagger$ contains $x$ and whether
$x\in E\cap V$ holds for each end-node $V$ of $E$.  For any nonempty
$S\subseteq X$, we say that an $R\subseteq S$ {\em represents} $S$ and
call $R$ a {\em representative set} of $S$ if $|R|\leq 3$ and, for any
$V\subseteq X$, $R\cup V$ is tamed if and only if $S\cup V$ is tamed.
If $S$ is untamed, then each untamed two-vertex subset of $S$
represents $S$.  If $R_1$ represents $S_1$, $R_2$ represents $S_2$,
and $R$ represents $R_1\cup R_2$, then $R$ represents $S_1\cup S_2$.

\begin{lem}
\label{lemma:lemma5.5}
Any nonempty $S\subseteq X$ admits a representative set obtainable
from the colors of the vertices of $S$ in $O(|S|)$ time.
\end{lem}

\begin{proof}
Let $E_1,\ldots,E_\ell$ be the arcs of $\HHdagger$ intersecting $S$.
If $\ell=1$, then $S$ is tamed. Let $V_1$ and $V_2$ be the end-nodes
of $E_1$.  Choose an arbitrary vertex from each of the sets $S\cap
V_1$, $S\cap V_2$, and $S\setminus (V_1\cup V_2)$ that are nonempty to
form a representative set of $S$. The rest of the proof assumes
$\ell\geq 2$.  It takes $O(|S|)$ time to either (1) identify distinct
$i$ and $j$ in $\{1,\ldots,\ell\}$ such that $E_i$ and $E_j$ do not
share a common end-node or (2) ensure that $E_i$ and $E_j$ for any
distinct $i$ and $j$ in $\{1,\ldots,\ell\}$ share a common end-node.
Case~1 implies that $S$ is untamed and any two-vertex subset of $S$
intersecting both $E_i$ and $E_j$ represents $S$.

Case~2(a): $E_1,\ldots,E_\ell$ have a common end-node $V$.  If
$S\nsubseteq V$, then $S$ is untamed and any $\{u,v\}\subseteq S$ with
$u\notin V$ intersecting distinct arcs represents $S$.  If $S\subseteq
V$, then $S$ is tamed.  If $\ell=2$, then any two-vertex subset of $S$
intersecting both of $E_1$ and $E_2$ represents $S$.  If $\ell\geq 3$,
then any three-vertex subset of $S$ intersecting all of $E_1$, $E_2$,
and $E_3$ represents $S$.

Case~2(b): $E_1,\ldots,E_\ell$ have no common end-node.  Therefore,
$\ell=3$ and $E_1$, $E_2$, and $E_3$ form a triangle.  For indices
$i,j,k$ with $\{i,j,k\}=\{1,2,3\}$, let $V_i$ and $V_j$ be the
end-nodes of $E_k$.  If $S\subseteq \Delta(V_1,V_2,V_3)$, then $S$ is
tamed and any three-vertex subset of $S$ intersecting all of $E_1$,
$E_2$, and $E_3$ represents $S$.  If $S\nsubseteq
\Delta(V_1,V_2,V_3)$, then $S$ is untamed and $\{u,v\}$ with $u\in
(S\cap E_i)\setminus V_j$ and $v\in S\cap E_k$ for
$\{i,j,k\}=\{1,2,3\}$ represents $S$.
\end{proof}

For each $y\in V(G-X)$, we maintain a balanced binary search tree
$T_y$ on $N(y,X)$.  For each vertex $x$ of $T_y$, we maintain a
representative set $R_y(x)$ of the vertices in the subtree of $T_y$
rooted at $x$.  Thus, $R_y=R_y(\text{root}(T_y))$ represents $N(y,X)$.
We also maintain a doubly linked list $D_1$ for the vertices $y\in
V(G-X)$ with untamed $N(y,X)$.  When a vertex joins $N(y,X)$ or a
vertex in $N(y,X)$ changes color, $R_y$ and $D_1$ can be updated in
$O(\log n)$ time by Lemma~\ref{lemma:lemma5.5}.  Thus, as long as
$D_1\ne\varnothing$, $\HHdagger$ is not taming and an
$\HHdagger$-wild set consisting of a single vertex can be obtained
from $D_1$ in $O(1)$ time, implying
Lemmas~\ref{lemma:lemma5.1}\eqref{item1:lemma5.1},~\ref{lemma:lemma5.1}\eqref{item2:lemma5.1},
and~\ref{lemma:lemma5.1}\eqref{item3:lemma5.1}.  The rest of the
subsection handles the case $D_1=\varnothing$.

\begin{lem}[Holm, de Lichtenberg, and Thorup~\cite{Holm:2001}]
\label{thm:dyn-graph}
A spanning forest of an $n$-vertex dynamic graph can be maintained in
amortized $O(\log^2 n)$ time per edge insertion and deletion such that
each update to the graph only adds and deletes at most one edge in the
spanning forest.
\end{lem}

We maintain a spanning forest $F$ of the decremental graph $G-X$ by
Lemma~\ref{thm:dyn-graph}.  For each maximal connected $U\subseteq
V(F)$, we maintain a balanced binary search tree $T_U$ on $U$.  For
each $y\in U$, we maintain a representative set $R_U(y)$ for the union
of $R_z$ over all vertices $z$ in the subtree of $T_U$ rooted at $y$.
It takes $O(1)$ time to determine if $U$ is tamed from
$R_U=R_U(\text{root}(T_U))$.  We also maintain a doubly linked list
$D_2$ for the untamed maximal connected subsets $U$ of $V(F)$.  When
$R_y$ for a vertex $y\in V(G-X)$ changes, $D_2$ and $R_U$ for the
maximal connected $U\subseteq V(F)$ containing $y$ can be updated in
$O(\log n)$ time by Lemma~\ref{lemma:lemma5.5}.  If deleting an edge
of $F$ decomposes a maximal connected $U\subseteq V(F)$ into $U_1$ and
$U_2$ with $|U_1|\leq |U_2|$, then it takes $O(|U_1| \log n)$ time to
delete the vertices of $U_1$ from $T_U$, construct $T_{U_1}$, and
obtain $R_{U_1}$.  The resulting $T_U$ and $R_U$ become $T_{U_2}$ and
$R_{U_2}$.  $D_2$ can be updated in $O(1)$ time.  Whenever a vertex
$y$ moves to a new connected component, the number of vertices of the
connected component containing $y$ is halved.  Hence, the $T_U$ for
all maximal connected sets $U\subseteq V(F)$ are changed overall
$O(n\log n)$ times.  Thus, the overall time throughout the algorithm
to maintain $D_2$ and all representative sets $R_U$ is $O(n\log^2 n)$,
not affecting the correctness of
Lemmas~\ref{lemma:lemma5.1}\eqref{item1:lemma5.1}
and~\ref{lemma:lemma5.1}\eqref{item2:lemma5.1} and the first half of
Lemma~\ref{lemma:lemma5.1}\eqref{item3:lemma5.1}.  It remains to prove
the second half of Lemma~\ref{lemma:lemma5.1}\eqref{item3:lemma5.1}
for the case $D_1=\varnothing$ and $D_2\ne\varnothing$, i.e., each
$N(y,X)$ with $y\in V(G-X)$ is tamed and $\HHdagger$ is not taming.

\begin{figure}[t]
\centerline{\scalebox{0.47}{\includegraphics{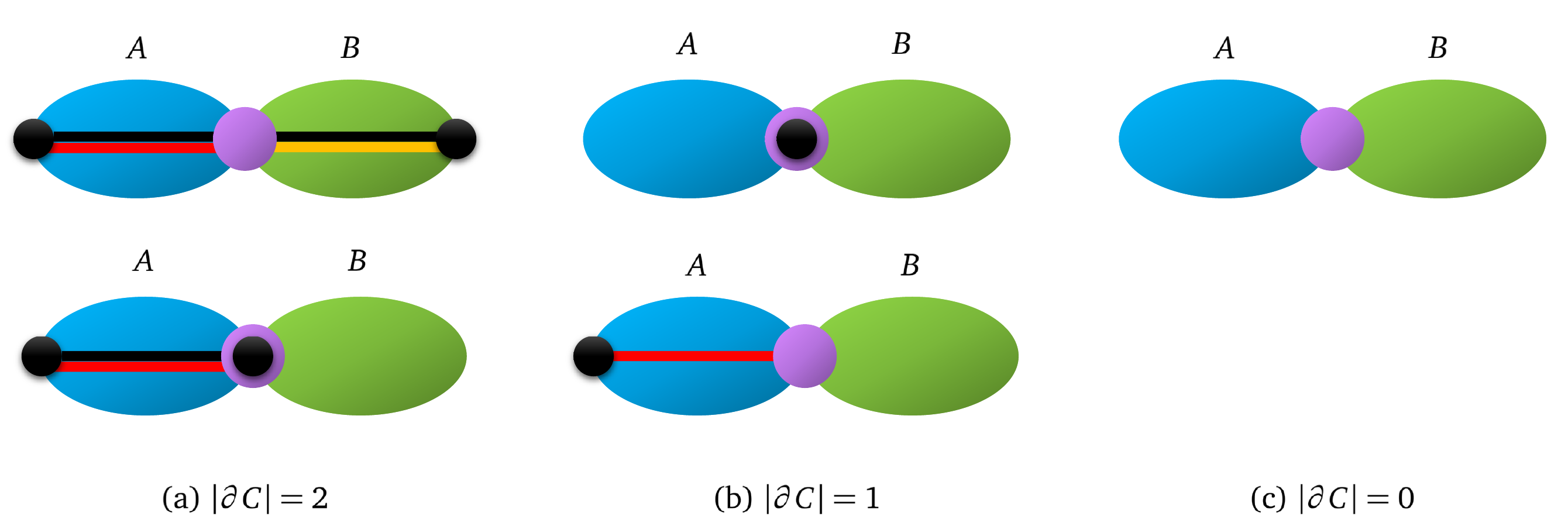}}}
\caption{The cases of joining the child clusters $A$ and $B$ with
  $|\partial A|\geq |\partial B|$ into their parent cluster $C=A\cup
  B$ on a top tree.  The first row shows the three cases with
  $|\partial A|=|\partial B|$.  The second row shows the two cases
  with $|\partial A|>|\partial B|$.  The vertex in $A\cap B$ is in
  purple.  The vertices in $\partial C$ are in black.  If $|\partial
  C|=2$, then the black line indicates $\Pi(C)$.  If $|\partial A|=2$,
  then the red line indicates $\Pi(A)$. If $|\partial B|=2$, then the
  yellow line indicates $\Pi(B)$.}
\label{figure:figure11}
\end{figure}

A top tree is defined over a dynamic tree $T$ and a dynamic set
$\partial T$ of at most two vertices of $T$. For any subtree $C$ of
$T$, $\partial C=\partial_{(T,\partial T)}C$ consists of the vertices
of $C$ belonging to $\partial T$ or adjacent to $V(T)\setminus V(C)$.
A {\em cluster}~\cite{Alstrup:2005} of $(T,\partial T)$ is a subtree
$C$ of $T$ with $|E(C)|\geq 1$ and $|\partial C|\leq 2$.  If
$|\partial C|=2$, then let $\Pi(C)$ denote the path of $T$ between the
vertices of $\partial C$.  If $|E(T)|=0$, then $(T,\partial T)$ admits
no cluster and the {\em top tree} over $(T,\partial T)$ is empty.  If
$|E(T)|\geq 1$, then a {\em top tree} $\cT$ over $(T,\partial T)$ is a
binary tree on clusters of $(T,\partial T)$ such that (1) the root of
$\cT$ is the maximal cluster $T$ of $(T,\partial T)$, (2) the leaves
of $\cT$ are the edges of $T$, i.e., the minimal clusters of
$(T,\partial T)$, and (3) the children $A$ and $B$ of any cluster $C$
of $(T,\partial T)$ on $\cT$ are edge disjoint clusters of
$(T,\partial T)$ with $C=A\cup B$ and $|V(A)\cap V(B)|=1$.
Figure~\ref{figure:figure11} illustrates all possible cases of joining
child clusters $A$ and $B$ into their parent cluster $C$ on $\cT$.  If
$|\partial A|=|\partial C|=2$, then $\Pi(A)\subseteq
\Pi(C)$. Moreover, $\Pi(A)=\Pi(C)$ if and only if $|\partial B|\leq
1$.  For each vertex $v\in V(T)\setminus\partial T$, let $C_v$ denote
the lowest cluster of $(T,\partial T)$ on $\cT$ with $v\in
V(C_v)\setminus \partial C_v$.  If $|\partial C|=2$, then $v\in V(C)$
is an internal vertex of $\Pi(C)$ if and only if $|\partial A|=2$
holds for every cluster $A$ on the $CC_v$-path of $\cT$.  A {\em top
  forest} $\cF$ over a forest $F$ consists of top trees, one for each
maximal subtree of $F$.  According to Lemma~\ref{thm:dyn-graph}, each
update to $F$ either deletes an edge of $F$ or adds an edge between
two maximal subtrees of $F$.  In addition to that, $\cF$ also needs be
modified if $\partial T$ for a maximal subtree $T$ of $F$ is updated.
To accommodate each update to $F$ or $\partial T$, we modify $\cF$ via
a sequence of operations such that there can be temporary top trees
$\cT_C$ rooted at clusters $C$ that are not maximal subtrees of $F$.
Specifically, $\cF$ is modified via the following $O(1)$-time {\em
  top-tree operations}:
\begin{itemize}[leftmargin=*]
\item 
{\em Create} or {\em destroy} a top tree on a single cluster that is
an edge.

\item 
{\em Split} a top tree $\cT_C$ into the two immediate subtrees of
$\cT_C$ by deleting the root $C$.

\item {\em Merge} top trees $\cT_A$ and $\cT_B$ with $|V(A)\cap
  V(B)|=1$ into a top tree $\cT_C$ rooted at $C=A\cup B$.
\end{itemize}

\begin{lem}[Alstrup, Holm, de Lichtenberg, and Thorup~\cite{Alstrup:2005}]
\label{lemma:lemma5.6}
An $n$-vertex forest $F$ admits an $O(n)$-space top forest $\cF$
consisting of $O(\log n)$-height top trees such that for any maximal
subtree $T$ of~$F$,
\begin{enumerate}[leftmargin=*]
\item 
\label{item1:lemma5.6}
it takes $O(1)$ time to obtain on the top tree $\cT$ for $T$ (a) the
cluster $C_v$ for any $v\in V(T)\setminus \partial T$, (b) the parent
of a nonroot cluster, (c) the children of a non-leaf cluster, and (d)
$\partial C$ for a cluster $C$ and

\item 
\label{item2:lemma5.6}
it takes $O(\log n)$ time to identify a sequence of $O(\log n)$
top-tree operations with which $\cF$ can be modified in $O(\log n)$
time with respect to (a) updating $\partial T$, (b) deleting an edge
of $T$, or (c) adding an edge between $T$ and another maximal subtree
of $F$.
\end{enumerate}
\end{lem}

We use Lemma~\ref{lemma:lemma5.6} to maintain a top forest $\cF$ over
the spanning forest $F$ of $G-X$ maintained by
Lemma~\ref{thm:dyn-graph}.  For each cluster $C$ on each nonempty top
tree $\cT$ of $\cF$, we maintain a representative set $R_C$ of
$N(V(C)\setminus \partial C,X)$.  We first show that maintaining the
representative sets $R_C$ does not affect the complexity of
maintaining $\cF$ stated in Lemma~\ref{lemma:lemma5.6} and that of
maintaining the colors of the vertices of $X$ stated in
Lemmas~\ref{lemma:lemma5.1}\eqref{item1:lemma5.1}
and~\ref{lemma:lemma5.1}\eqref{item2:lemma5.1}.  By
Lemma~\ref{lemma:lemma5.5}, the following {\em bottom-up update} for a
cluster $B$ on a top tree $\cT$ of $\cF$ takes $O(\log n)$ time: For
each cluster $C$ on the $BT$-path of $\cT$ from $B$ to $T$, if $C$ is
an edge $uv$ of $T$, then an $R_C$ can be obtained from $R_u\cup R_v$
in $O(1)$ time; if $C$ is not an edge of $T$, then an $R_C$ can be
obtained from $R_{C_1}\cup R_{C_2}\cup R_c$ in $O(1)$ time, where
$C_1$ and $C_2$ are the children of $C$ on $\cT$ and $c$ is the vertex
in $V(C_1)\cap V(C_2)$.  Hence, the initial $R_C$ for all clusters $C$
of all top trees $\cT$ of $\cF$ can be obtained in overall $O(m\log
n)$ time by performing a bottom-up update for each leaf cluster of
each top tree.  With respect to each top-tree operation, the
representative sets $R_C$ can be updated in $O(1)$ time: For destroy
and split, we simply delete $R_C$ together with the root $C$ of
$\cT_C$.  For create and merge, we just perform a bottom-up update for
$C$ in $O(1)$ time.  Thus, maintaining the representative sets $R_C$
does not affect the complexity of maintaining $\cF$ stated in
Lemma~\ref{lemma:lemma5.6}.  If a vertex $v\in V(G-X)$ moves to $X$ or
a vertex $v\in X$ changes color, we update $R_C$ for all $O(d(v)\log
n)$ clusters $C$ with $v\in N(V(C)\setminus\partial C,X)$.
Specifically, for each of the $O(d(v))$ vertices $y\in V(G-X)$ with
$v\in N(y,X)$, we perform a bottom-up update for $C_y$ in $O(\log n)$
time.  Thus, maintaining the representative sets $R_C$ does not affect
the correctness of Lemmas~\ref{lemma:lemma5.1}\eqref{item1:lemma5.1}
and~\ref{lemma:lemma5.1}\eqref{item2:lemma5.1}.

The rest of the subsection proves the second half of
Lemma~\ref{lemma:lemma5.1}\eqref{item3:lemma5.1} for the case
$D_1=\varnothing$ and $D_2\ne\varnothing$ in two steps.  Let $T=F[U]$
for an arbitrary $U$ kept in $D_2$.  Step~1 calls
$\textsc{tree-wild}(T)$ to obtain $\{u,w\}$ for distinct vertices $u$
and $w$ of $T$ such that the vertices of the $uw$-path of $T$ is a
minimal untamed connected vertex set of $T$. Step~2 calls
$\textsc{graph-wild}(\{u,w\})$ to obtain a minimal untamed set $Y$
such that $G[Y]$ is a $uw$-path of $G$.

Step~1: Let $\cT$ be the top tree of $\cF$ for $T$.  For any cluster
$C$ on $\cT$, let $R_{\partial C}$ be the union of $R_v$ over the
vertices $v\in \partial C$.  Let $\textsc{closest}(S,C,c)$ for
\begin{itemize}[leftmargin=*]
\item a tamed set $S\subseteq X$ with $|S|\leq 6$,
\item a cluster $C$ on $\cT$ with untamed $S\cup R_C\cup R_{\partial
  C}$, and
\item a vertex $c\in \partial C$ such that $S\cup R_c$ is tamed
\end{itemize} 
be the following $O(\log n)$-time recursive algorithm that outputs a
$y\in V(C)$ such that
\begin{itemize}[leftmargin=*]
\item $S\cup R_y$ is untamed and \item $S\cup R_z$ is tamed for every
  internal vertex $z$ of the $yc$-path of $T$:
\end{itemize}
If $C$ is an edge $bc$, then return $b$.  If $C$ is not an edge, then
let $C_1$ and $C_2$ be the children of $C$ and let $b$ be the vertex
in $C_1\cap C_2$.  
If there is an $i\in\{1,2\}$ with $c\in\partial C_i$ such that $S\cup
R_{C_i}\cup R_{\partial C_i}$ is untamed, then return
$\textsc{closest}(S,C_i,c)$.
Otherwise, we have $b\ne c$ and that $S\cup R_{C_i}\cup R_{\partial
  C_i}$ is untamed for the index $i\in\{1,2\}$ with $c\notin\partial
C_i$.  Return $\textsc{closest}(S,C_i,b)$.

Let $\textsc{tree-wild}(C)$ for a cluster $C$ on $\cT$ with untamed
$R_C\cup R_{\partial C}$ be the following recursive subroutine: If $C$
is an edge $uw$ of $T$, then return $\{u,w\}$.  Otherwise, let $C_1$
and $C_2$ be the children of $C$ on $\cT$.  If there is an $i\in
\{1,2\}$ with untamed $R_{C_i}\cup R_{\partial C_i}$, then return
$\textsc{tree-wild}(C_i)$.  Otherwise, $R_C\cup R_{\partial C}$ is
untamed and $R_{C_1}\cup R_{\partial C_1}$ is tamed.  Let $c$ be the
vertex in $V(C_1)\cap V(C_2)$.  Call $\textsc{closest}(R_{C_1}\cup
R_{\partial C_1},C_2,c)$ to obtain in $O(\log n)$ time a $w\in V(C_2)$
such that
\begin{itemize}[noitemsep,leftmargin=*]
\item $R_{C_1}\cup R_{\partial C_1}\cup R_w$ is untamed and
\item $R_{C_1}\cup R_{\partial C_1}\cup R_v$ is tamed for every
  internal vertex $v$ of the $wc$-path of $T$.
\end{itemize}
Call $\textsc{closest}(R_w,C_1,c)$ to obtain in $O(\log n)$ time a
$u\in V(C_1)$ such that
\begin{itemize}[noitemsep,leftmargin=*]
\item $R_w\cup R_u$ is untamed and
\item $R_w\cup R_v$ is tamed for every internal vertex $v$ of the
  $uc$-path of $T$.
\end{itemize}
Let $P$ be the $uw$-path of $T$.  $V(P)$ is a minimally untamed subset
of $V(T)$ that is connected in~$T$: Let $u'$ and $w'$ be distinct
vertices of $V(P)$ with $\{u',w'\}\ne\{u,w\}$ such that $R_{u'}\cup
R_{w'}$ is untamed and $u'$ is closer to $u$ than $w$ in $P$.  Since
$R_{C_1}\cup R_{\partial C_1}$ and $R_{C_2}\cup R_{\partial C_2}$ are
both tamed, we have $u'\in V(C_1)\setminus \partial C_1$ and $w'\in
V(C_2)\setminus \partial C_2$.  Since $R_{C_1}\cup R_{\partial
  C_1}\cup R_v$ is tamed for every internal vertex $v$ of the
$wc$-path of $T$ and $u'\in V(C_1)$, we have $w'=w$.  Since $R_w\cup
R_v$ is tamed for every internal vertex $v$ of the $uc$-path of $T$,
we have $u'=u$.

Step~2: To obtain in $O(d(Y)\log n)$ time a set $Y$ such that $G[Y]$
is a $uw$-path of $G-X$, it suffices to show an $O(d(u)\log n)$-time
subroutine $\textsc{jump}(u,w)$ returning for any distinct vertices
$u$ and $w$ of $T$ the vertex $v\in N_G(u,V(P))$ that is closest to
$w$ in the $uw$-path $P$ of $T$: With $Y=\{u\}$ initially, we
repeatedly add $v=\textsc{jump}(u,w)$ into $Y$ and let $u=v$ until
$v=w$.  The subroutine $\textsc{jump}(u,w)$ starts with updating $\cT$
for setting $\partial T=\{u,w\}$ in $O(\log n)$ time by
Lemma~\ref{lemma:lemma5.6}\eqref{item2:lemma5.6}.  Recall that
$U=N_G(u,V(P-w))$ consists of the vertices $v\in N_G(u)$ such that
$|\partial B|=2$ holds for every cluster $B$ on the $TC_v$-path of
$\cT$.  By Lemma~\ref{lemma:lemma5.6}\eqref{item1:lemma5.6}, it takes
$O(d(u)\log n)$ time for $\textsc{jump}(u,w)$ to obtain $U$ and the
set $\mathcal{C}$ consisting of the clusters on the $TC_v$-path of
$\cT$ for all vertices $v\in U$.  If $U=\varnothing$, then
$\textsc{jump}(u,w)$ returns $w$, since $uw$ is an edge of $T$.  If
$U\ne\varnothing$, then $\textsc{jump}(u,w)$ returns
$v=\textsc{next}(T,w)$, where $\textsc{next}(C,w)$ for a cluster $C\in
\mathcal{C}$ and a vertex $w\in \partial C$ is the following $O(\log
n)$-time recursive subroutine: If $w\in N_G(u)$, then
$\textsc{next}(C,w)$ returns $w$.  If $w\notin N_G(u)$, then $C$ is
not an edge of $T$.  Let $C_1$ and $C_2$ be the children of $C$ on
$\cT$ with $w\in \partial C_2\setminus \partial C_1$.  Let $c$ be the
vertex in $V(C_1)\cap V(C_2)$.  If $C_2\in \mathcal{C}$, then
$\textsc{next}(C,w)$ returns $\textsc{next}(C_2,w)$; otherwise,
$\textsc{next}(C,w)$ returns $\textsc{next}(C_1,c)$.

\section{Improved graph recognition and detection algorithms}
\label{section:section6}
Section~\ref{subsection:subsection6.1} gives our 
algorithms for detecting thetas, pyramids, and beetles.
Section~\ref{subsection:subsection6.2} gives our 
algorithms for recognizing perfect graphs and detecting odd holes.
Section~\ref{subsection:subsection6.3} gives our 
algorithm for detecting even holes.

\subsection{Improved theta, pyramid, and beetle detection}
\label{subsection:subsection6.1}

Each previous algorithm for detecting a family $\FF$ of graphs in $G$
via the three-in-a-tree algorithm identifies a set $\GG$ of a
polynomial number of subgraphs $H$ of $G$, each associated with a set
$L(H)$ of three terminals, such that $G$ is $\FF$-free if and only if
each graph $H$ in $\GG$ does not admit an induced tree containing
$L(H)$.  In addition to Theorem~\ref{theorem:theorem1.1}, our
improvement are obtained via exploiting that the graphs $H$ in $\GG$
need not be subgraphs of $G$.  For instance, if $\FF$ are thetas, then
Chudnovsky and Seymour~\cite{ChudnovskyS10} obtained a set $\GG$ of
$O(n^7)$ subgraphs of $G$.  Each $H\in \GG$ with
$L(H)=\{a_1,a_2,a_3\}$ is uniquely determined from vertices $b$,
$b_1$, $b_2$, $b_3$, $a_1$, $a_2$, and $a_3$ of $G$ such that $bb_1$,
$bb_2$, $bb_3$, $a_1b_1$, $a_2b_2$, and $a_3b_3$ are the distinct
edges of $G[\{b,b_1,b_2,b_3,a_1,a_2,a_3\}]$.  We observe that the
requirement that $a_1b_1$, $a_2b_2$, and $a_3b_3$ are the distinct
edges of $G[\{a_1,a_2,a_3,b_1,b_2,b_3\}]$ can be achieved by making
the neighbors of each $b_i$ with $i\in\{1,2,3\}$ in $V(G)\setminus
\{b,b_1,b_2,b_3\}$ a clique. As a result, each $H\in \GG$ is
determined from four vertices $b$, $b_1$, $b_2$, and $b_3$ such that
$bb_1$, $bb_2$, and $bb_3$ are the distinct edges of
$G[\{b,b_1,b_2,b_3\}]$.  Thus, there is a set $\GG$ of $O(n^4)$
$n$-vertex graphs $H$ with $L(H)=\{b_1,b_2,b_3\}$ such that $G$ is
theta-free if and only if each graph $H$ in $\GG$ does not admit an
induced tree containing $L(H)$.  An $n^3$-factor is reduced from the
number of the three-in-a-tree problems to be solved in order to
determine whether $G$ is theta-free.  Beetle detection can be improved
similarly.  Improving the algorithm for pyramid detection needs
additional care, since a pyramid has to contain exactly one triangle.

\subsubsection[]{Proving Theorem~\ref{theorem:theorem1.2}}
\label{subsubsection:subsubsection6.1.1}

Theorem~\ref{theorem:theorem1.2} is immediate from
Theorem~\ref{theorem:theorem1.1} and the next lemma.

\begin{lem}
\label{lemma:lemma6.1}
Thetas in an $n$-vertex $m$-edge graph $G$ can be detected by solving
the three-in-a-tree problem on $O(mn^2)$ linear-time-obtainable
$n$-vertex graphs.
\end{lem}

\begin{proof}
Observe that $H$ is a theta of $G$ if and only if there are vertices
$b$, $b_1$, $b_2$, and $b_3$ of $H$ such that $bb_1$, $bb_2$, and
$bb_3$ are the distinct edges of $G[\{b,b_1,b_2,b_3\}]$ and $H-b$ is
an induced subtree of $G-b$ having exactly three leaves $b_1$, $b_2$,
and $b_3$. See Figure~\ref{figure:figure2}(a).  For each of the
$O(mn^2)$ choices of vertices $b$, $b_1$, $b_2$, and $b_3$ such that
$bb_1$, $bb_2$, and $bb_3$ are the distinct edges in
$G[\{b,b_1,b_2,b_3\}]$, let $G(b,b_1,b_2,b_3)$ denote the graph that
is $O(m+n)$-time obtainable from $G$ by (1) deleting $N[b]\setminus
\{b_1,b_2,b_3\}$ and (2) adding edges to make the remaining vertices
in each $N(b_i)$ with $i\in\{1,2,3\}$ a clique.  We show that $G$
admits a theta $H$ if and only if one of the $O(mn^2)$ graphs
$G^*=G(b,b_1,b_2,b_3)$ admits an induced subtree $T^*$ containing
$\{b_1,b_2,b_3\}$.

($\Rightarrow$)\quad $G^*=G(b,b_1,b_2,b_3)$ exists for the vertices
$b$, $b_1$, $b_2$, and $b_3$ of $H$.  The vertices deleted from $G$ in
Step~(1) are not in $T=H-b$, implying that $T$ is a subtree of $G^*$
containing $\{b_1,b_2,b_3\}$.  Since $b_1$, $b_2$, and $b_3$ are the
leaves of $T$, each edge added by Step~(2) is incident to at most one
vertex of $T$, implying that $T$ is an induced subtree $T^*$ of $G^*$
containing $\{b_1,b_2,b_3\}$.

($\Leftarrow$)\quad The distinct edges of $G[\{b,b_1,b_2,b_3\}]$ are
$bb_1$, $bb_2$, and $bb_3$.  By Step~(2), $b_1$, $b_2$, and $b_3$ are
the leaves of~$T^*$.  Since each edge deleted in Step~(1) is incident
to at most one vertex of $T^*$, $T^*$ is an induced subtree of $G-b$,
implying that $G[T^*\cup \{b\}]$ is a theta $H$ of $G$.
\end{proof}

\subsubsection[]{Proving Theorem~\ref{theorem:theorem1.3}}
\label{subsubsection:subsubsection6.1.2}
A {\em pyramid}~\cite{ChudnovskyS10} of graph $G$ is the subgraph of
$G$ induced by the vertices of an induced subtree $T$ of
$G-\{b_1b_2,b_2b_3,b_3b_1\}$ having exactly three leaves $b_1$, $b_2$,
and $b_3$ such that $G[\{b_1,b_2,b_3\}]$ is the only triangle of
$G[T]$.  See Figure~\ref{figure:figure2}(b).
Theorem~\ref{theorem:theorem1.3} is immediate from
Theorem~\ref{theorem:theorem1.1} and the next lemma.

\begin{lem}
\label{lemma:lemma6.2}
Pyramids in an $n$-vertex $m$-edge graph $G$ can be detected by
solving the three-in-a-tree problem on $O(mn)$ linear-time-obtainable
$n$-vertex graphs.
\end{lem}

\begin{proof}
For each of the $O(mn)$ choices of distinct vertices $b_1$, $b_2$, and
$b_3$ such that $G[\{b_1,b_2,b_3\}]$ is a triangle, let
$G(b_1,b_2,b_3)$ be the graph obtained from $G$ by (1) adding edges to
make each $N(b_i)\setminus \{b_1,b_2,b_3\}$ with $i\in\{1,2,3\}$ a
clique, (2) deleting edges $b_1b_2$, $b_2b_3$, and $b_3b_1$, and (3)
deleting $(N(b_i)\cap N(b_j))\setminus \{b_1,b_2,b_3\}$ for any
distinct indices $i$ and $j$ in $\{1,2,3\}$.  We show that $G$ admits
a pyramid $H$ if and only if one of the $O(mn)$ graphs
$G^*=G(b_1,b_2,b_3)$ admits an induced subtree $T^*$ containing
$\{b_1,b_2,b_3\}$.

($\Rightarrow$)\quad $G^*$ exists for the vertices $b_1$, $b_2$, and
$b_3$ of $H$.  Since $H[\{b_1,b_2,b_3\}]$ is the only triangle of $H$,
$H$ does not intersect any $(N(b_i)\cap N(b_j))\setminus
\{b_1,b_2,b_3\}$ with $1\leq i<j\leq 3$.  Hence, Steps~(2) and~(3) do
not delete any edge of $T$, implying that $T$ is a subtree of $G^*$.
Since $T$ is an induced tree of $G-\{b_1b_2,b_2b_3,b_3b_1\}$ having
exactly three leaves $b_1$, $b_2$, and $b_3$, each edge added by
Step~(1) is incident to at most one vertex of $T$.  Thus, $T$ is an
induced subtree $T^*$ of $G^*$ containing $\{b_1,b_2,b_3\}$.

($\Leftarrow$)\quad By Step~(1), vertices $b_1$, $b_2$, and $b_3$ are
the leaves of the subtree $T^*$ of $G$.  Since each edge deleted in
Step~(3) is incident to at most one vertex of $T^*$, $T^*$ is an
induced subtree of $G-\{b_1b_2,b_2b_3,b_3b_1\}$ by Step~(2).  By
Steps~(2) and~(3), $G[\{b_1,b_2,b_3\}]$ is the only triangle of
$G[T^*]$.  Thus, $G[T^*]$ is a pyramid $H$ of $G$.
\end{proof}

\subsubsection[]{Proving Theorem~\ref{theorem:theorem1.5}}
\label{subsubsection:subsubsection6.1.3}
A {\em beetle}~\cite{ChangL15} of graph $G$ is an induced subgraph of
$G$ consisting of a cycle $b_1b_2b_3b_4b_1$ with a chord $b_2b_4$
(i.e., a {\em diamond}~\cite{ConfortiCKV02b,KloksMV09} of $G$) and a
tree $T$ of $G-b_4$ having exactly three leaves~$b_1$, $b_2$, and
$b_3$.  See
Figure~\ref{figure:figure2}(c). Theorem~\ref{theorem:theorem1.5} is
immediate from Theorem~\ref{theorem:theorem1.1} and the next lemma.

\begin{lem}
\label{lemma:lemma6.3}
Beetles in an $n$-vertex $m$-edge graph $G$ can be detected by solving
the three-in-a-tree problem on $O(m^2)$ linear-time-obtainable
$n$-vertex graphs.
\end{lem}

\begin{proof}
For each of the $O(m^2)$ choices of vertices $b_1$, $b_2$, $b_3$, and
$b_4$ such that $G[\{b_1,b_2,b_3,b_4\}]$ is a cycle $b_1b_2b_3b_4b_1$
with exactly one chord $b_2b_4$, let $G(b_1,b_2,b_3,b_4)$ be the
$O(m+n)$-time obtainable graph from $G$ by (1) deleting
$N[b_4]\setminus \{b_1,b_2,b_3\}$ and (2) adding edges to make the
remaining vertices in each $N(b_i)\setminus\{b_1,b_2,b_3\}$ with
$i\in\{1,2,3\}$ a clique.  We show that $G$ admits a beetle $H$ if and
only if one of the $O(m^2)$ graphs $G(b_1,b_2,b_3,b_4)$ admits an
induced subtree $T^*$ containing $\{b_1,b_2,b_3\}$.

($\Rightarrow$)\quad $G^*=G(b_1,b_2,b_3,b_4)$ exists for the vertices
$b_1$, $b_2$, $b_3$, and $b_4$ of $H$.  The vertices deleted fro $G$
in Step~(1) are not in $T$, implying that $T$ is a subtree of $G^*$
containing $\{b_1,b_2,b_3\}$.  Since $T$ intersects each
$N(b_i)\setminus\{b_1,b_2,b_3\}$ with $i\in\{1,2,3\}$ at exactly one
vertex, each edge added by Step~(2) is incident to at most one vertex
of $T$.  Thus, $T$ is an induced subtree $T^*$ of $G^*$ containing
$\{b_1,b_2,b_3\}$.

($\Leftarrow$)\quad $G[\{b_1,b_2,b_3,b_4\}]$ is a cycle
$b_1b_2b_3b_4b_1$ with exactly one chord $b_2b_4$.  By Step~(2),
$b_1$, $b_2$, and $b_3$ are the leaves of $T^*$.  Since each edge
deleted in Step~(1) is incident to at most one vertex of $T^*$, we
have $G[T^*]=T^*\cup \{b_1b_2,b_2b_3\}$, implying that $G[T^*\cup
  \{b_4\}]$ is a beetle $H$ of $G$.
\end{proof}

\subsection{Improved perfect-graph recognition and odd-hole detection}
\label{subsection:subsection6.2}

As summarized by Maffray and Trotignon~\cite[\S2]{MaffrayT05}, the
algorithm of Chudnovsky et al.~\cite{ChudnovskyCLSV05} consists of two
$O(n^9)$-time phases.  The first phase (a) detects pyramids in $G$ in
$O(n^9)$ time, (b) detects the so-called $\mathcal{T}_i$
configurations with $i\in\{1,2,3\}$ in $O(n^6)$ time,\footnote{
  in~\cite{MaffrayT05} we omit the complicated definitions of
  $\mathcal{T}_i$ configurations, which are not needed by our improved
  algorithms.}
and (c) detects jewels in $\bar{G}$ in $O(n^6)$ time.  If any of them
is detected, then either $G$ or $\bar{G}$ contains odd holes, implying
that $G$ is not perfect.  Otherwise, each shortest odd hole $C$ of $G$
is amenable, i.e., any anti-connected component of the $C$-major
vertices is contained by $N_G(u)\cap N_G(v)$ for some edge $uv$ of
$C$.  The second phase (a) computes in $O(n^5)$ time a set
$\mathbbmsl{X}$ of $O(n^5)$ subsets of $V(G)$ such that if $G$
contains an amenable shortest odd hole, then $\mathbbmsl{X}$ contains
a near cleaner of $G$ and (b) spends $O(n^4)$ time on each $X\in
\mathbbmsl{X}$ to either obtain an odd hole of $G$ or ensure that $X$
is not a near cleaner of $G$.  Theorem~\ref{theorem:theorem1.3}
reduces the time of detecting pyramids to $O(n^6)$.
Lemma~\ref{lemma:lemma6.5} reduces the time of Phase~2(b) from
$O(n^4)$ to the time of performing $O(n)$ multiplications of Boolean
$n\times n$ matrices~\cite{CoppersmithW90,LeGall14,will12}.
Therefore, the time of recognizing perfect graphs is already reduced
to $O(n^{8.377})$ without resorting to our improved odd-hole detection
algorithm.

Let $G$ be an $n$-vertex $m$-edge graph.  A {\em $k$-hole}
(respectively, {\em $k$-cycle} and {\em $k$-path}) is a $k$-vertex
hole (respectively, cycle and path).  For any odd hole $C$ of $G$, a
vertex $x\in V(G)\setminus V(C)$ is {\em
  $C$-major}~\cite{ChudnovskyCLSV05} if $N_G(x,C)$ is not contained by
any $3$-path of $C$.  Let $M_G(C)$ consist of the $C$-major
vertices. We have $M_G(C)\cap V(C)=\varnothing$.  A shortest odd hole
$C$ of $G$ is {\em clean} if $G$ does not contain any $C$-major
vertex.  A set $X\subseteq V(G)$ is a {\em near
  cleaner}~\cite{ChudnovskyCLSV05} if there is a shortest odd hole $C$
of $G$ such that (1) $C[X]$ is contained by a $3$-path of $C$ and (2)
all $C$-major vertices of $G$ are in $X$.  A jewel of $G$ is an
$O(n^6)$-time detectable induced subgraph of
$G$~\cite{ChudnovskyCLSV05}.  If $G$ contains jewels or beetles, then
$G$ contains odd holes.  Let $\bar{G}$ denote the complement of graph
$G$.

\begin{lem}[{Chudnovsky, Cornu{\'{e}}jols, Liu, Seymour, and Vu\v{s}kovi\'{c}
\cite[4.1]{ChudnovskyCLSV05}}]
\label{lemma:lemma6.4}
Let $u$ and $v$ be distinct vertices of a clean shortest odd hole $C$
of a pyramid-free jewel-free graph $G$.
\begin{enumerate*}[label=(\arabic*), ref=\arabic*]
\item
\label{item1:lemma6.4}
The shortest $uv$-path of $C$ is a shortest $uv$-path of $G$.
\item
\label{item2:lemma6.4}
The graph obtained from $C$ by replacing the shortest $uv$-path of $C$
with a shortest $uv$-path of $G$ remains a clean shortest odd hole of
$G$.
\end{enumerate*}
\end{lem}

\subsubsection{An improved algorithm for recognizing perfect graphs}
\label{subsubsection:subsubsection6.2.1}

Although Theorem~\ref{theorem:theorem1.4}\eqref{item1:theorem1.4}
implies Theorem~\ref{theorem:theorem1.4}\eqref{item2:theorem1.4}, this
subsection shows that we already have an improved algorithm for
recognizing perfect graphs without resorting to
Theorem~\ref{theorem:theorem1.4}\eqref{item1:theorem1.4}.  The next
lemma reduces the time of Chudnovsky et al.'s algorithms~\cite[4.2 and
  5.1]{ChudnovskyCLSV05} from $O(n^4)$ to $O(n^{3.377})$.

\begin{lem}
\label{lemma:lemma6.5}
For any given vertex set $X$ of an $n$-vertex pyramid-free jewel-free
graph $G$, it takes the time of performing $O(n)$ multiplications of
$n\times n$ Boolean matrices to either obtain an odd hole of $G$ or
ensure that $X$ is not a near cleaner of a shortest odd hole of $G$.
\end{lem}

\begin{proof}
It takes overall $O(n^3)$ time to obtain for any distinct vertices $u$
and $v$ of $G$ that are connected in $G(u,v)=G-(X\setminus \{u,v\})$
(i) the length $d(u,v)$ of a shortest $uv$-path $P(u,v)$ in $G(u,v)$
and (ii) the neighbor $N(u,v)$ of $u$ in $P(u,v)$.  Assume
$P(u,v)=P(v,u)$ for all $u$ and $v$ without loss of generality.  If
$u$ and $v$ are not connected in $G(u,v)$, then let $d(u,v)=\infty$.
It takes overall $O(n^3)$ time to compute for any distinct vertices
$x$ and $y$ of $G$ the set $Z(x,y)$ represented by an $n$-bit array,
consisting of the vertices $z$ of $G$ with $d(z,x)=1$ and
$d(z,y)>d(x,y)$.
If
\begin{equation}
\label{equation:near-cleaner}
\begin{array}{rcl}
d(x_1,x_2)&\geq &2\\
d(x_1,y_1)&=&d(x_2,y_2)\ \ =\ \ d(x_1,y_2)-1\ \ =\ \ d(x_2,y_1)-1\\ 
Z(x_1,y_1)\cap Z(x_2,y_2)&\ne&\varnothing
\end{array}
\end{equation}
with $y_1=N(y_2,x_1)$ hold for any distinct vertices $x_1$, $x_2$, and
$y_2$ with minimum $d(x_2,y_2)$, then the $O(n^2)$-time obtainable
$C=G[P(x_1,y_1)\cup P(x_2,y_2)\cup \{z\}]$ for any $z\in
Z(x_1,y_1)\cap Z(x_2,y_2)$ is an odd hole of $G$: Paths $P(x_1,y_1)$
and $P(x_2,y_2)$ are chordless. By $z\in Z(x_1,y_1)\cap Z(x_2,y_2)$,
the only neighbors of $z$ in $C$ are $x_1$ and $x_2$. By
$d(x_1,x_2)\geq 2$, $d(x_i,y_i)=d(x_i,y_{3-i})-1$ for each $i\in
\{1,2\}$, and the minimality of $d(x_2,y_2)$, the only edge between
$P(x_1,y_1)$ and $P(x_2,y_2)$ is $y_1y_2$.  Thus, $C$ is an odd hole
of $G$.  For each $y_2$, we construct a directed acyclic tripartite
graph $G(y_2)$ on three $n$-vertex sets $X_1,Z,X_2$ such that (1)
$x_1z$ with $x_1\in X_1$ and $z\in Z$ is a directed edge of $G(y_2)$
if and only if $z\in Z(x_1,N(y_2,x_1))$ and (2) $zx_2$ with $z\in Z$
and $x_2\in X_2$ is a directed edge of $G(y_2)$ if and only if $z\in
Z(x_2, y_2)$. It takes the time of multiplying two $n\times n$ Boolean
matrices to obtain the $O(n^2)$ pairs of reachability in $G(y_2)$ from
$X_1$ to $X_2$.  Thus, the time required to determine whether there is
a choice of $x_1$, $x_2$, and $y_2$ satisfying
Equation~\eqref{equation:near-cleaner} is that of performing $O(n)$
multiplications for $n\times n$ Boolean matrices.

It remains to show that such a choice of $x_1$, $x_2$, and $y_2$
exists for the case that $X$ is a near cleaner of a shortest odd hole
$C$ of $G$. Let $P$ be a $3$-path of $C$ such that $C-V(P)$ does not
intersect the $C$-major vertices of $G$, implying that $C$ is a clean
shortest odd hole of $H=G-(X\setminus V(P))$.  Let $x_1$ and $x_2$ be
the end-vertices of $P$.  Let $y_2$ be the vertex of $C$ such that the
shortest $x_1y_2$-path of $C$ is one edge longer than the shortest
$x_2y_2$-path of $C$.  By Lemma~\ref{lemma:lemma6.4}, each shortest
$x_iy_2$-path $P_i$ of $C$ with $i\in\{1,2\}$ is a shortest
$x_iy_2$-path of $H$.  Since $X$ does not intersect the interior of
$P_1$ and $P_2$, each $P(x_i,y_2)$ with $i\in\{1,2\}$ is a shortest
$x_iy_2$-path of $H$.  Applying
Lemma~\ref{lemma:lemma6.4}\eqref{item2:lemma6.4} on $C$ to replace
$P_i$ with $P(x_i,y_2)$ for each $i\in\{1,2\}$, we obtain a clean
shortest odd hole $C^*$ of $H$, via which one can verify
Equation~\eqref{equation:near-cleaner} for the chosen $x_1$, $x_2$,
and $y_2$: Let $y_1=N(y_2,x_1)$.  Since $C^*$ is chordless in $G$,
$d(x_1,x_2)\geq 2$. Since $X$ does not intersect the vertices of $C^*$
other than $x_1$, $x_2$, and the internal vertex $z$ of the shortest
$x_1x_2$-path of $C^*$, we have
$d(x_1,y_1)=d(x_2,y_2)=d(x_1,y_2)-1=d(x_2,y_1)-1$ by
Lemma~\ref{lemma:lemma6.4}\eqref{item1:lemma6.4}.  We have
$d(z,x_1)=d(z,x_2)=1$.  By
Lemma~\ref{lemma:lemma6.4}\eqref{item1:lemma6.4}, $d(z,y_i)>
d(x_i,y_i)$ for both $i\in\{1,2\}$ or else the shortest $zy_i$-path of
$C^*$ for an $i\in\{1,2\}$ would not be a shortest $zy_i$-path of $H$.
Thus, $z\in Z(x_1,y_1)\cap Z(x_2,y_2)$.
\end{proof}

\begin{lem}[Chudnovsky,
Cornu{\'{e}}jols, Liu, Seymour, and Vu\v{s}kovi\'{c}~\cite{ChudnovskyCLSV05}]
\label{lemma:lemma6.6}
Let $G$ be an $n$-vertex graph such that $G$ and $\bar{G}$ are
pyramid-and-jewel-free.  It takes $O(n^6)$ time to (1) ensure that $G$
contains odd holes or (2) obtain a set $\mathbbmsl{X}$ of $O(n^5)$
vertex subsets of $G$ such that if $G$ contains odd holes, then
$\mathbbmsl{X}$ contains a near cleaner of $G$.
\end{lem}

By Theorem~\ref{theorem:theorem1.3}, it takes $O(n^6)$ time to detect
pyramids or jewels in $G$ and $\bar{G}$.  If $G$ or $\bar{G}$ contains
pyramids or jewels, then $G$ is not perfect.  By
Lemma~\ref{lemma:lemma6.6}, it suffices to consider the case that we
are given a set $\mathbbmsl{X}$ of $O(n^5)$ vertex subsets such that
if $G$ or $\bar{G}$ is not odd-hole-free, then $\mathbbmsl{X}$
contains a near cleaner of $G$ or $\bar{G}$.  By
Lemma~\ref{lemma:lemma6.5}, it takes overall $O(n^{8.377})$
time~\cite{CoppersmithW90,LeGall14,will12} to either obtain an odd
hole of $G$ or $\bar{G}$ or ensure that both $G$ and $\bar{G}$ are
odd-hole-free.

\subsubsection{Proving Theorem~\ref{theorem:theorem1.4}}
\label{subsubsection:subsubsection6.2.2}
The recent odd-hole detection algorithm of Chudnovsky, Scott, Seymour,
and Spirkl has seven $O(n^9)$-time bottleneck subroutines.  One is for
pyramid detection, which is eliminated by
Theorem~\ref{theorem:theorem1.3}.  The remaining six are in two
groups~\cite[\S4]{ChudnovskySSS19}. The first (respectively, second)
group handles the case that the longest $x$-gap (i.e., a path $D$ of
$C$ such that $G[D\cup \{x\}]$ is a hole of $G$) over all $C$-major
vertices $x$ for a shortest odd hole $C$ is shorter (respectively,
longer) than one half of $C$.  We give a two-phase algorithm to handle
both cases in $O(n^8)$ time.  For the first case, Phase~1 tries all
$O(n^5)$ choices of five vertices to obtain an approximate cleaner for
$C$, with which a shortest odd hole can be identified in $O(n^3)$ time
via Lemmas~\ref{lemma:lemma6.5} and~\ref{lemma:lemma6.8}.  For the
second case, Phase~2 tries all $O(n^6)$ choices of six vertices to
obtain an approximate cleaner for $C$, with which a shortest odd hole
can be identified in $O(n^2)$ time via Lemma~\ref{lemma:lemma6.9}.

\begin{lem}[{Chudnovsky, Scott, Seymour, and Spirkl~\cite[Theorem~3.4]{ChudnovskySSS19}}]
\label{lemma:lemma6.7}
Let $G$ be a jewel-free, pyramid-free, and $5$-hole-free graph.  Let
$C$ be a shortest odd hole in $G$.  If $x\in M_G(C)$, then there is an
edge of $C$ adjacent to each vertex of $M_G(C)\setminus N_G(x)$ in
$G$.
\end{lem}

A vertex set $X\subseteq V(G)$ is an {\em approximate cleaner} of $C$
if $X$ contains all $C$-major vertices and $X\cap V(C)\subseteq
\{c_1,c_2\}$ holds for two vertices $c_1$ and $c_2$ with
$d_C(c_1,c_2)=3$. The second statement of the next lemma reduces the
running time of an $O(n^8)$-time subroutine of Chudnovsky et
al.~\cite[Theorem~2.4]{ChudnovskySSS19} to $O(n^5)$.

\begin{lem}
\label{lemma:lemma6.8}
For any given vertex set $X$ of an $n$-vertex $m$-edge pyramid-free
jewel-free $5$-hole-free graph $G$,
\begin{enumerate*}[label=(\arabic*), ref=\arabic*]
\item  
\label{item1:lemma6.8}
it takes $O(n^3)$ time to obtain an odd hole of $G$ or ensure that $X$
is not an approximate cleaner of any shortest odd hole of $G$ and
\item  
\label{item2:lemma6.8}
it takes $O(mn^3)$ time to either obtain an odd hole of $G$ or ensure
that there is no shortest odd hole $C$ of $G$ such that an edge of $C$
is adjacent to all $C$-major vertices of $G$.
\end{enumerate*}
\end{lem}

\begin{proof}
We first show that Statement~\ref{item1:lemma6.8} implies
Statement~\ref{item2:lemma6.8}: For each edge $b_1b_2$ of $G$, we
apply Statement~\ref{item1:lemma6.8} with $X=(N_G(b_1)\cup
N_G(b_2))\setminus \{b_1,b_2\}$ in overall $O(mn^3)$ time.  If no odd
hole is detected, then report that there is no shortest odd hole $C$
of $G$ such that an edge of $C$ is adjacent to all $C$-major vertices
of $G$.  To see the correctness, observe that if $C$ is a shortest odd
hole of $G$ such that an edge $b_1b_2$ is adjacent to all $C$-major
vertices of $G$, then $(N_G(b_1)\cup N_G(b_2))\setminus \{b_1,b_2\}$
is an approximate cleaner of $C$. Thus, Statement~\ref{item2:lemma6.8}
holds.

It remains to prove Statement~\ref{item1:lemma6.8}.  It takes overall
$O(n^3)$ time to obtain for any distinct vertices $u$ and $v$ of $G$
that are connected in $G(u,v)=G-(X\setminus \{u,v\})$ (i) the length
$d(u,v)$ of a shortest $uv$-path $P(u,v)$ in $G(u,v)$ and (ii) the
neighbor $N(u,v)$ of $u$ in $P(u,v)$.  Assume $P(u,v)=P(v,u)$ for all
$u$ and $v$ without loss of generality.  If $u$ and $v$ are not
connected in $G(u,v)$, then let $d(u,v)=\infty$.  It takes overall
$O(n^3)$ time to determine whether $C=G[P(c_1,c_2)\cup P(c_1,b)\cup
  P(c_2,b)]$ is a $7$-hole or the following equation holds for any
distinct vertices $b$, $c_1$, and $c_2$ of $G$:
\begin{equation}
\label{equation:approximate-cleaner}
\begin{array}{rcl}
d(c_1,c_2)&=&3\\
d(c_1,N(c_2,b))&>&3\\
d(c_2,N(c_1,b))&>&3\\
 d(c_1,b)&=&d(c_2,b)\ \ =\ \ d(c_1,N(b,c_2))-1\ \ =\ \ d(c_2,N(b,c_1))-1.
\end{array}
\end{equation}
If Equation~\eqref{equation:approximate-cleaner} holds for distinct
vertices $b$, $c_1$, and $c_2$ with minimum $d(c_1,b)$, then $C$ is an
odd hole of $G$: Both $P(b,c_1)$ and $P(b,c_2)$ are chordless.  By
$d(c_1,b)=d(c_2,b)=d(c_1,N(b,c_2))-1=d(c_2,N(b,c_1))-1$ and the
minimality of $d(c_1,b)$, paths $P(b,c_1)-b$ and $P(b,c_2)-b$ are
disjoint and nonadjacent.  The interior of $P(c_1,c_2)$ is disjoint
from and nonadjacent to $P((c_1,b)-c_1)\cup (P(c_2,b)-c_2)$, since
otherwise $d(c_i,N(c_{3-i},b))\leq 3$ or $d(c_i,b)\geq
d(c_i,N(b,c_{3-i}))$ would hold for an $i\in\{1,2\}$. Thus, $C$ is an
odd hole of $G$.  It remains to show that if $X$ is an approximate
cleaner for a shortest odd hole $C$ of $G$, then there is a choice of
$b$, $c_1$, and $c_2$ such that
Equation~\eqref{equation:approximate-cleaner} holds or
$C^*=G[P(c_1,c_2)\cup P(c_1,b)\cup P(c_2,b)]$ is a $7$-hole.  Let
$c_1$ and $c_2$ be two vertices of $C$ with $X\cap
V(C)\subseteq\{c_1,c_2\}$).  Thus, $C$ is a clean shortest odd hole of
$H=G-(X\setminus \{c_1,c_2\})$.  By $d_C(c_1,c_2)=3$, $|V(C)|\geq 7$,
and Lemma~\ref{lemma:lemma6.4}, we have $d(c_1,c_2)=3$.  Let $b$ be
the vertex of $C$ with $d_C(b,c_1)=d_C(b,c_2)$.  Apply
Lemma~\ref{lemma:lemma6.4} on $C$ to replace the shortest $bc_1$-path
of $C$ with $P(b,c_1)$, replace the shortest $bc_2$-path of $C$ with
$P(b,c_2)$, and replace the shortest $c_1c_2$-path of $C$ with
$P(c_1,c_2)$.  We obtain the clean shortest odd hole $C^*$ of
$H$. Suppose $|V(C^*)|\geq 9$.  By $X\cap V(C)\subseteq \{c_1,c_2\}$,
$|V(C^*)|\geq 9$, and Lemma~\ref{lemma:lemma6.4}, we have
$d(c_1,b)=d(c_2,b)=d(c_1,N(b,c_2))-1=d(c_2,N(b,c_2))-1$.  By
Lemma~\ref{lemma:lemma6.4} and $|V(C^*)|\geq 9$, we have
$d(c_i,N(c_{3-i},b))> 3$ for both $i\in\{1,2\}$.  Thus,
Equation~\eqref{equation:approximate-cleaner} holds.
\end{proof}

\begin{lem}
\label{lemma:lemma6.9}
Let $d$, $b_1$, and $b_2$ be distinct vertices of an $n$-vertex graph
$G$.  Let each $T_i$ with $i\in\{1,2\}$ be a subtree of
$G-\{b_1,b_2\}$ containing $d$.  It takes $O(n^2)$ time to determine
whether there is a leaf $c_i$ of $T_i$ for each $i\in\{1,2\}$ such
that if each $P_i$ with $i\in\{1,2\}$ is the $dc_i$-path of $T_i$,
then $G[P_1\cup \{b_1,b_2\}\cup P_2]$ is an odd hole of $G$.
\end{lem}

\begin{proof}
For each $i\in\{1,2\}$, let $T'_i$ (respectively, $T''_i$) be the
union of all $d$-to-leaf paths of $T_i$ with odd (respectively, even)
lengths.  In order for $G[P_1\cup \{b_1,b_2\}\cup P_2]$ to be an odd
hole, if $P_1$ is path of $T'_1$ (respectively, $T''_1$), then $P_2$
is a path of $T'_2$ (respectively, $T''_2$).  Therefore, it suffices
to work on the case that if each $c_i$ with $i\in\{1,2\}$ is a leaf of
$T_i$, then (1) the union of path $c_1b_1b_2c_2$ and the $dc_1$-path
$P_1$ of $T_1$ is an induced path of $G$, (2) the union of path
$c_1b_1b_2c_2$ and the $dc_2$-path $P_2$ of $T_2$ is an induced path
of $G$, and (3) $|E(P_1)|+|E(P_2)|$ is even.  It remains to show how
to determine in $O(n^2)$ time whether there is an induced
$c_1c_2$-path $P_1\cup P_2$.  For each vertex $v$ of $T_2-d$, let set
$S(v)$, implemented by an $n$-bit array associated with a counter for
$|S(v)|$, be initially empty.  Perform a depth-first traversal of
$T_1$.  When a vertex $u$ of $T_1-d$ is reached from its parent in
$T_1$, insert $u$ into $S(v)$ for each vertex $v$ of $T_2-d$ with
$u=v$ or $uv\in E(G)$ in overall $O(n)$ time.  When the traversal is
about to leave a vertex $u$ of $T_1-d$ for its parent in $T_1$, run
the following $O(n)$-time steps: If $u$ is a leaf $c_1$ of $T_1$, then
check whether there is a $dc_2$-path $P_2$ of $T_2$ for some leaf
$c_2$ of $T_2$ such that $S(v)=\varnothing$ holds for all vertices $v$
of $P_2-d$.  If there is such a $P_2$, then quit the traversal and
report an odd hole $G[P_1\cup \{b_1,b_2\}\cup P_2]$.  If $u$ is not a
leaf of $T_1$ or there is no such a $P_2$, then delete $u$ from $S(v)$
for each vertex $v$ of $T_2-d$ with $u\in S(v)$.  If the traversal
ends normally, then report negatively.  The overall running time is
$O(n^2)$.  To see the correctness, let $c_1$ be a traversed leaf of
$T_1$.  Let $c_2$ be an arbitrary leaf of $T_2$.  Let each $P_i$ with
$i\in\{1,2\}$ be the $dc_i$-path of $T_i$.  Consider the moment when
the traversal is about to leave $c_1$ for its parent in $T_1$.  By the
depth-first nature of the traversal, $S(v)\subseteq V(P_1)$ holds for
each vertex $v$ of $T_2-d$.  Therefore, $P_1\cup P_2$ is an induced
$c_1c_2$-path if and only if $S(v)=\varnothing$ holds for each vertex
$v$ of $P_2-d$.
\end{proof}

\begin{proof}[Proof of Theorem~\ref{theorem:theorem1.4}]
It suffices to prove Statement~\ref{item1:theorem1.4}.  By
Theorem~\ref{theorem:theorem1.3} and
Lemma~\ref{lemma:lemma6.8}\eqref{item1:lemma6.8}, and the fact that
jewels and $5$-holes are $O(n^6)$-time detectable, we may assume that
$G$ does not contain pyramids, jewels, $5$-holes, and clean shortest
odd holes.  By Lemma~\ref{lemma:lemma6.8}\eqref{item2:lemma6.8}, we
may further assume that $G$ does not contain any shortest odd hole $C$
such that an edge of $C$ is adjacent to all $C$-major vertices.  The
algorithm consists of two $O(m^2n^4)$-time phases. If none of them
identifies an odd hole of $G$, then report that $G$ is odd-hole-free.
Let $x$, $d$, $d_1$, $d_2$, $c_1$, $b_1$, and $b_2$ be vertices of $G$
that are not necessarily distinct.  Let
\begin{eqnarray*}
X_1&=&(N_G(b_1)\cup N_G(b_2))\setminus \{b_1,b_2\}\\
X_2&=&N_G(d_1)\cap N_G(d_2)\\
S_0&=&\{d_1,d_2\}\\
S_1&=&\{d_1,d_2,c_1\}\\
S_2&=&\{d_1,d_2,c_1,b_1\}.
\end{eqnarray*}
For each $k\in\{0,1,2\}$, let
\begin{displaymath}
H_k=G-((X_1\cup N_G(x))\setminus S_k),
\end{displaymath}
let $I_k$ consist of the internal vertices of all shortest
$d_1d_2$-paths of $H_k$, let $J_k$ consist of vertex $d$ and the
internal vertices of all shortest $dd_1$-paths and $dd_2$-paths of
$H_k$, let $Y_k=N_G(x)\cap N_G(I_k)$, and let $Z_k=N_G(x)\cap
N_G(J_k)$.  If no odd hole of $G$ is identified via the following two
phases, then report that $G$ is odd-hole-free.

Phase~1:
\begin{itemize}[leftmargin=*]
\item 
For each of the $O(m^2n)$ choices of vertices $x,d_1,d_2,b_1,b_2$ with
$x\in N_G(d_1)\cap N_G(d_2)$ and $b_1b_2\in E(G)$, apply
Lemma~\ref{lemma:lemma6.8}\eqref{item1:lemma6.8} with $X=(X_1\cup X_2
\cup Y_0)\setminus S_0$ in $O(n^3)$ time.

\item 
For each of the $O(m^2n)$ choices of vertices $x,c_1,b_1=d_1, b_2,
d_2$ with $x\in N_G(d_1)\cap N_G(d_2)$ and $b_1b_2\in E(G)$, apply
Lemma~\ref{lemma:lemma6.8}\eqref{item1:lemma6.8} on $X=(X_1\cup X_2
\cup Y_1)\setminus S_1$ in $O(n^3)$ time.

\item 
For each of the $O(m^2n)$ choices of vertices $x,c_1,b_1,b_2=d_1,d_2$
with $x\in N_G(d_1)\cap N_G(d_2)$ and $b_1b_2\in E(G)$, apply
Lemma~\ref{lemma:lemma6.8}\eqref{item1:lemma6.8} on $X=(X_1\cup X_2
\cup Y_2)\setminus S_2$ in $O(n^3)$ time.
\end{itemize}
Phase~2:
\begin{itemize}[leftmargin=*]
\item 
For each of the $O(m^2n^2)$ choices of vertices $x,d,d_1,d_2,b_1,b_2$
with $x\in N_G(d_1)\cap N_G(d_2)$ and $b_1b_2\in E(G)$, apply the
following procedure with $X=(X_1\cup X_2 \cup Z_0)\setminus S_0$ in
$O(n^2)$ time.

\item 
For each of the $O(m^2n^2)$ choices of vertices $x,d,c_1,b_1=d_1, b_2,
d_2$ with $x\in N_G(d_1)\cap N_G(d_2)$ and $b_1b_2\in E(G)$, apply the
following procedure on $X=(X_1\cup X_2 \cup Z_1)\setminus S_1$ in
$O(n^2)$ time.

\item 
For each of the $O(m^2n^2)$ choices of vertices
$x,d,c_1,b_1,b_2=d_1,d_2$ with $x\in N_G(d_1)\cap N_G(d_2)$ and
$b_1b_2\in E(G)$, apply the following procedure on $X=(X_1\cup X_2
\cup Z_2)\setminus S_2$ in $O(n^2)$ time.
\end{itemize}

Let $C_1$ (respectively, $C_2$) consist of the vertices $c$ such that
$cb_1b_2$ (respectively, $b_1b_2c$) is an induced path of $G$.  Let
$T^*_1$ be a tree that is the union of a shortest $dc$-path in
$G-(X\setminus \{c,d\})$ over all vertices $c\in C_1$.  Let each $T_i$
with $i\in\{1,2\}$ be a tree that is the union of a shortest
$dd_i$-path and a shortest $d_ic$-path in $G-(X\setminus \{c,d\})$
over all vertices $c\in C_i$.  Apply Lemma~\ref{lemma:lemma6.9} on
$d$, $b_1$, $b_2$, $T_1$ (respectively, $T^*_1$), and $T_2$ to
identify an odd hole of $G$ in $O(n^2)$ time.

The rest of the proof assumes that $C$ is a shortest odd hole of $G$
and shows that the above $O(m^2n^4)$-time algorithm outputs an odd
hole of $G$. Since $G$ does not contain any clean shortest odd hole,
$M_G(C)\ne\varnothing$.  For any $x\in M_G(C)$, a path $D$ of $C$ is
an {\em $x$-gap}~\cite{ChudnovskySSS19} if $G[D\cup \{x\}]$ is a hole
of $G$.  There is an $x\in M_G(C)$ with an $x$-gap or else each edge
of $C$ would be adjacent to all vertices of $M_G(C)$.  Let $x\in
M_G(C)$ maximize the length of a longest $x$-gap $D$.  Let $b_1b_2$ be
an edge of $C$ adjacent to each vertex of $M_G(C)\setminus N_G(x)$ as
ensured by Lemma~\ref{lemma:lemma6.7}, implying $M_G(C)\setminus
X_1\subseteq N_G(x)$.  Let $d_1$ and $d_2$ be the end-vertices of $D$.
By the maximality of $D$, each vertex of $M_G(C)\setminus X_2$ is
adjacent to the interior of $D$.  Thus, each vertex of
$M_G(C)\setminus (X_1\cup X_2)$ is adjacent to $x$ and the interior of
$D$.  Let $c_1$ and $c_2$ be the vertices such that $c_1b_1b_2c_2$ is
a path of $C$.  We have $k=|V(D)\cap \{b_1,b_2\}|\in\{0,1,2\}$.  If
$k=0$, then $S_k=\{d_1,d_2\}$ and the interior of $D$ is disjoint from
$c_1b_1b_2c_2$.  If $k=1$, then assume without loss of generality
$d_1=b_1$ and that $c_1$ is the neighbor of $d_1$ in $D$, implying
$S_k=\{c_1,b_1=d_1,d_2\}$.  If $k=2$, then assume without loss of
generality $d_1=b_2$, by $x\in N_G(b_1)\cup N_G(b_2)$ and that $b_1$
is the neighbor of $d_1$ in $D$, implying
$S_k=\{c_1,b_1,b_2=d_1,d_2\}$.

For each $k\in \{0,1,2\}$, $D$ is a path of $H_k$: We have $N_G(x)\cap
V(D)=\{d_1,d_2\}\subseteq S_k$.  By $X_1\cap V(D)=\{c_1,c_2\}\cap
V(D)\subseteq S_k$, we have $D\subseteq H_k$.  By $M_G(C)\cap
S_k=\varnothing$ and $M_G(C)\setminus X_1\subseteq N_G(x)$, we have
$M_G(C)\subseteq (X_1\cup N_G(x))\setminus S_k$, implying
$H_k\subseteq G-M_G(C)$.

Phase~1 handles the case $|E(D)|<0.5\cdot |E(C)|$: By
Lemma~\ref{lemma:lemma6.4}\eqref{item1:lemma6.4}, $D$ is a shortest
$d_1d_2$-path of $G-M_G(C)$, implying that $D$ is a shortest
$d_1d_2$-path of $H_k$.  Since no edge of $C$ is adjacent to all
$C$-major nodes of $G$, we have $|E(D)|\geq 3$ by the maximality of
$D$.  Thus, all internal vertices of $D$ are contained by $I_k$,
implying $M_G(C)\setminus (X_1\cup X_2)\subseteq Y_k$ by the
maximality of $D$.  Let $D^*$ be an arbitrary shortest $d_1d_2$-path
of $H_k$.  By $|E(D^*)|=|E(D)|$ and $H_k\subseteq G-M_G(C)$, $D^*$ is
a shortest $d_1d_2$-path of $G-M_G(C)$.  By
Lemma~\ref{lemma:lemma6.4}\eqref{item2:lemma6.4}, the graph $C^*$
obtained from $C$ by replacing $D$ with $D^*$ is a clean shortest odd
hole of $G-M_G(C)$.  Therefore, the interior of $D^*$ is disjoint from
and nonadjacent to $C-V(D)$, implying that $I_k$ is disjoint from and
nonadjacent to $C-V(D)$.  One can verify that $X=(X_1\cup X_2\cup
Y_k)\setminus S_k$ is either an approximate cleaner for $C$ with
$X\cap V(C)=\{c_1,c_2\}$ or $X\cap V(C)=\{c_2\}$.  Thus, Phase~1
outputs an odd hole of $G$.

Phase~2 handles the case $|E(D)|>0.5\cdot |E(C)|$: Let $d$ be a middle
vertex of $D$.  For each index $i\in\{1,2\}$, the $dd_i$-path $D_i$ of
$C$ is a shortest $dd_i$-path of $G-M_G(C)$ by
Lemma~\ref{lemma:lemma6.4}\eqref{item1:lemma6.4}, implying that $D_i$
is a shortest $dd_i$-path of $H_k$.  Thus, all internal vertices of
$D$ are contained by $J_k$, implying $M_G(C)\setminus (X_1\cup
X_2)\subseteq Z_k$.  Let each $D^*_i$ with $i\in\{1,2\}$ be an
arbitrary shortest $dd_i$-path of $H_k$.  By $|E(D^*_i)|=|E(D_i)|$ and
$H_k\subseteq G-M_G(C)$, $D^*_i$ is a shortest $dd_i$-path of
$G-M_G(C)$.  By Lemma~\ref{lemma:lemma6.4}\eqref{item2:lemma6.4}, the
graph $C^*$ obtained from $C$ by replacing $D$ with $D^*_1\cup D^*_2$
is a clean shortest odd hole of $G-M_G(C)$.  Therefore, the interior
of the $d_1d_2$-path $D^*_1\cup D^*_2$ is disjoint from and
nonadjacent to $C-V(D)$, implying that $J_k$ is disjoint from and
nonadjacent to $C-V(D)$.  One can verify that $X=(X_1\cup X_2\cup
Z_k)\setminus S_k$ is an approximate cleaner for $C$ with $X\cap
V(C)=\{c_1,c_2\}$ or $X\cap V(C)=\{c_2\}$.  We have $c_1\in C_1$ and
$c_2\in C_2$.
\begin{itemize}[leftmargin=*]
\item 
If $k=0$, then the $dc_1$-path $P_1$ of $T_1$ is the union of a
shortest $dd_1$-path $P'_1$ and a shortest $d_1c_1$-path $P''_1$ of
$G-(X\setminus \{c_1,d\})$ even if $c_1=d_1$.  By $M_G(C)\subseteq X$,
$X\cap V(C)\subseteq \{c_1,c_2\}$, and the fact that the shortest
$dd_1$-path and $d_1c_1$-path of $C$ are in $G-(X\setminus\{c_1,d\})$,
Lemma~\ref{lemma:lemma6.4}\eqref{item1:lemma6.4} implies that $P'_1$
(respectively, $P''_1$) is a shortest $dd_1$-path (respectively,
$d_1c_1$-path) of $G-M_G(C)$.

\item 
If $k\in \{1,2\}$, then $c_1$ is an internal vertex of $D$.  The
$dc_1$-path $P_1$ of $T^*_1$ is a shortest $dc_1$-path of
$G-(X\setminus \{c_1,d\})$.  By $M_G(C)\subseteq X$ and $X\cap
V(C)=\{c_2\}$, Lemma~\ref{lemma:lemma6.4}\eqref{item1:lemma6.4}
implies that $P_1$ is a shortest $dc_1$-path of $G-M_G(C)$.
\end{itemize}
The $dc_2$-path $P_2$ of $T_2$ is the union of a shortest $dd_2$-path
$P'_2$ and a shortest $d_2c_2$-path $P''_2$ of $G-(X\setminus \{c_2,d
\})$ even if $k=0$ and $c_2=d_2$.  By $M_G(C)\subseteq X$, $X\cap
V(C)\subseteq \{c_1,c_2\}$, and the fact that the shortest $dd_2$-path
and $d_2c_2$-path of $C$ are in $G-(X\setminus \{c_2,d\})$,
Lemma~\ref{lemma:lemma6.4}\eqref{item1:lemma6.4} implies that $P'_2$
(respectively, $P''_2$) is a shortest $dd_2$-path (respectively,
$d_2c_2$-path) of $G-M_G(C)$.  By applying
Lemma~\ref{lemma:lemma6.4}\eqref{item2:lemma6.4} at most four times on
$C$, $G[P_1\cup \{b_1,b_2\}\cup P_2]$ is a clean shortest odd hole of
$G-M_G(C)$.  Thus, Phase~2 outputs an odd hole of $G$.
\end{proof}

\subsection{Improved even-hole detection}
\label{subsection:subsection6.3}

Chang and Lu's algorithm consists of two $O(n^{11})$-time phases. The
first phase detects beetles in $O(n^{11})$ time, which is now reduced
to $O(n^7)$ time by Theorem~\ref{theorem:theorem1.5}.  The second
phase maintains a set $\TT$ of induced subgraphs of $G$ with the
property that if $G$ is even-hole-free, then so is each graph in $\TT$
until either $\TT$ becomes empty or an $H\in \TT$ is found to contain
even holes.  The initial $\TT$ consists of $O(n^5)$ graphs obtained
from guesses of (1) a $3$-path $P$ on a shortest even hole $C$ of $G$,
(2) an $X\subseteq V(G)$ that contains the major vertices of $C$
without intersecting $C$, and (3) a $Y\subseteq V(G)$ that contains
$N_G^{2,2}(C)$ (see~\S\ref{subsubsection:subsubsection6.3.2} for
definition) without intersecting $C$.  Each iteration of Phase~2 takes
$O(n^4)$ time to either ensure that an $H\in\TT$ is an extended clique
tree that contains even holes or replaces $H$ with $0$ (respectively,
$1$ and $2$) smaller graphs via ensuring that $H$ is an even-hole-free
extended clique tree (respectively, decomposing $H$ by a star-cutset
and decomposing $H$ by a $2$-join).  The guessed $P$ and $Y$ are
crucial in arguing that $H$ can be decomposed by a star-cutset without
increasing $|\TT|$, implying that each initial $H\in \TT$ incurs
$O(n)$ decompositions by star-cutsets.  Therefore, the overall time
for decompositions by star-cutsets is $O(n^{10})$, i.e., $O(n^5)$
times the initial $|\TT|$.  Each initial $H\in\TT$ incurs $O(n^2)$
decompositions by $2$-joins, implying that the overall time for
detecting even holes in extended clique trees and decompositions by
$2$-joins is $O(n^{11})$, i.e., $O(n^6)$ times the initial $|\TT|$.
We reduce the time of Phase~2 from $O(n^{11})$ to $O(n^9)$.  As in the
proof of Lemma~\ref{lemma:lemma6.10}, a factor of $n$ is removed by
reducing the initial $|\TT|$ from $O(n^5)$ to $O(n^4)$ via ignoring
$Y$ and the internal vertex of $P$.  Guessing only $X$ and the
end-vertices of $P$ does complicate the task of decomposing $H$ by a
star-cutset, but we manage to handle each decomposition by a
star-cutset in the same time bound (see the proof of
Lemma~\ref{lemma:lemma6.11}).  Another factor of $n$ is removed by
reducing the number of decompositions by $2$-joins incurred by each
initial $H\in \TT$ from $O(n^2)$ to $O(n)$ via carefully handling the
boundary cases (see the proof of Lemma~\ref{lemma:lemma6.12}).

Let $G$ be an $n$-vertex $m$-edge graph.  A {\em major
  vertex}~\cite{ChudnovskyKS05} of an even hole $C$ is a $v\in
V(G)\setminus V(C)$ with at least three distinct vertices in
$N_G(v)\cap V(C)$ that are pairwise nonadjacent in $G$.  Let $M_G(C)$
consist of the major vertices of an even hole $C$.  A hole without
major vertices is {\em clear}.  A {\em $v_1v_2$-hole} of $G$ is a
clear shortest even hole $C$ of $G$ such that $v_1$ and $v_2$ are the
end-vertices of a $3$-path of $C$.  A {\em tracer} of $G$ is a triple
$\langle H,v_1,v_2\rangle$ such that $v_1$ and $v_2$ are vertices of
an induced subgraph $H$ of $G$.  A tracer~$\langle H,v_1,v_2\rangle$
of $G$ is {\em lucky} if $H$ contains a $v_1v_2$-hole.  A set $\TT$ of
tracers of $G$ is {\em reliable} if $\TT$ satisfies the condition that
if $G$ contains even holes, then $\TT$ contains lucky tracers.

\begin{lem}
\label{lemma:lemma6.10}
If $G$ is beetle-free, then it takes $O(m^2n^2)$ time to either ensure
that $G$ contains even holes or obtain a reliable set of $O(mn^2)$
tracers of $G$.
\end{lem}

Subset $S$ of $V(H)$ is a {\em star-cutset}~\cite{Chvatal85} of a
graph $H$ if $S\subseteq N_H[s]$ holds for an $s\in S$ and the number
of connected components of $H-S$ is more than that of $H$.

\begin{lem}
\label{lemma:lemma6.11}
For any tracer $T$ of a beetle-free graph $G$, it takes~$O(mn^3)$ time to
complete one of the following tasks.  
Task~1: ensure that $G$ contains even holes.
Task~2: ensure that $T$ is not lucky.  
Task~3: obtain a star-cutset-free induced subgraph $H$ of $G$ such
that if $T$ is lucky, then $H$ contains even holes.
\end{lem}

The next lemma improves upon the $O(mn^4)$-time algorithm of
Chang and Lu.~\cite[Lemma~4.2]{ChangL15}.

\begin{lem}
\label{lemma:lemma6.12}
It takes $O(mn^3)$ time to detect even holes in an $n$-vertex $m$-edge
star-cutset-free graph.
\end{lem}

We first reduce Theorem~\ref{theorem:theorem1.6} via
Theorem~\ref{theorem:theorem1.5} to Lemmas~\ref{lemma:lemma6.10},
\ref{lemma:lemma6.11}, and \ref{lemma:lemma6.12}.

\begin{proof}[Proof of Theorem~\ref{theorem:theorem1.6}]
By Theorem~\ref{theorem:theorem1.5}, it takes $O(m^2n^3)$ time to
detect beetles in $G$.  If $G$ contains beetles, then $G$ contains
even holes.  Otherwise, we apply Lemma~\ref{lemma:lemma6.10} on the
beetle-free $G$ in $O(m^2n^2)$ time.  If $G$ is ensured to contain
even holes, then the theorem is proved.  Otherwise, we have a reliable
set $\TT$ of $O(mn^2)$ tracers of $G$.  It takes overall $O(m^2n^5)$
time to apply Lemma~\ref{lemma:lemma6.11} on all $T\in\TT$.  If Task~1
is completed for any $T\in \TT$, then $G$ contains even holes.  If
Task~2 is completed for all $T\in \TT$, then $G$ is even-hole-free.
Otherwise, we apply Lemma~\ref{lemma:lemma6.12} in overall $O(m^2n^5)$
time on each of the $O(mn^2)$ star-cutset-free induced subgraphs $H$
of $G$ corresponding to the tracers $T\in \TT$ for which Task~3 is
completed.  If an $H$ contains even holes, then so does~$G$.
Otherwise, $G$ is even-hole-free.
\end{proof}

Lemmas~\ref{lemma:lemma6.10}, \ref{lemma:lemma6.11}, and
\ref{lemma:lemma6.12} are proved in
\S\ref{subsubsection:subsubsection6.3.1},
\S\ref{subsubsection:subsubsection6.3.2}, and
\S\ref{subsubsection:subsubsection6.3.3}, respectively.

\subsubsection{Proving Lemma~\ref{lemma:lemma6.10}}
\label{subsubsection:subsubsection6.3.1}

\begin{lem}[{da~Silva and Vu\v{s}kovi\'{c}~{\cite{daSilvaV07}}}]
\label{lemma:lemma6.13}
Let $G$ be an $n$-vertex $m$-edge graph.  It takes $O(mn^2)$ time to
either ensure that $G$ contains even holes or obtain all $O(m)$
maximal cliques of $G$.
\end{lem}

\begin{lem}[{Chang and Lu~\cite[Lemma~3.4]{ChangL15}}]
\label{lemma:lemma6.14}
If $C$ is a shortest even hole of a $4$-hole-free graph $G$, then
either $M_G(C)\subseteq N_G(v)$ holds for a vertex $v$ of $C$ or
$G[M_G(C)]$ is a clique.
\end{lem}

\begin{proof}[Proof of Lemma~\ref{lemma:lemma6.10}]
It takes $O(m^2)$ time to detect $4$-holes in $G$, so we assume that
$G$ is $4$-hole-free.  By Lemma~\ref{lemma:lemma6.13}, it suffices to
consider that the set $\KK$ of $O(m)$ maximal cliques of $G$ is
available.  It takes $O(m^2n^2)$ time to obtain the set
$\mathbbmsl{T}$ of $O(mn^2)$ tracers of $G$ in the form of (1)
$\langle G-(N_G(v)\setminus \{v_1,v_2\}), v_1,v_2\rangle$ with
$\{v_1,v,v_2\}\subseteq V(G)$ or (2) $\langle G-V(K),v_1,v_2\rangle$
with $K\in \KK$ and $\{v_1,v_2\}\subseteq V(G)$.  To see that $\TT$ is
reliable, let $C$ be a shortest even hole of $G$.  Case~1:~$M_G(C)
\subseteq N_G(v)$ holds for a vertex $v$ of $C$.  Let $v_1$ and $v_2$
be the neighbors of $v$ in~$C$. By $M_G(C) \subseteq N_G(v) \setminus
\{v_1,v_2\}$ and $(N_G(v)\setminus\{v_1,v_2\})\cap C=\varnothing$, $C$
is a $v_1v_2$-hole of $G-(N_G(v)\setminus\{v_1,v_2\})$.
Case~2:~$M_G(C)\not\subseteq N_G(v)$ holds for all vertices $v$ of
$C$.  By Lemma~\ref{lemma:lemma6.14}, $G[M_G(C)]$ is a clique.  Let
$K$ be a maximal clique with $M_G(C)\subseteq V(K)$.  We have
$V(K)\cap C= \varnothing$ or else $M_G(C)\cap C=\varnothing$ would
imply $M_G(C) \subseteq V(K) \setminus \{v\} \subseteq N_G(v)$ for any
$v\in V(K) \cap C$, contradiction.  Thus, $C$ is a $v_1v_2$-hole of
$G-V(K)$ for any $v_1v_2$-path of $C$ with $3$ vertices.
\end{proof}

\subsubsection{Proving Lemma~\ref{lemma:lemma6.11}}
\label{subsubsection:subsubsection6.3.2}

Vertex $x$ {\em dominates} vertex~$y$ in graph $H$ if $x \neq y$
and~$N_H[y] \subseteq N_H[x]$.  Vertex $y$ is {\em dominated} in $H$
if some vertex of $H$ dominates $y$ in $H$.  A star-cutset $S$ of
graph $H$ is {\em full} if $S=N_H[s]$ holds for some vertex $s$ of
$S$.

\begin{lem}[{Chv\'{a}tal~\cite[Theorem~1]{Chvatal85}}]
\label{lemma:lemma6.15}
A graph without dominated vertices and full star-cutsets is
star-cutset-free.
\end{lem}

\begin{lem}[{Chudnovsky, Kawarabayashi, and Seymour~{\cite[Lemma 2.2]{ChudnovskyKS05}}}]
\label{lemma:lemma6.16}
If $x$ is a major vertex of a shortest even hole $C$ of graph $G$,
then $|N_G(x, C)|$ is even.
\end{lem}

Let $N_G^{i}(C)$ consist of the vertices $x\in N_G(C)\setminus M_G(C)$
such that $|N_G(x,C)|=i$ and $C[N_G(x,C)]$ is connected.  Let
$N_G^{i,i}(C)$ consist of the vertices $x\in N_G(C)\setminus M_G(C)$
such that $C[N_G(x,C)]$ has two connected components, each of which
has $i$ vertices.

\begin{lem}[{Chang and Lu~\cite[Lemma~2.2]{ChangL15}}]
\label{lemma:lemma6.17}
For any clear shortest even hole $C$ of a beetle-free graph~$G$, we
have
\begin{displaymath}
N_G(C) \subseteq N_G^1(C)\cup N_G^2(C)\cup N_G^3(C) \cup
N_G^{1,1}(C)\cup N_G^{2,2}(C).
\end{displaymath}
\end{lem}

\begin{proof}[Proof of Lemma~\ref{lemma:lemma6.11}]
We first prove the lemma using the following two claims for any tracer
$T=\langle H,v_1,v_2\rangle$ of an $n$-vertex $m$-edge beetle-free
connected graph $G$:
\begin{enumerate}[label={\em Claim~\arabic*:}, ref={Claim~\arabic*}, leftmargin=*]
\item 
\label{claim:star-cutset1}
It takes~$O(mn^2)$ time to obtain a tracer $T'=\langle
H',v'_1,v'_2\rangle$ of $G$, where $H'$ is an induced subgraph of $H$
having no dominated vertices, such that if $T$ is lucky, then so is
$T'$.

\item
\label{claim:star-cutset2}
It takes $O(mn^2)$ time to (1) ensure that $H$ is
full-star-cutset-free, (2) obtain an even hole of $G$, or (3) obtain a
proper induced subgraph $H'$ of $H$ such that if $T$ is lucky, then so
is $\langle H',v_1,v_2\rangle$.
\end{enumerate}
The algorithm proceeds in $O(n)$ iterations to update $T=\langle
H,v_1,v_2\rangle$.  Each iteration starts with
applying~\ref{claim:star-cutset1} to update $T$ without destroying its
luckiness by replacing $\langle H,v_1,v_2\rangle$ with the ensured
$\langle H',v'_1,v'_2\rangle$ such that $H'$ is an induced subgraph of
$H$ that does not contain any dominated vertex.  It then
applies~\ref{claim:star-cutset2} on the resulting $T=\langle
H,v_1,v_2\rangle$.  If $H$ is ensured to be full-star-cutset-free,
then Task~3 is completed by Lemma~\ref{lemma:lemma6.15}.  If we obtain
an even hole of $G$, then Task~2 is completed.  Otherwise, it updates
$T$ without destroying its luckiness by replacing $H$ with the
obtained proper induced subgraph $H'$ of $H$ and proceed to the next
iteration.  The overall running time is $O(mn^3)$.

To prove~\ref{claim:star-cutset1}, the $O(mn^2)$-time algorithm
outputs the resulting $T$ after iteratively updating the initial
$T=\langle H,v_1,v_2\rangle$ by the following procedure until $H$
contains no dominated vertices: (1) spend $O(mn)$ time to detect
vertices $x$ and $y$ of $H$ such that $x$ dominates $y$ in $H$, (2)
let $H=H-\{y\}$, and (3) if $y=v_i$ with $i\in\{1,2\}$, then let
$v_i=x$.  The resulting $H$ is an induced subgraph of the initial $H$.
For the correctness, it suffices to prove that if a tracker $T$ is
lucky, then so is the resulting $T$ after an iteration of the loop.
Suppose that a $v_1v_2$-hole $C$ of $H$ contains $y$ or else $C$
remains a $v_1v_2$-hole of $H'=H-\{y\}$.  Since $C$ is an even hole,
we have $x\notin V(C)$ and $|N_C[y]| = 3$, implying a connected
component of~$C[N_G(x,C)]$ with at least $3$ vertices.  By
Lemma~\ref{lemma:lemma6.17}, we have $x\in N_H^{3}(C)$, implying that
$N_G(x,C)$ consists of $y$ and the two neighbors of $y$ in $C$.  Thus,
$C'=H[C\cup \{x\} \setminus \{y\}]$ remains a shortest even hole of
$H'$.  Let $v_0$ be a vertex of $C$ such that $v_1v_0v_2$ is a
$3$-path of $C$.  For each $i\in\{0,1,2\}$, if $y=v_i$, then let
$u_i=x$; otherwise, let $u_i=v_i$.  Clearly, $u_1u_0u_2$ is a $3$-path
of $C'$.  It remains to show that $C'$ is clear.  Assume for
contradiction $z\in M_{H'}(C')$, implying $y\ne z$ and $z\in
M_{H}(C')$.  By Lemma~\ref{lemma:lemma6.16}, $|N_{C'}(z)| \ge 4$ and
$|N_{C'}(z)|\ne 5$.  By Lemma~\ref{lemma:lemma6.17},
$M_H(C)=\varnothing$ implies $|N_{C}(z)| \leq 4$.  By $C - \{y\} = C'
- \{x\}$, exactly one of $x$ and $y$ is adjacent to $z$ in $H$ or else
$z\in M_H(C')$ would imply $z\in M_H(C)$.  Thus, $z\in N_H(x)\setminus
N_H(y)$, implying $|N_{C}(z)|=|N_{C'}(z)|-1=3$.
Lemma~\ref{lemma:lemma6.17} implies $z\in N_H^{3}(C)$.  Since
$C[N_G(z,C)]$ is a $3$-path, $H[C'\cup \{z\}]$ is a beetle $B$ of $H$
in which $B[N_B[z]\setminus\{x\}]$ is a diamond, contradiction.

To prove~\ref{claim:star-cutset2}, it takes $O(mn)$ time to detect
full star-cutsets in $H$.  It suffices to focus on the case that $H$
contains a full star-cutset $S=N_{H}[s]$.  Let $\BB$ consist of the
connected components of $H-S$.  It takes $O(n^3)$ time to obtain, for
every two nonadjacent vertices $s_1$ and $s_2$ of $S$, the list
$L(s_1,s_2)$ of elements in $\BB$ that are adjacent to both $s_1$ and
$s_2$.  It takes $O(m^2)$ time to check whether the following
conditions hold:
\begin{enumerate}[leftmargin=*]
\item 
\label{chunk1}
There are distinct $B_i\in L(s_1,s_2)$ with $\{s_1,s_2\}\subseteq S$
for $i\in\{1,2\}$.

\item 
\label{chunk2}
There are disjoint edges $s_is_{i+2}$ of $H[S]$ with distinct $B_i\in
L(s_{2i-1},s_{2i})$ for $i\in\{1,2\}$.
\end{enumerate}
If Condition~\ref{chunk1} holds, then $H[P_1\cup P_2\cup s]$ (is a
theta and thus) contains even holes for any shortest $s_1s_2$-path
$P_i$ in $H[B_i\cup \{s_1,s_2\}]$.  If Condition~\ref{chunk2} holds,
then $H[P_1\cup P_2\cup s]$ contains even holes for any shortest
$s_{2i-1}s_{2i}$-path $P_i$ in $H[B_i\cup \{s_{2i-1},s_{2i}\}]$.  The
rest of the proof assumes that neither condition holds.  If there were
a $v_1v_2$-hole $C$ of $H$ intersecting distinct $B_1$ and $B_2$ of
$\BB$, then $s\notin C$, implying that $C[N_G(s,C)]$ is not
connected. By Lemma~\ref{lemma:lemma6.17}, either $s\in N_H^{1,1}(C)$,
implying Condition~\ref{chunk1}, or $s\in N_H^{2,2}(C)$, implying
Condition~\ref{chunk2}.  Hence, each $v_1v_2$-hole $C$ of $H$
intersects at most one element of $\BB$.  If a $B\in \BB$ contains one
or both of $v_1$ and $v_2$, then the claim is proved with $H'=H[B\cup
  S]$.  It remains to consider the case $\{v_1,v_2\}\subseteq S$.  Let
$C$ be a $v_1v_2$-hole intersecting exactly one $B\in \BB$.  If $s\in
C$, then $V(C)\cap S=\{v_1,s,v_2\}$, implying $B\in L(v_1,v_2)$.  If
$s\notin C$, then $s\in N_H^{3}(C)\cup N_H^{1,1}(C)\cup N_H^{2,2}(C)$
by Lemma~\ref{lemma:lemma6.17}, also implying $B\in L(v_1,v_2)$.
Since Condition~\ref{chunk1} does not hold, $|L(v_1,v_2)|\leq 1$.
Therefore, if $|L(v_1,v_2)|=1$, then the claim is proved with
$H'=H[B\cup S]$, where $B$ is the only element in $L(v_1,v_2)$.  If
$|L(v_1,v_2)|=0$, then the claim is proved with $H'=H[S]$.
\end{proof} 

\subsubsection{Proving Lemma~\ref{lemma:lemma6.12}}
\label{subsubsection:subsubsection6.3.3}

\begin{figure}[t]
\centering
\centerline{\scalebox{0.38}{{\includegraphics{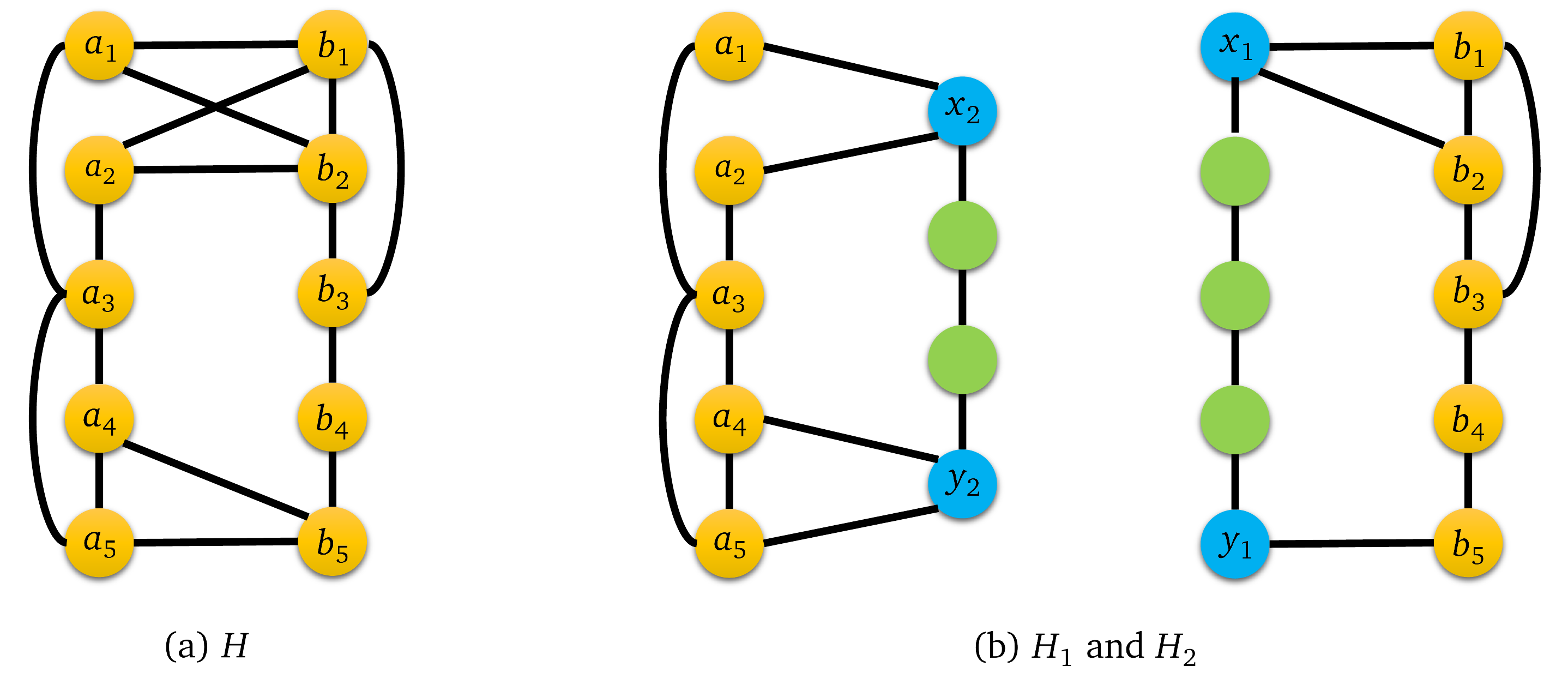}}}}
\caption{A $2$-join $J=(V_1,V_2,X_1,X_2,Y_1,Y_2)$ of $H$ with
  $V_1=\{a_1,\ldots,a_6\}$, $V_2=\{b_1,\ldots,b_6\}$,
  $X_1=\{a_1,a_2\}$, $X_2=\{b_1,b_2\}$ $Y_1=\{a_5,a_6\}$, and
  $Y_2=\{b_6\}$ and and the parity-preserving blocks of decomposition
  $H_1$ and $H_2$ for $J$.}
\label{figure:figure12}
\end{figure}

$(V_1,V_2,X_1,X_2,Y_1,Y_2)$ is a {\em
  $2$-join}~\cite[\S1.3]{daSilvaV13} (which is called a non-path
$2$-join in, e.g.,~\cite{ChangL15,Trotignon08,TrotignonV12}) of a
connected graph $H$ if
\begin{enumerate}[leftmargin=*]
\item 
$V_1$ and $V_2$ form a disjoint partition of $V(H)$ with $|V_1|\geq 3$
  and $|V_2|\geq 3$,

\item 
$X_i$ and $Y_i$ are disjoint nonempty subsets of $V_i$ for each $i$,

\item 
$H[V_i]$ is not a minimal $X_iY_i$-path for each $i$, and 

\item 
if $v_i\in V_i$ for each $i$, then $v_1v_2\in E(H)$ if and only if
$v_i\in X_i$ for each $i$ or $v_i\in Y_i$ for each $i$.
\end{enumerate}
See Figure~\ref{figure:figure12}(a) for an example.

\begin{lem}[Trotignon and Vu\v{s}kovi\'{c}~{\cite[Lemma~3.2]{TrotignonV12}}]
\label{lemma:lemma6.18}
If $(V_1,V_2,X_1,Y_1,X_2,Y_2)$ is a $2$-join of a star-cutset-free
connected graph $H$, then the following statements hold for each
$i\in\{1,2\}$:
\begin{enumerate}[leftmargin=*]
\item 
\label{item1:lemma6.18}
Each connected component of $H[V_i]$ intersects both $X_i$ and $Y_i$.

\item 
\label{item2:lemma6.18}
Each vertex of $X_i$ (respectively, $Y_i$) has a non-neighbor of $H$
in $Y_i$ (respectively, $X_i$).
\end{enumerate}
\end{lem}

\begin{lem}[Charbit, Habib, Trotignon, and Vu\v{s}kovi\'{c}~{\cite[Theorem~4.1]{CharbitHTV12}}]
\label{lemma:lemma6.19}
Given an $n$-vertex $m$-edge connected graph $H$, it takes $O(mn^2)$
time to either obtain a $2$-join of $H$ or ensure that $H$ is
$2$-join-free.
\end{lem}

\begin{lem}[da~Silva and~Vu\v{s}kovi\'{c}~{\cite[Corollary~1.3]{daSilvaV13}}]
\label{lemma:lemma6.20}
A connected even-hole-free star-cutset-free $2$-join-free graph is an
extended clique tree.
\end{lem}

Let $J=(V_1,V_2,X_1,X_2,Y_1,Y_2)$ be a $2$-join of a star-cutset-free
connected graph $H$.  Let $P_i$ with $i\in\{1,2\}$ be a shortest
induced $X_iY_i$-path $P_i$ of $H[V_i]$ as ensured by
Lemma~\ref{lemma:lemma6.18}\eqref{item1:lemma6.18}.  If $|V(P_i)|$ is
even (respectively, odd), then let $p_i=4$ (respectively, $p_i=5$).
The {\em parity-preserving blocks of
  decomposition}~\cite{TrotignonV12} for $J$ are the graphs $H_i$ with
$i\in\{1,2\}$ consisting of $H[V_i]$, a $p_j$-vertex $x_jy_j$-path
with $j=3-i$, edges $xx_j$ for all vertices $x$ of $X_i$, and edges
$yy_j$ for all vertices $y$ of $Y_i$.  See
Figure~\ref{figure:figure12}(b) for an example.

\begin{lem}[{Trotignon and Vu\v{s}kovi\'{c}~\cite[Lemma~3.8]{TrotignonV12}}]
\label{lemma:lemma6.21}
Let $H_1$ and $H_2$ be the parity-preserving blocks of decomposition
for a $2$-join of an $m$-edge star-cutset-free connected graph $H$.
\begin{enumerate}[leftmargin=*]
\item 
Both $H_1$ and $H_2$ are star-cutset-free.

\item 
Both $H_1$ and $H_2$ are even-hole-free if and only if $H$ is
even-hole-free.
\end{enumerate}
\end{lem}

\begin{lem}[{Chang and Lu~\cite[Lemma~4.12]{ChangL15}}]
\label{lemma:lemma6.22}
Each of the parity-preserving blocks of decomposition for a $2$-join
for an $n$-vertex $m$-edge star-cutset-free connected graph has at
most $n$ vertices and $m$ edges.
\end{lem}

Graph $H$ is an {\em extended clique tree}~\cite{daSilvaV13} if there
is a set $S$ of two or fewer vertices of $H$ such that each
biconnected component of $H-S$ is a clique.  It takes $O(mn^2)$ time
to determine whether an $n$-vertex $m$-edge graph is an extended
clique tree.

\begin{lem}[{Chang and Lu~\cite[Lemma~4.6]{ChangL15}}]
\label{lemma:lemma6.23}
It takes $O(n^4)$ time to detect even holes in an $n$-vertex connected
extended clique tree.
\end{lem}

\begin{proof}[Proof of Lemma~\ref{lemma:lemma6.12}]
Let $W(H)$ consist of the $v\in V(H)$ with $|N_H(v)|\geq 3$.  Let
$h(H)=|V(H)|+|W(H)|$.  We first prove the claim that if $H_1$ and
$H_2$ are the parity-preserving blocks of decomposition for a $2$-join
$(V_1,V_2,X_1,X_2,Y_1,Y_2)$ of a star-cutset-free connected graph $H$,
then (a) $X_i\cup V_j$, $Y_i\cup V_j$, or $X_i\cup Y_i\cup V_j$ with
$\{i,j\}=\{1,2\}$ induces a $6$-hole of $H$ or (b) we have
\begin{eqnarray}
h(H_1)+h(H_2)&\leq&h(H)+14\label{eq:eq1}\\
\max\{h(H_1),h(H_2)\}&\leq& h(H)-1.\label{eq:eq2}
\end{eqnarray}
By definition of $H_i$ and $H_j$ with $\{i,j\}=\{1,2\}$, (i) if $v\in
V_i$, then $|N_{H_i}(v)|\leq |N_{H}(v)|$, (ii) if $x_j\in W(H_i)$,
then $X_j\subseteq W(H)$, and (iii) if $y_j\in W(H_i)$, then
$Y_j\subseteq W(H)$.  Thus, $|W(H_i)|\leq |W(H)|$.  By
Lemma~\ref{lemma:lemma6.22}, $h(H_i)\leq h(H)$.  By
$|V(H_i)|=|V_i|+p_j\leq |V_i|+5$ and $W(H_i)\setminus W(H)\subseteq
\{x_j,y_j\}$, Equation~\eqref{eq:eq1} holds.  To see
Equation~\eqref{eq:eq2}, assume $h(H_i)=h(H)$, implying
\begin{eqnarray}
|V(H_i)|&=&|V(H)|\label{eq:eq3}\\
|W(H_i)|&=&|W(H)|\label{eq:eq4}.
\end{eqnarray}
By $|V(H_i)|=|V(H)|-|V_j|+p_j$ and Equation~\eqref{eq:eq3},
$|V_j|=p_j$.  If $|V(P_j)|\in\{4,5\}$, then $|V_j|=p_j=|V(P_j)|$
contradicts $H[V_j]\ne P_j$.  By $p_j\in\{4,5\}$, we have
$|V(P_j)|\in\{2,3\}$.

Case~1: $|V(P_j)|=2$.  $|V_j|=p_j=4$.  By
Lemma~\ref{lemma:lemma6.18}\eqref{item2:lemma6.18}, $|X_j|=|Y_j|=2$.
Thus, $|X_i|=|Y_i|=1$ or else $X_j\subseteq W(H)$ or $Y_j\subseteq
W(H)$, contradicting Equation~\eqref{eq:eq4}.  Hence,
$|N_{H_i}(x_j)|=|N_{H_i}(y_j)|=2$. By Equation~\eqref{eq:eq4},
$X_j\cap W(H)=Y_j\cap W(H)=\varnothing$. By
Lemma~\ref{lemma:lemma6.18}\eqref{item1:lemma6.18}, $H[X_i\cup Y_i
  \cup V_j]$ is a $6$-hole.

Case~2: $|V(P_j)|=3$.  $|V_j|=p_j=5$.  Let $Z=V_j\setminus V(P_j)$.
Thus, $Z\cap (X_j\cup Y_j)\ne\varnothing$ or else $V(P_j)$ is a
star-cutset of $H$.  Let $z\in Z \cap X_j$ without loss of generality.
$|X_i|=1$ or else $X_j\subseteq W(H)$ with $|X_j|\geq 2$ contradicts
Equation~\eqref{eq:eq4}.  Hence, $|N_{H_i}(x_j)|=2$, implying $X_j\cap
W(H)=\varnothing$ by Equation~\eqref{eq:eq4}.  By
Lemma~\ref{lemma:lemma6.18}\eqref{item1:lemma6.18},
$|N_{H}(z)|=2$. Let $z'$ be the neighbor of $z$ in $V_j$.  We know
$z'\notin Y_j$ or else $zz'$ is shorter than $P_j$.  By
Equation~\eqref{eq:eq4}, the internal vertex of $P_j$ has degree $2$
in $H$.  Thus, $Z=\{z,z'\}$ and $z'y_j\in E(H)$ by
Lemma~\ref{lemma:lemma6.18}\eqref{item1:lemma6.18}.  $H[X_i\cup V_j]$
is a $6$-hole.

It suffices to prove the lemma for any given $n$-vertex $m$-edge
star-cutset-free connected graph $H_0$.  Let $\HH$ initially consist
of $H_0$.  Repeat the following loop until $\HH=\varnothing$ or the
current $H$ is ensured to contain an even hole: Each iteration starts
with getting a current $H\in\HH$ and deleting $H$ from $\HH$.  If
$w(H)\leq 15$, then detect even holes in $H$ in $O(1)$ time. If $H$ is
even-hole-free, then proceed to the next iteration; otherwise, exit
the loop.  If $w(H)\geq 16$, then apply Lemma~\ref{lemma:lemma6.19} on
$H$ in $O(mn^2)$ time.
\begin{itemize}[leftmargin=*]
\item 
Case~1: $H$ is $2$-join-free.  Determine whether $H$ is an extended
clique tree in $O(mn^2)$ time.  If $H$ is an extended clique tree,
then apply Lemma~\ref{lemma:lemma6.23} to detect even holes in $H$ in
$O(n^4)$ time; otherwise, $H$ contains an even hole by
Lemma~\ref{lemma:lemma6.20}.  If $H$ contains an even hole, then exit
the loop; otherwise, proceed to the next iteration.

\item 
Case~2: $H$ admits a $2$-join $J=(V_1,V_2,X_1,X_2,Y_1,Y_2)$ of $H$.
Spend $O(1)$ time to detect $6$-holes in $H$ from $H[X_i\cup V_j]$,
$H[Y_i\cup V_j]$, or $H[X_i\cup Y_i\cup V_j]$ with $\{i,j\}=\{1,2\}$.
If $H$ contains a $6$-hole, then exit the loop.  Otherwise, add to
$\HH$ the $O(m)$-time obtainable parity-preserving blocks of
decomposition for $J$, each of which has at most $n$ vertices and $m$
edges according to Lemma~\ref{lemma:lemma6.22}, and proceed to the
next iteration.
\end{itemize}
By Lemma~\ref{lemma:lemma6.21}, if the loop stops with an empty $\HH$,
then $H_0$ is even-hole-free; otherwise, $H_0$ contains an even hole.
We bound the number of iterations by $O(n)$ as follows.  Let Case~2
occur $f(h)$ times with $h=h(H_0)$.  By Equations~\eqref{eq:eq1}
and~\eqref{eq:eq2}, if $h\leq 15$, then $f(h)=0$; otherwise,
$$f(h)\leq\max\{1+f(h_1)+f(h_2):h_1,h_2\leq h-1, h_1+h_2\leq h+14\}.$$
By induction on $h$, we prove $f(h)\leq\max(h-15,0)$, which 
holds for $h\leq 15$.  
For~$h\geq 16$, 
\begin{eqnarray*}
f(h)
&\leq&
\max\{1+\max(h_1-15,0)+\max(h_2-15,0): h_1,h_2\leq h-1, h_1+h_2\leq h+14\}\\
&\leq&\max\{\max(h_1+h_2-29,h_1-14,h_2-14,1): h_1,h_2\leq h-1, h_1+h_2\leq h+14\}\\
&\leq&\max(h-15,h-15,h-15,1)\\
&=&\max(h-15,0).
\end{eqnarray*}
Since the number of iterations is $O(h)=O(n)$, the overall running
time is $O(mn^3)$ except for that of applying
Lemma~\ref{lemma:lemma6.23}.  Since each iteration increases the
overall number of vertices of graphs in $\HH$ by $O(1)$, the overall
number of vertices of the graphs in $\HH$ remains $O(n)$ throughout.
Thus, all $O(n)$ iterations of applying Lemma~\ref{lemma:lemma6.23}
take overall $O(n^4)=O(mn^3)$ time.
\end{proof}

\section{Concluding remarks}
\label{section:section7}
We solve the three-in-a-tree problem on an $n$-vertex $m$-edge
undirected graph in $O(m\log^2 n)$ time, leading to improved
algorithms for recognizing perfect graphs and detecting thetas,
pyramids, beetles, and odd and even holes.  It would be interesting to
see if the complexity of the three-in-a-tree problem can be further
reduced.  The amortized cost of maintaining the connectivity
information for the dynamic graph $G-X$ can be improved to $O(\log^2
n/\log\log n)$ using~\cite{Wulff-Nilsen13} or even to $O(\log n
\log\log^{O(1)}n)$ using~\cite{Thorup00}.  Since $G-X$ is purely
decremental, we can use the randomized algorithm in~\cite{Thorup99}
for further speedup.  However, this is not our only $O(\log^2 n)$
bottleneck: At the moment we pay $O(\log n)$ time for each neighbor of
a vertex in $X$ when it changes color, so if it changes color $O(\log
n)$ times, then it will be hard to beat the $O(\log^2 n)$ factor.

\section*{Acknowledgments}
We thank the anonymous reviewers of STOC 2020 and Evangelos Kipouridis
for helpful comments. We thank Ho-Lin Chen and Meng-Tsung Tsai for
commenting on a preliminary version~\cite{Lai2018} of the paper.

\bibliographystyle{abbrv}
\bibliography{tiat}
\end{CJK*}
\end{document}